\documentclass[11pt]{article}
\usepackage{jheppub}
\usepackage{subcaption}
\usepackage{amsmath, amssymb}
\usepackage{mathtools}

\usepackage[dvipsnames]{xcolor}
\usepackage{caption}
\usepackage{hyperref}
\usepackage{url}
\usepackage{graphicx}
\usepackage{fancyhdr}
\usepackage{rotating}
\usepackage{tabularx}
\usepackage{amsthm}
\usepackage{tabu}
\usepackage{multirow}
\usepackage{longtable}

\newcommand{\qbps}{$\frac14$-BPS }
\newcommand{\lw}{\ell_\gamma}
\newcommand{\nw}{n_\gamma}
\newcommand{\mw}{m_\gamma}
\newcommand{\hbps}{$\frac12$-BPS }
\newcommand{\tsz}{(\tau,z,\sigma)}

\newcommand{\slz}{$ \displaystyle PSL(2,\mathbb{Z}) $ }
\newcommand{\wmat}[4]{\begin{pmatrix}
    #1 & #2 \\ #3 & #4
  \end{pmatrix}
}

\newcommand{\cvec}[2]{\begin{pmatrix}
#1 \\ #2
  \end{pmatrix}
}

\newcommand{\sgn}{\text{sgn}}
\newcommand{\abs}[1]{\left \vert #1 \right \vert}

\newcommand{\Z}{\mathbb{Z}}

\newtheorem{definition}{Definition}[section]

\newtheorem{proposition}{Proposition}[section]

\newcommand{\mat}[4]{\begin{pmatrix} #1 & #2 \\ #3 & #4 \end{pmatrix}}

\renewcommand{\t}{\tilde}
\renewcommand{\d}{\cdot}
\newcommand{\be}{\begin{equation}}
\newcommand{\ee}{\end{equation}}
\newcommand{\ba}{\begin{eqnarray}}
\newcommand{\ea}{\end{eqnarray}}
\newcommand{\lp}{\left(}
\newcommand{\rp}{\right)}

\newcommand{\ndt}{\noindent}

\newcommand{\mpage}[1]{\begin{minipage}{2.1in}
		#1
\end{minipage}}
\newcommand{\CN}{\mathcal{N}}
\newcommand{\IZ}{\mathbb{Z}}

\renewcommand{\=}{\;=\;}
\newcommand{\defeq}{\; \coloneqq \;} 
\renewcommand{\i}{\mathrm{i}}
\newcommand{\bea}{\begin{eqnarray}}
\newcommand{\eea}{\end{eqnarray}}

\newcommand{\gset}{\text{W}(n,\ell,m)}
\newcommand{\gbsm}{\Gamma_\text{BSM}(n,\ell,m)}

\newcommand{\hormat}[4]{ #2/#4  \, \rightarrow  \,#1/ #3 }

\newcolumntype{C}{>{$}c<{$}}

\title{Dyonic black hole degeneracies in~$\mathcal{N} = 4$ string theory from Dabholkar-Harvey degeneracies}
\author[a,b]{Abhishek Chowdhury }
\author[a,c]{, Abhiram Kidambi }
\author[d]{, Sameer Murthy }
\author[e]{, Valentin Reys }
\author[a]{, Timm Wrase }

\affiliation[a]{Institute of Theoretical Physics, TU Wien, Wiedner Hauptstra\ss e 8-10, A-1040 Vienna, Austria}
\affiliation[b]{IIT Bhubaneshwar, SBS Building, Argul, Khordha, 752050, Odisha, India}
\affiliation[c]{Stanford Insititute for Theoretical Physics, Stanford University, Palo Alto, California, USA}
\affiliation[d]{King's College London, Strand, London WC2R 2LS, United Kingdom}
\affiliation[e]{Institute of Theoretical Physics, KU Leuven, Celestijnenlaan 200D, B-3001 Leuven, Belgium}

\abstract{
	The degeneracies of single-centered dyonic \qbps black holes (BH) in 
	Type II string theory on~K3$\times T^2$ are known to be coefficients of certain 
	mock Jacobi forms arising from the Igusa cusp form~$\Phi_{10}$.
	In this paper we present an exact analytic formula for these BH degeneracies
	purely in terms of the degeneracies of the perturbative~$\frac12$-BPS states of the theory. 
	We use the fact that the degeneracies are completely controlled by the polar coefficients 
	of the mock Jacobi forms, using the Hardy-Ramanujan-Rademacher circle method.
	Here we present a simple formula for these polar coefficients as a 
	quadratic function of the $\frac12$-BPS degeneracies.
	We arrive at the formula by using the physical interpretation of polar coefficients 
	as negative discriminant states, and then making use of previous results in the literature to 
	track the decay of such states into pairs of \hbps states in the moduli space.
	Although there are an infinite number of such decays, 
	we show that only a finite number of them contribute to the formula. 
	The phenomenon of BH bound state metamorphosis (BSM) plays a crucial role in our analysis. 
	We show that the dyonic BSM orbits with $U$-duality invariant~$\Delta<0$
	are in exact correspondence with the solution sets of the Brahmagupta-Pell 
	equation, which implies that they are isomorphic to the group of units in the 
	order~$\mathbb{Z}[\sqrt{|\Delta|}]$ in the 
	real quadratic field~$\mathbb{Q}(\sqrt{|\Delta|})$.
	We check our formula against the known numerical data arising from the Igusa cusp form,
	for the first~1650 polar coefficients, and find perfect agreement.}

\begin{document}
	\maketitle
	
	\section{Introduction and summary of results}
	\label{sec:introduction}
	
	String theory has proven to be a powerful description of the microscopic degeneracy of states of 
	supersymmetric black holes. The original breakthrough of \cite{Sen:1995in,Strominger:1996sh} gave us two complementary
	pictures of viewing the black hole---as a bound state of microscopic excitations of strings and branes
	carrying a statistical entropy,
	and as a solution to the equations of motion of macroscopic supergravity having a thermodynamic
	Bekenstein-Hawking-Wald entropy associated to the black hole horizon. 
	A lot of progress has occurred since then on sharpening both these pictures: on the macroscopic side 
	we have learned how to include stringy effects as well as quantum gravitational effects in the calculation 
	of the exact quantum gravitational entropy; and 
	on the microscopic side we have learned how to take into account subtle effects like wall-crossing to 
	isolate the states of the black hole from the full statistical ensemble of string theory. 
	The two complementary pictures of a black hole can now be recast as 
	the exact AdS$_2$/CFT$_1$ correspondence in the near-horizon region of the black hole \cite{Sen:2008vm},
	with two well-defined quantum systems each having its own rules of calculations. 
	This formulation has allowed us to quantitatively test the picture of a black hole as an ensemble of microstates 
	well beyond the thermodynamic approximation, and has provided examples of AdS/CFT valid at finite~$N$.
	
	Perhaps more importantly, the formulation of this correspondence has allowed us to ask sharp questions 
	on both sides of the story which would not have been possible earlier---each side of the correspondence 
	provides a novel guide for organizing the observables and calculations on the other side which may not be a priori
	obvious.
	In this context, we present in this paper an exact analytic formula for the integer quantum degeneracies of dyonic black holes
	in the four-dimensional~$\CN=4$ string theory arising from Type II string theory compactified on~K3$\times T^2$.
	The origins of this formula involve an intricate interplay between physical ideas (AdS$_2$/CFT$_1$, localization in supergravity, 
	instanton sums, black hole metamorphosis), 
	and mathematical ones (Siegel modular forms, mock Jacobi forms, and their Rademacher expansions). We do not have 
	a rigorous mathematical proof of our formula, but we are able to use the above ideas to obtain 
	a precise conjectural statement relating the Fourier coefficients of 
	the inverse of the Igusa cusp form~$1/\Phi_{10}$ and the Fourier coefficients of the power of the Dedekind 
	eta function~$1/\eta^{24}$, 
	which we have checked numerically to high order. In the rest of the introduction we present, successively, 
	the context and the physics motivation, the mathematical formula, 
	the idea of the calculation, and its interpretation in gravity.

	\subsection{Motivation and context}

	A prototype for an exact gravitational entropy formula can be found for~$\frac18$-BPS dyonic black holes in~$\CN=8$ 
	string theory (Type II string theory compactified on $T^6$) \cite{Maldacena1999}. 
	In that case the exact degeneracies of supersymmetric black holes are known to be coefficients of a certain 
	Jacobi form of weight~$-2$ and 
	index~$1$,~$ \displaystyle Z_{\CN=8}(\tau,z) = \sum_{n,\ell} C_{\CN=8}(n,\ell) \, e^{2\pi\mathrm{i} n \tau}\,e^{2\pi\mathrm{i} \ell z}$.  
	The black hole is labelled by the discriminant~$4n-\ell^2$ which grows, at large charges, as the square of the area of the horizon.
	The Hardy-Ramanujan-Rademacher formula for Jacobi forms 
	provides an exact analytic expression for the coefficients~$C_{\CN=8}(n,\ell)$ of this Jacobi form as an infinite  
	sum over Bessel functions with successively decreasing arguments. Importantly, this formula
	has no free parameters, and the only inputs are the modular transformation properties (the weight and index) 
	of the Jacobi form~$Z_{\CN=8}$ and its overall normalization which is fixed.\footnote{The Rademacher 
		formula, reviewed in Appendix \ref{app:Jac}, typically has a finite number of integers (the polar degeneracies) as input, 
		but in this case the large symmetry of the theory implies there is only one independent polar degeneracy which 
		can be normalized to one.} 
	Each term in the formula is then interpreted, via a localization calculation in the gravitational theory, as a functional integral 
	over the smooth fluctuations around certain asymptotically~AdS$_2$ 
	configurations~\cite{Dabholkar:2010uh,Dabholkar:2011ec,Dabholkar:2014ema}.
	
	We would now like to do the same for situations with less supersymmetry.
	The next-simplest example is~$\frac14$-BPS dyonic black holes in $\CN=4$ string theories, which are 
	labelled by the $T$-duality invariants of the charges~$n=Q^2/2$, $\ell = Q \cdot P$, $m=P^2/2$.
	(As above, the area of the horizon grows as~$\sqrt{\Delta}$  where the discriminant is~$\Delta = 4mn-\ell^2$.)
	In a class of~$\CN=4$ string theories the generating function of~$\frac14$-BPS states 
	is also known in terms of Siegel modular forms.
	The simplest example is Type II compactified on K3$\times T^2$ in which case the generating function of the supersymmetric index 
	that counts the~$\frac14$-BPS microstates is the inverse of the Igusa cusp form~$\Phi_{10}\tsz$~\cite{Dijkgraaf1997d}. 
	The main subtlety in~$\CN=4$ string theory compared to~$\CN=8$
	string theory is that, at strong coupling, the supersymmetric index receives contributions from~$\frac14$-BPS 
	single-centered black holes as well as bound states of two $\frac12$-BPS black 
	holes. 
	The question then arises to isolate the microstates that contribute to the single-centered black hole only. 
	Doing so breaks the modular invariance which was crucial to interpret the formula as a gravitational 
	functional integral. 
	
	It was shown in~\cite{Dabholkar2012} that one can isolate the degeneracies of single-centered black holes 
	in~$\CN=4$ string theory while keeping the essence 
	of modularity intact. More precisely, the degeneracies of single-centered dyonic black holes in Type II on K3$\times T^2$
	are the Fourier coefficient of certain \emph{mock} Jacobi forms. 
	One can calculate the mock Jacobi forms and their coefficients for any set of given charges,  
	using a computer algorithm. 
	This indeed clarifies the modular nature of the degeneracies of black hole microstates, 
	but we would like to do better and find an explicit formula for them, as in the~$\CN=8$ theory. 
	Upon applying the circle method of Hardy-Ramanujan-Rademacher 
	to the known modular completions of the mock Jacobi forms, one obtains an analytic formula for the 
	black hole degeneracies~\cite{Ferrari2017a} which is 
	similar to, but more complicated than, the one in the~$\CN=8$ theory---there are some additional terms 
	coming from the fact that one has mock Jacobi and not true Jacobi forms, but 
	the bottom line is that for a given mock Jacobi form one has an infinite series of terms controlled 
	purely by a finite number of integers.
	(The explicit formula is presented in~\eqref{eq:app-mixed-mock-coeffs}.)
	The integers in question are the polar coefficients of the mock Jacobi forms themselves. Here \emph{polar state} 
	(and correspondingly \emph{polar coefficient}) 
	means states with discriminant~$\Delta = 4mn - \ell^2<0$. 
	
	In this paper we present a simple analytic formula for the polar coefficients 
	in terms of the degeneracies 
	of~$\frac12$-BPS states in~$\CN=4$ string theory. In the heterotic duality frame, these are realized as perturbative
	fluctuations of the fundamental strings i.e.,~the Dabholkar-Harvey states~\cite{Dabholkar1989}.
	This means that the full quantum degeneracy of the black hole---which is a non-perturbative bound state of 
	strings, branes, and KK-monopoles---is completely controlled by simple perturbative elements of string theory!
	The nature of the formula, presented below in~\eqref{eq:OSV-type}, is also noteworthy: the polar coefficients of the 
	mock Jacobi forms are simply linear combinations of quadratic functions of the~$\frac12$-BPS degeneracies. 
	The latter can be interpreted as counting worldsheet instantons or, more precisely, genus-one Gromov-Witten invariants.
	This structure is clearly reminiscent of the OSV formula~$Z_{\mathrm{BH}} = |Z_{\mathrm{top}}|^2$~\cite{Ooguri:2004zv},
	but the details are somewhat different. 
	The right-hand side of our formula, involving instanton degeneracies, is controlled by 
	$Z_{\mathrm{top}}$, while the left-hand side is the  ``seed" for the~$\frac14$-BPS BH degeneracies 
	via an intricate series which is dictated by the mock modular symmetry.  
	The idea of exploiting modular symmetry in order to reach a precise non-perturbative definition 
	for the OSV formula was already initiated in~\cite{Denef:2007vg}, but the technical complications of~$\CN=2$ string theories
	did not allow for an explicit formula. Here we use the fact that many aspects of~$\CN=4$ string theories are solvable in order 
	to reach such a formula.

	\subsection{The main formula}

	In order to present our main formula we briefly review the procedure~\cite{Dabholkar2012} 
	to calculate the single-centered black hole degeneracies in Type II string theory on~K3$\times T^2$. 
	One first expands the partition function in the chemical potential conjugate to the magnetic charge invariant~$m$ to obtain
	the Fourier-Jacobi expansion
	\begin{equation}
	\label{eq:Fourier-Jacobi-intro}
	\frac{1}{\Phi_{10}(\tau,z,\sigma)} \= \sum_{m\geq-1}\,\psi_m(\tau,z)\,e^{2\pi\mathrm{i} m \sigma} \, .
	\end{equation}
	The Igusa cusp form~$\Phi_{10}$ is a Siegel modular form of weight~$10$ which implies that the 
	functions~$\psi_m$ are \emph{meromorphic} Jacobi forms of weight $-10$ and index $m$. 
	Since~$\Phi_{10}$ has a double zero at~$z=0$, $\psi_m$ is meromorphic in~$z$ with a double pole at~$z=0$.
	This meromorphicity has its physical origin in the wall-crossing phenomenon, whereupon bound states 
	of two~$\frac12$-BPS centers appear or decay as one moves around the moduli space of the compactification. 
	It was shown in~\cite{Dabholkar2012} that the functions $\psi_m$ have a canonical decomposition into two pieces,
	\begin{equation}
	\psi_m(\tau,z) \= \psi_m^\mathrm{F}(\tau,z) + \psi_m^\mathrm{P}(\tau,z) \, ,
	\end{equation}
	where~$\psi_m^{\mathrm{F}}$ and~$\psi_m^{\mathrm{P}}$ count the degeneracies of dyonic~$\frac14$-BPS
	\emph{single-centered} black holes, and two-centered~$\frac12$-BPS black hole bound states, respectively.
	Further, the function $\psi_m^{\mathrm{F}}$ is a \emph{mock} Jacobi form~\cite{Dabholkar2012,Zwegers:2008zna,zagierramanujan}, 
	which is holomorphic in $z$. 
	This implies a Fourier expansion of the form
	\begin{equation}
	\label{eq:psi-F-exp}
	\psi_m^\mathrm{F}(\tau,z) \= \sum_{n,\ell}\,c_m^\mathrm{F}(n,\ell)\,e^{2\pi\mathrm{i} n \tau}\,e^{2\pi\mathrm{i} \ell z} \, .
	\end{equation}
	The microsocpic degeneracies of~$\frac14$-BPS single-centered black holes are related to these Fourier coefficients as
	\be
	d^\text{BH}_\text{micro} (n,\ell,m) \=  (-1)^{\ell+1} \, c_m^\mathrm{F}(n,\ell)  \quad \text{for} \quad  \Delta \= 4 m n -\ell^2 > 0 \,.
	\ee
	
	We now present the analytic formula for the black hole degeneracies which is a combination of 
	the following two formulas:
	\begin{enumerate}
		\item The BH coefficients~$c_m^\mathrm{F}(n,\ell)$,~$\Delta > 0$ are completely controlled by the polar 
		coefficients~$c_m^\mathrm{F}(n,\ell)$,~$\Delta<0$. The relevant formula follows from the ideas of Hardy-Ramanujan-Rademacher
		applied to mock Jacobi forms, which by now has become a well-established technique in analytic number 
		theory~\cite{bringmann2006f,Bringmann:2010sd,bringmann2012coefficients}. 
		We review this in Section \ref{sec:exact-entropy}. For the particular mock
		Jacobi forms~$\psi_m^\mathrm{F}$ the formula was obtained in~\cite{Ferrari2017a}, which we recall 
		in~\eqref{eq:app-mixed-mock-coeffs}.
		\item The polar coefficients~$c_m^\mathrm{F}(n,\ell)$, $\Delta<0$ are given by 
		\be \label{eq:OSV-type}
		c_m^\mathrm{F}(n,\ell) \= \sum_{\gamma \in \gset} (-1)^{\ell_\gamma+1}\, |\ell_\gamma|\,d (m_\gamma)\,
		d (n_\gamma) \quad \text{for} \quad \Delta \= 4 m n -\ell^2 < 0 \, .
		\ee
		We obtain this formula using the ideas and results of~\cite{Sen2011}, by tracking all possible ways that 
		a two-centered black hole bound state of total charge~$(n,\ell,m)$ decays into its constituents across a wall of 
		marginal stability. 
		Here~$\gset$ is a set of~$SL(2,\IZ)$ matrices that encode the relevant Walls of marginal stability. 
		This set is finite and we will spend a large part of the paper characterizing this set. The 
		precise formulas are given in~\eqref{eq:degwmeta}, \eqref{eq:degwmetaW}.
		The quantities $(n_\gamma, \ell_\gamma, m_\gamma)$ are the $T$-duality invariants of the charges~$(Q,P)$
		transformed by~$\gamma$,
		and~$d(n)$ is the degeneracy of~$\frac12$-BPS states with charge invariant~$n$, given by~\cite{Dabholkar1989}
		\be
		\label{eq:Fourier-eta}
		\frac{1}{\eta(\tau)^{24}} \= \sum_{n=-1}^\infty d(n) \, e^{2 \pi \i n \tau} \,.
		\ee 
	\end{enumerate}
	We have checked the main formula~\eqref{eq:OSV-type} against the polar coefficients of $\psi_m^\mathrm{F}$ 
	extracted from the Igusa cusp form~$\Phi_{10}$ using formula~\eqref{eq:Fourier-Jacobi-intro} for magnetic charge 
	invariant up to $m=30$, which corresponds to 1650 coefficients. In order to extract the polar coefficients 
	from~$\Phi_{10}$ the steps required are to (a) build~$\Phi_{10}$ from the \emph{additive lift}~\cite{eichler1985theory}, 
	(b) inverting it, and (c) extracting the relevant polar coefficients. The main computational bottleneck in this procedure 
	is the inversion. Using recursion relations~\cite{Dabholkar2012}, which is much faster than a direct division, already 
	took us a computing time of the order of hours on a MacBookPro.
	In contrast, the right-hand side of the formula~\eqref{eq:OSV-type} for a given value of~$m$ can be computed in milliseconds 
	on the same computer, which is a factor of~$\mathcal{O}(10^5)$.

	\subsection{The idea of the calculation}

	When the charges have a negative discriminant they cannot form a single-centered black hole 
	(recall that the discriminant is proportional to the square of the classical horizon area).
	We know that the only other configurations that contribute to the~$\frac14$-BPS index in~$\CN=4$
	string theory are two-centered bound states of~$\frac12$-BPS black holes~\cite{Dabholkar:2009dq,Sen2011}.
	Thus the problem becomes one of calculating all possible ways a given set of charges with negative 
	discriminant contributing to~$c_m^\mathrm{F}$ can be represented as two-centered black hole bound states. 
	
	Now, any such bound state is an~$S$-duality ($SL(2,\IZ)$) transformation of the basic bound state,
	which consists of an electrically charged~$\frac12$-BPS black hole with invariant~$n=Q^2/2$, 
	a magnetically charged~$\frac12$-BPS black hole with invariant~$m=P^2/2$, and the electromagnetic
	fields carry angular momentum~$\ell = Q \cdot P$.
	The indexed degeneracy of this system equals 
	\be \label{basicdeg}
	(-1)^{\ell+1}\, |\ell| \cdot d(n) \cdot d(m) \, .
	\ee
	The factors~$d(n)$ and~$d(m)$ in this formula are, respectively, the internal degeneracies of the electric
	and magnetic~$\frac12$-BPS black holes,
	and the factor~$(-1)^{\ell+1} |\ell |$ is the indexed number of supersymmetric ground states of the 
	quantum mechanics of the relative motion between the two centers~\cite{Denef:2000nb}.
	The degeneracy of an arbitrary bound state can be calculated by acting on the charges~$(Q,P)$ 
	by the appropriate~$S$-duality transformation and replacing the charge invariants in~\eqref{basicdeg} by their 
	transformed versions. This is precisely the structure of the formula~\eqref{eq:OSV-type}. 
	
	The final ingredient of the formula is to state precisely what are the allowed values of~$\gamma$ which labels  
	all possible bound states. A very closely related problem was solved in an elegant manner in~\cite{Sen2011},
	which we use after making small adaptations (note that the modular and elliptic structures are manifest in our presentation).
	The basic intuition comes from particle physics---any 
	bound state must decay into its fundamental constituents somewhere, and so the question of which bound 
	states exist is the same as the question of what are all the possible decays of two-centered~$\frac12$-BPS black holes. 
	As was shown in~\cite{Sen2011} the possible decays are labelled by a certain set of~$SL(2,\IZ)$ matrices. 
	The exact nature of this set is a little subtle due to a phenomenon called black hole bound state 
	metamorphosis (BSM)~\cite{Andriyash:2010yf,Sen2011,Chowdhury:2012jq}, which identifies different-looking physical configurations
	with each other. This is the step which lacks a rigorous mathematical proof, but the physical picture 
	is well-supported. The sum over~$\gset$ in~\eqref{eq:OSV-type} is precisely the sum over all possible 
	decay channels after taking metamorphosis into account. 
	Thus our checks of the formula~\eqref{eq:OSV-type} can be thought of as providing 
	more evidence for the phenomenon of metamorphosis. 
	
	The metamorphosis can be of three types: electric, magnetic, and dyonic. The corresponding identifications generate 
	orbits of length~2 in the first two cases and of infinite length in the third. In the first two cases the metamorphosis has a 
	simple~$\IZ/2\IZ$ structure, while the group structure of the dyonic case was less clear so far.  
	We show in this paper that the identifications 
	due to dyonic BSM have a group structure of~$\IZ$. Moreover, the problem of finding BSM orbits maps precisely to finding the 
	solutions to a well-studied	Diophantine equation, namely the Brahmagupta-Pell equation, whose structure is completely known. 
	In the language of algebraic number theory, this is the problem of finding the group of units in the order~$\IZ[\sqrt{|\Delta|}]$
	in the real quadratic field~$\mathbb{Q}(\sqrt{|\Delta|})$.

	\subsection{Gravitational intepretation}

	Recall that the quantum entropy of the gravitational theory is formulated as a functional integral over asymptotically~AdS$_2$
	configurations. Using the technique of supersymmetric localization in the variables of supergravity, a formula for the 
	exponential of the quantum entropy was derived in \cite{Dabholkar:2010uh,Dabholkar:2011ec, Dabholkar:2014ema}. 
	The result takes the form of an infinite sum of 
	finite-dimensional integrals over the (off-shell) fluctuations of the scalar fields 
	around the attractor background, where the integrand includes a tree-level and a one-loop factor in the off-shell theory. 
	The infinite sum is interpreted as different orbifold configurations in string theory with the same~AdS$_2$ 
	boundary~\cite{Dabholkar:2014ema}. 
	In the $\mathcal{N}=8$ theory, this result agrees exactly with the Rademacher expansion for the coefficient of the Jacobi form
	controlling the microscopic index. 
	
	We can now offer a physical interpretation of our exact degeneracy 
	formula~\eqref{eq:app-mixed-mock-coeffs},~\eqref{eq:OSV-type}
	from this point of view. The sum over~$k$ in~\eqref{eq:app-mixed-mock-coeffs} runs 
	over all positive integers with the argument of the Bessel function suppressed as~$1/k$ 
	and the Kloosterman sum depending on~$k$. This part of the structure comes 
	from a sum over~$\Gamma_\infty\backslash SL(2,\IZ)$ of the circle method, and 
	can be interpreted in the gravitational theory exactly as in the~$\CN=8$ theory,
	namely as a sum over orbifolds of the type (AdS$_2 \times S^1\times S^2)/\IZ_k$~\cite{Dabholkar:2014ema}.
	The Kloosterman sum arises from an analysis of Chern-Simons terms in the full geometry. 
	The degeneracies of polar states~$c^\text{F}_m(n,\ell)$ (with~$\Delta = 4mn-\ell^2 <0$) 
	is interpreted as the number of states of a given~$(n,\ell,m)$ which do \emph{not} 
	form a big single-centered BH. The finite sum over~$\gamma \in \gset$ in~\eqref{eq:OSV-type}
	is indicative of a further fine structure where the smallest units are the~$\frac12$-BPS instanton 
	states with their corresponding degeneracy. This is the sense in which the final degeneracy 
	formula is constructed out of the instantonic elements. 
	
\vspace{0.2cm}
	
	The outline of the paper is as follows: In Section~\ref{sec:exact-entropy} we explain in detail why one can 
	reduce the counting of \qbps states to counting bound states of \hbps states. In Section~\ref{sec:localization} 
	we discuss the macroscopic supergravity counting of \qbps states. Section~\ref{sec:neg-disc} discusses 
	several details that are important for the proper counting of negative discriminant states. In Sections~\ref{sec:negwomet} 
	and~\ref{sec:metamorphosis} we present all the relevant calculations that will lead to the explicit formula for 
	the negative discriminant states. In Section~\ref{sec:metamorphosis}, we characterize the orbits 
	of dyonic metamorphosis in terms of the orbits of the solutions to the Brahmagupta-Pell equation.
	The final Section~\ref{sec:expchecks} presents our exact black hole formula 
	and lists some numerical data that shows its validity. In Appendix~\ref{app:Jac} we review the details of the 
	Rademacher formula applied to our case of interest. Appendix~\ref{sec:finiteness} provides numerical evidence 
	for one special case that we could not solve analytically. Lastly, in Appendix~\ref{sec:furtherchecks} we provide 
	further explicit data for the interested reader.

	\section{Exact dyonic black hole degeneracies and the attractor region}
	\label{sec:exact-entropy}

	In this section we first review the microscopic counting formula for single-centered \qbps states 
	in~$\CN=4$ string theory. We then explain how an exact analytic formula for the corresponding 
	black hole degeneracies reduces to the problem of counting bound states of \hbps centers, and how this problem can be 
	efficiently solved using the results of~\cite{Sen2011}. \\
	
	Four-dimensional $\mathcal{N}=4$ string theory can be described either in terms of heterotic string theory compactified 
	on~$ \displaystyle T^6 $, or in another duality frame in terms of Type II string theory compactified on K3$\times T^2$.
	Dyonic states are charged under the 28 $U(1)$ gauge fields, with the charge vector $(Q,P)$ taking values in 
	the integral second cohomology lattice of $\Gamma_{6,22} \oplus \Gamma_{6,22}$. The theory 
	has $S$-duality group $SL(2,\mathbb{Z})$ and $T$-duality group $O(6,22,\mathbb{Z})$. 
	The $T$-duality invariants are
	\begin{equation}
	(n,\ell,m) \defeq (Q^2/2,\,Q\cdot P,\,P^2/2) \, . 
	\end{equation}
	The dyonic charges $(Q,P)$ transform as a doublet under $S$-duality, and the 
	discriminant~$\Delta = 4mn -\ell^2$ is invariant under $U$-duality. \\
	
	The microscopic degeneracies of dyonic \qbps 
	states\footnote{Here and in the following, we refer to dyons with torsion 1, i.e. $\operatorname{gcd}\left\{Q_{i} P_{j}-Q_{j} P_{i} \, , \; 1 \leq i, j \leq 28\right\}=1$. 
	A similar story for generic dyons \cite{Banerjee:2008ri,Dabholkar:2008zy} should follow along the same lines.} 
	in the above theory are given by a Fourier transform of the inverse of the Igusa 
	cusp form~$\Phi_{10}$, the unique automorphic form of weight 10 defined on~$\displaystyle Sp(2, \Z)$ 
	\cite{Dijkgraaf1997d,Dijkgraaf:1996xw,Shih:2005uc,David:2006yn}
	\begin{equation}
	\label{eq:DVV}
	d_{\frac14}(Q,P) \=(-1)^{\ell+1}  \int_{\mathcal{C}} \mathrm{d} \tau \, \mathrm{d} \sigma \, \mathrm{d} z \, 
	e^{-2 \pi \mathrm{i}\left(\tau n + z \ell + \sigma m\right)} \, \frac{1}{\Phi_{10}\tsz} \, .
	\end{equation}
	Here, $\mathcal{C}$ indicates a certain contour in the three complex dimensional space spanned by $ \tsz$, so that 
	the degeneracies on the left-hand side depend on this contour (this dependence has been suppressed in the notation). 
	Importantly, the contour~$\mathcal{C}$ depends on the moduli of the compactification \cite{Cheng2007}. 
	When moving through the moduli space for a fixed set of charges~$(Q,P)$, the degeneracies 
	jump when $\mathcal{C}$ crosses a pole in the partition function $\Phi_{10}^{-1}$.
	This is a manifestation of the \emph{wall-crossing phenomenon} where a \qbps bound state of two \hbps states appears or 
	decays upon crossing codimension-one surfaces in the moduli space.
	Thus, the moduli space is divided into chambers separated by walls of marginal stability. 
	In a given chamber the degeneracies for a given $(Q,P)$ are constant, and they jump as one crosses 
	a pole in moving to another chamber. This phenomenon will be central to our investigations in the rest of the paper, 
	and we will discuss the contour~$\mathcal{C}$ in more detail below.

	\subsection*{Single- and multi-centered degeneracies}

	In the macroscopic supergravity description, the gravitational configurations captured by the 
	index \eqref{eq:DVV} correspond to either (a) \qbps single-centered black holes, or (b) \qbps bound states 
	of two \hbps black holes \cite{Dabholkar:2009dq}.\footnote{Multi-centered 
		black hole bound states exist in theories with any number of supercharges. A bulk-analysis of preserved and 
		broken supersymmetry \cite{Dabholkar:2009dq} shows that only certain types of configurations contribute to the relevant indices:
		single-centered~$\frac18$-BPS BHs contribute to the $\frac18$-BPS index in~$\CN=8$ string theory, 
		single centered~$\frac14$-BPS BHs and two-centered bound states of $\frac12$-BPS BHs contribute to 
		the $\frac14$-BPS index in~$\CN=4$ string theory, and all single and multi-BH bound states which are $\frac12$-BPS
		contribute to the $\frac12$-BPS index in~$\CN=2$ string theory. This makes~$\CN=4$ string theory a simple 
		starting point to analyze effects of black hole bound states on the index.
	} 
	The bound states 
	exist depending on the region of the moduli space and the values of the charges $(Q,P)$ \cite{Denef:2000nb}, in accordance 
	with the wall-crossing phenomenon discussed above. In contrast, single-centered solutions exist everywhere in the 
	moduli space provided the $U$-duality invariant $\Delta$ is
	positive. For large values of~$\Delta$, this is consistent with the semi-classical picture of BHs where 
	the classical area of the black hole is~$4\pi \sqrt{\Delta}$. The conjecture of~\cite{Sen2009} extrapolates this intuition to 
	all positive values of~$\Delta$.
	Far away from the black hole, the massless 
	moduli can take any value, but the attractor mechanism~\cite{Ferrara1995} implies that they are ``attracted'' to 
	values that are completely determined by the charges near the black hole horizon. In the $\mathcal{N}=4$ string theory under consideration, this is 
	elegantly captured by the attractor contour \cite{Cheng2007}
	\begin{equation}
	\label{eq:attractor-contour}
	\mathcal{C}_* = \{ \mathrm{Im}(\tau) = 2m/\varepsilon,\,\mathrm{Im}(\sigma) = 2n/\varepsilon,\,\mathrm{Im}(z) = -\ell/\varepsilon,\;\; 0 \leq \mathrm{Re}(\tau),\mathrm{Re}(\sigma),\mathrm{Re}(z) < 1 \} \, ,
	\end{equation}
	where $\varepsilon \rightarrow 0^+$,
	so that the single-centered degeneracies $d_*(Q,P)$ evaluated using \eqref{eq:DVV} 
	and the contour $\mathcal{C}_*$ are functions of the charges $(Q,P)$ only. \\
	
	It was shown in \cite{Dabholkar2012} that these single-centered degeneracies are Fourier coefficients of certain mock Jacobi forms $\psi_m^\mathrm{F}$. It is important to note however that the converse is not true, namely that not all coefficients of $\psi_m^\mathrm{F}$ are degeneracies of single-centered black holes, and this will play an important role in what follows. 
	To construct $\psi_m^\mathrm{F}$, one begins by 
	expanding the partition function $\Phi_{10}^{-1}$ around the $\sigma \rightarrow \i\infty$ point,
	\begin{equation}
	\label{eq:Fourier-Jacobi}
	\frac{1}{\Phi_{10}(\tau,z,\sigma)} \= \sum_{m\geq-1}\,\psi_m(\tau,z)\,e^{2\pi\mathrm{i} m \sigma} \, .
	\end{equation}
	The functions $\psi_m(\tau,z)$ in this expansion are Jacobi forms of weight $-10$ and index $m$ that are \emph{meromorphic}\footnote{This meromorphicity descends from the poles in the \qbps states partition function responsible for the wall-crossing phenomenon discussed above.} in the $z$ variable.
	The attractor contour \eqref{eq:attractor-contour} then shows that to extract~$d_*(Q,P)$ from $\psi_m$, the inverse Fourier transform in $z$ should be taken along a path such that 
	\begin{equation}
	\label{eq:DMZ-contour}
	\text{Im}(z)/\text{Im}(\tau) \= -\ell/2m \, .
	\end{equation}
	This contour was called the ``attractor contour'' in~\cite{Dabholkar2012}, and applies to general meromorphic 
	Jacobi forms of index~$m$. Without loss of generality one can choose\footnote{This is allowed since the 
		full physical system has a symmetry that exchanges $\tau$ and $\sigma$. For our purposes, it will be 
		convenient to work at fixed magnetic charge $m$.} $n>m$, 
	and furthermore the fact that the degeneracies~$d_*(Q,P)$ are invariant under spectral flow 
	implies that we can restrict~$\ell$ to the window $0 \leq \ell < 2m$.
	Taking the remaining inverse Fourier transforms,~\cite{Dabholkar2012} then showed that 
	the single-centered degeneracies~$d_*(Q,P)$ 
	are the Fourier coefficients of the so-called \emph{finite part} of $\psi_m$, 
	defined as 
	\begin{equation}
	\label{eq:psi-F}
	\psi_m^\mathrm{F}(\tau,z) \defeq \psi_m(\tau,z) - \psi_m^\mathrm{P}(\tau,z) \, ,
	\end{equation}
	where
	\begin{equation}
	\label{eq:psi-P}
	\psi_m^\mathrm{P}(\tau,z) \= \frac{d(m)}{\eta(\tau)^{24}}\,\sum_{s\in\mathbb{Z}}\,\frac{q^{m s^2+s}\,\zeta^{2ms + 1}}{(1-q^s \zeta)^2} \, .
	\end{equation}
	Above, $q := e^{2\pi\mathrm{i}\tau}$, $\zeta :=e^{2\pi\mathrm{i}z}$ and $d(m)$ is defined in \eqref{eq:Fourier-eta}. The Appell-Lerch sum in \eqref{eq:psi-P} exhibits wall-crossing since its Fourier expansion differs in the strips~$\alpha -1 < \mathrm{Im}(z)/\mathrm{Im}(\tau) \leq \alpha$, with $\alpha \in \mathbb{Z}$. Subtracting $\psi_m^\mathrm{P}$ from the functions $\psi_m$ in \eqref{eq:Fourier-Jacobi} implies that the resulting finite part $\psi_m^\mathrm{F}$ is \emph{holomorphic} in $z$ and as such has an unambiguous Fourier expansion
	\begin{equation}
	\label{eq:Fourier-psi-F}
	\psi_m^\mathrm{F}(\tau,z) \= \sum_{n,\ell}\,c_m^\mathrm{F}(n,\ell)\,q^n\,\zeta^\ell \, .
	\end{equation}
	The discussion of single-centered degeneracies so far can then be summarized as
	\begin{equation}
	\label{eq:SC-degen}
	d_*(Q,P) \= (-1)^{\ell + 1}\,c_m^\mathrm{F}\bigl(n,\ell\bigr) \quad \text{for} \quad \Delta = 4mn - \ell^2 > 0 \, .
	\end{equation}
	The central mathematical result of \cite{Dabholkar2012} is that $\psi_m^\mathrm{F}$ is a \emph{mock} Jacobi form. This means that 
	its usual modular behavior under $SL(2,\mathbb{Z})$ is modified. 
	To salvage modularity it is possible to add a correction term, known as the \emph{shadow}, to build a 
	function $\widehat{\psi}_m^\mathrm{F}$ that is modular. The shadow is however non-holomorphic in $\tau$, 
	so modularity is restored at the expense of holomorphicity \cite{zagierramanujan,Zwegers:2008zna}.

\vspace{0.2cm}
	
	Starting from the completion $\widehat{\psi}_m^\mathrm{F}$, the Fourier coefficients $c_m^\mathrm{F}(n,\ell)$ for $\Delta > 0$ can be computed using a generalization of the Hardy-Ramanujan-Rademacher formula suited for mixed mock modular forms~\cite{Bringmann:2010sd,Ferrari2017a}. We briefly review this in Appendix \ref{app:Jac} and the final result is presented in \eqref{eq:app-mixed-mock-coeffs}. One of the main points of the formula is that, in order to compute the Fourier coefficients of positive discriminant states entering \eqref{eq:SC-degen}, the only required input are the \emph{polar coefficients} of $\psi_m^\mathrm{F}$, defined as
	\begin{equation}
	\label{eq:polar-def}
	\widetilde{c}_m(n,\ell) \defeq c_m^\mathrm{F}(n,\ell) \quad \text{for} \quad \Delta = 4mn - \ell^2 < 0 \, .
	\end{equation}
	By construction, the polar coefficients $\widetilde{c}_m(n,\ell)$ count the number of negative discriminant states encoded in the generating function $\psi_m^\mathrm{F}$. They will be the central focus of the present paper, and we will give an explicit formula for them based on a careful analysis of wall-crossing and bound states.

	\subsection*{The moduli space and the attractor region}

	Having reviewed the \qbps single-centered degeneracies, we now discuss in a bit more detail the structure of walls in the moduli space. 
	The moduli-dependent contour $\mathcal{C}$ in \eqref{eq:DVV} can be written in terms of the moduli-dependent 
	central charge matrix $\mathcal{Z}$ \cite{Cheng2007}. The latter can be parameterized\footnote{This parametrization 
		corresponds to a projection from the full moduli space to the two-dimensional axio-dilaton moduli space of 
		the heterotic frame.} 
	by a complex scalar $\Sigma = \Sigma_1 + \mathrm{i}\Sigma_2$ as
	\begin{equation} \label{eq:defZSigma}
	\mathcal{Z} \= \Sigma_2^{-1} \begin{pmatrix} |\Sigma|^2 & \Sigma_1 \\ \Sigma_1 & 1 \end{pmatrix} \, .
	\end{equation}
	In terms of this matrix, the contour in \eqref{eq:DVV} reads (with $\varepsilon \rightarrow 0^+$)
	\begin{equation}
	\label{eq:moduli-contour}
	\mathcal{C} = \{ \mathrm{Im}(\tau) = \Sigma_2^{-1}/\varepsilon,\,\mathrm{Im}(\sigma) = \Sigma_2^{-1}|\Sigma|^2/\varepsilon,\,\mathrm{Im}(z) = -\Sigma_2^{-1}\Sigma_1/\varepsilon,\;\; 0 \leq \mathrm{Re}(\tau),\mathrm{Re}(\sigma),\mathrm{Re}(z) < 1 \} \, .
	\end{equation}
	The expansion~\eqref{eq:Fourier-Jacobi} then corresponds to taking the limit~$\Sigma_2 \rightarrow \infty$ while keeping~$\Sigma_1$ and~$\varepsilon\,\Sigma_2$ fixed. This limit has a physical interpretation as the M-theory limit, where one of the circles inside the internal $T^2$ of the Type II frame becomes large~\cite{Dabholkar2012}. In this limit, the expansion~\eqref{eq:Fourier-Jacobi} around~$\sigma \rightarrow \mathrm{i}\infty$ takes us high into the upper half-plane parameterized by~$\Sigma$, and varying~$\Sigma_1$ moves us horizontally. This is depicted in Figure~\ref{fig:M-theory-limit}.
	\begin{figure}
	\centering
		\begin{subfigure}{0.4\textwidth}
		\hspace{-1.5cm}
			\includegraphics[width=1.6\linewidth]{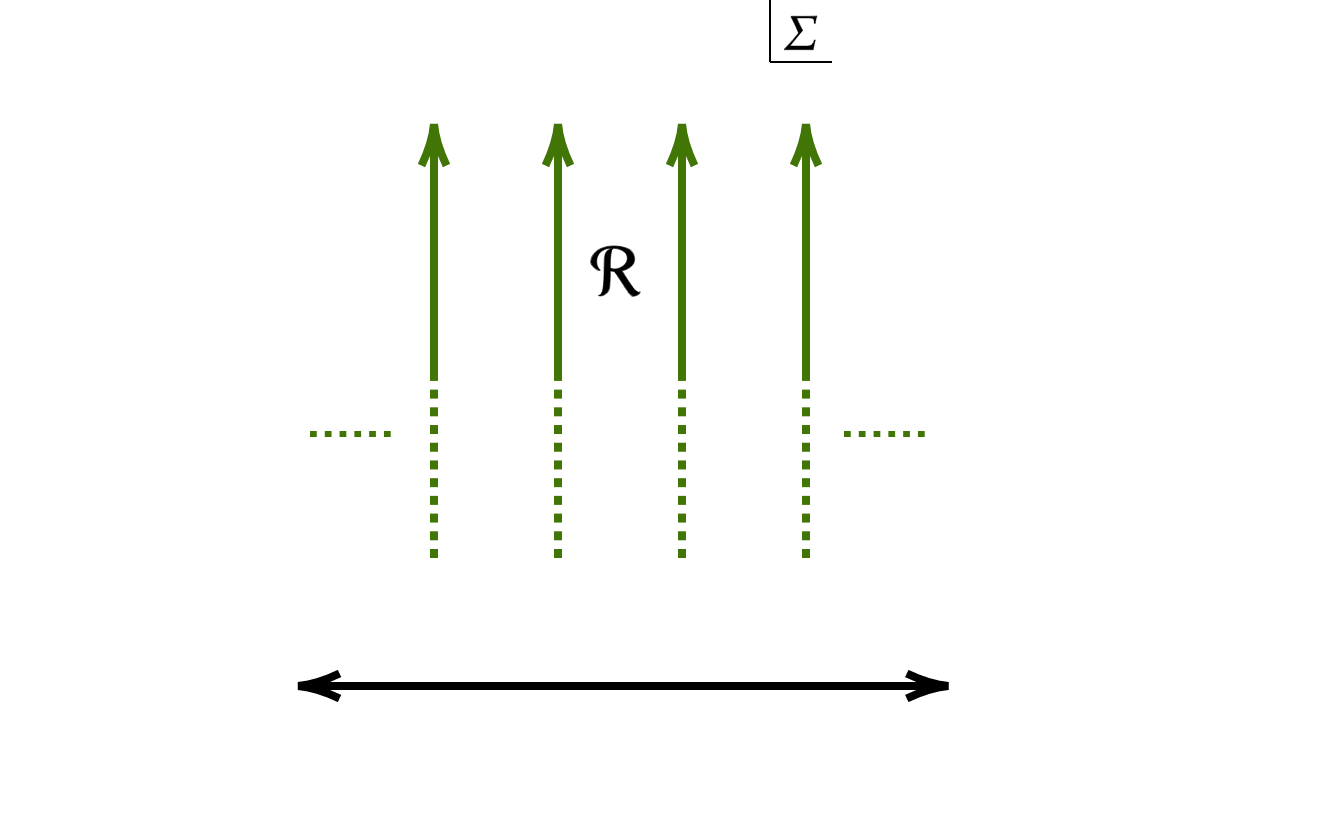}
			\caption{The M-theory limit \eqref{eq:Fourier-Jacobi}}
			\label{fig:M-theory-limit}
		\end{subfigure}
		\hspace{1cm}
		\begin{subfigure}{0.4\textwidth}
		\hspace{-1.5cm}
			\includegraphics[width=1.6\linewidth]{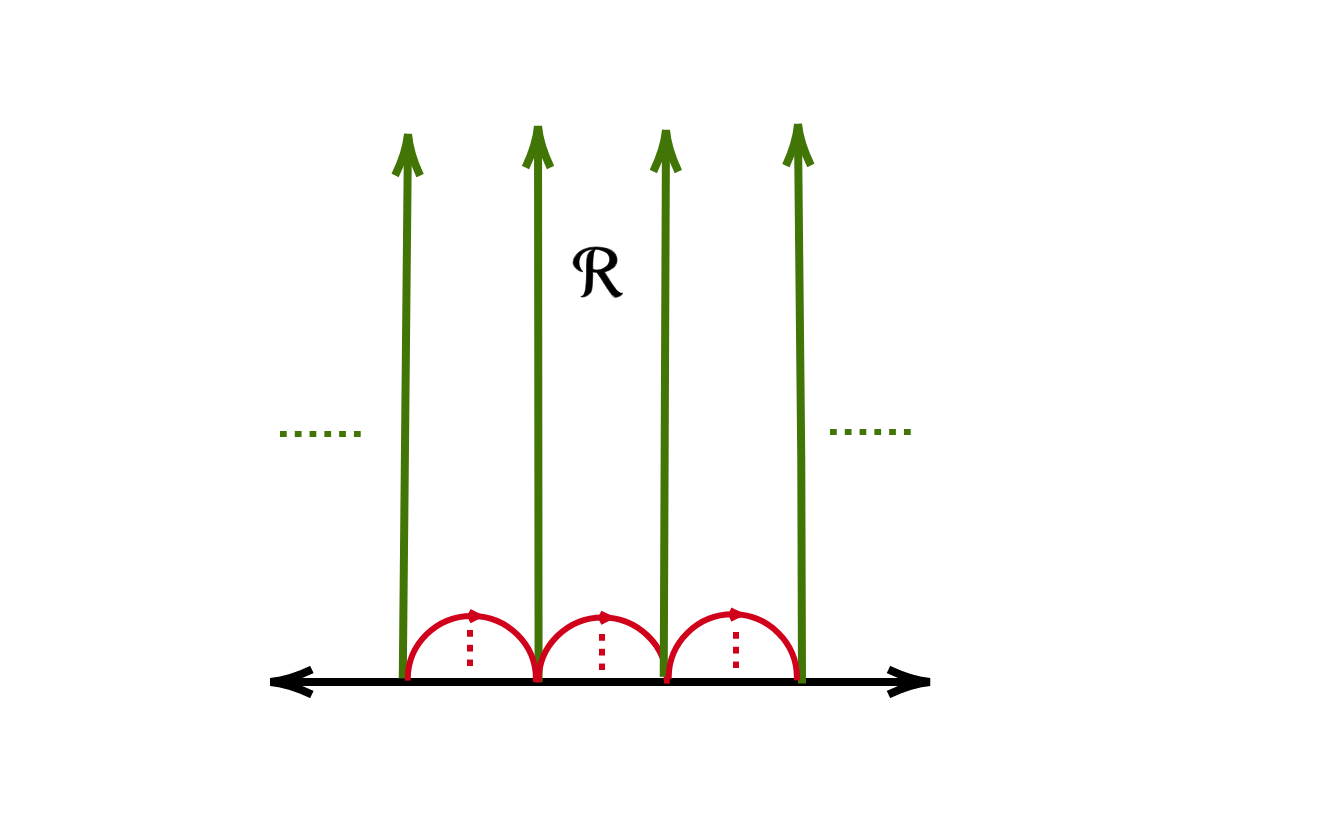}
			\caption{The full~$\Sigma$ moduli space}
			\label{fig:full-Sigma}
		\end{subfigure}
		\caption{The region~$\mathcal{R}$ in the moduli space}
	\end{figure}
	The wall-crossing captured, in the M-theory limit, by the Appell-Lerch sum \eqref{eq:psi-P} divides the moduli space into chambers separated by parallel marginal stability walls located at $\Sigma_1 = \alpha \in \mathbb{Z}$. The attractor contour \eqref{eq:DMZ-contour} then corresponds to picking a particular chamber, which we denote by $\mathcal{R}$. As mentioned below~\eqref{eq:DMZ-contour},~$\ell$ can be restricted to $0\leq \ell < 2m$, so it follows that $\mathcal{R}$ is the chamber between the walls located at $\alpha = 0$ and $\alpha = -1$. Thus,
	\begin{equation}
	\label{eq:R-region}
	\mathcal{R}\; : \;  -\text{Im}(\tau) < \text{Im}(z) \leq 0 \, .  
	\end{equation}
	
	Now, in the chamber~$\mathcal{R}$, the Fourier coefficients of~$\psi_m^\mathrm{P}$ in the range~$0 \leq \ell < 2m$ vanish because the coefficients of the Appell-Lerch sum vanish in this chamber, as can easily be checked. Therefore, the Fourier coefficients of the meromorphic Jacobi forms~$\psi_m$ in the chamber~$\mathcal{R}$ are equal to the coefficients of the finite part~$\psi_m^\mathrm{F}$. For~$\Delta > 0$, these coefficients correspond precisely to the single-centered black hole degeneracies, as given by \eqref{eq:SC-degen}. However,~$\psi_m^\mathrm{F}$ also has coefficients with~$\Delta < 0$ which live in the chamber~$\mathcal{R}$. Thus, we arrive at the following physical interpretation of the Fourier coefficients~$c_m^\mathrm{F}$: they count the indexed number of \qbps dyonic states in the chamber~$\mathcal{R}$. This is true for~$\Delta > 0$ (which are single-centered black holes) as well as, importantly, for~$\Delta < 0$ (which correspond to multi-centered black holes).
	
	Because of the Rademacher formula~\eqref{eq:app-mixed-mock-coeffs}, our task of finding an analytic formula for the single-centered degeneracies~\eqref{eq:SC-degen} is thus reduced to finding an analytic formula for the negative discriminant degeneracies, for which we will use the above physical picture. As explained in the introduction, the \qbps bound states counted by~\eqref{eq:polar-def} always decay upon crossing a wall of marginal stability. We therefore have to track them in the~$\Sigma$ moduli space and reconstruct them as a sum of their \hbps constituents. The decays corresponding to moving horizontally by varying~$\Sigma_1$ have been taken into account by~$\psi_m^\mathrm{P}$, so they will not contribute to~$c_m^\mathrm{F}(\Delta < 0)$. 
	As discussed above, the contribution from~$\psi_m^\mathrm{P}$ in the 
	region~$\mathcal{R}$ actually vanishes.\footnote{If we are however interested in another chamber of the moduli space where the Appell-Lerch sum does not vanish, it is important to remove the associated decays taking place when varying~$\Sigma_1$.} This is consistent with the analysis of~\cite{Sen2011}. 
	In addition, there can be other decays in the full moduli space, and to see them we need to go away from the M-theory limit.
	As shown in \cite{Sen2011}, the region $\mathcal{R}$ extends also vertically downwards back to the $\Sigma_2 = 0$ line. Therefore, away from the M-theory limit, the negative discriminant states contained in $\mathcal{R}$ can also decay further upon crossing circular walls (the precise shape of these walls is charge-dependent), which are shown in Figure~\ref{fig:full-Sigma}. We can now use the results of~\cite{Sen2011} to count how many negative discriminant states live in the region $\mathcal{R}$ and obtain the polar coefficients \eqref{eq:polar-def}. This will be reviewed in Section~\ref{sec:neg-disc}.
	
	\section{Localization of $ \displaystyle \mathcal{N} = 4 $ supergravity and black hole degeneracy }
	\label{sec:localization}
	
	Before presenting the derivation of the formula for the polar coefficients~$\widetilde{c}_m(n,\ell)$ defined in~\eqref{eq:polar-def}, 
	we review how the main idea originates from physical considerations.\footnote{This section can be skipped 
		without losing any of the logical steps, but it may provide some intuition towards the main result.}
	In~\cite{Murthy:2015zzy} two of the present authors were able to compute the asymptotic degeneracies of \qbps 
	single-centered BHs as a supergravity functional integral in the AdS$_2$ near-horizon geometry of the BHs, following the 
	ideas of~\cite{Sen:2008vm, Dabholkar:2010uh}.
	This computation relied on some approximations and assumptions that we spell out below, and as such did not yield 
	the exact answer matching the microscopic prediction. It did, however, lead to a result that could be interpreted as 
	an approximate relation between the polar coefficients of the counting function~$\psi_m^{\mathrm{F}}$ and the 
	Fourier coefficients of the Dedekind eta function, which was checked to be true to a good approximation.
	As we explain below, the main formula of the present paper~\eqref{eq:OSV-type} can be seen as 
	correcting the approximate result of~\cite{Murthy:2015zzy} to an exact formula.

	\subsection{The quantum entropy of \qbps single-centered black holes}

	In the introduction we mentioned the two pictures---macroscopic and microscopic---of BHs in string theory.
	While we mainly focus on the microscopic picture in the rest of the paper, the origins of our formula came 
	from a macroscopic intuition that we now review.  
	Using ideas of the~AdS$_2$/CFT$_1$ correspondence, a macroscopic supergravity description 
	for the degeneracies of microstates of supersymmetric BHs, called the quantum entropy formalism,  
	was put forward in~\cite{Sen:2008vm}.
	The near-horizon geometry of extremal black holes universally contains an AdS$_2$ factor, and the 
	proposal of~\cite{Sen:2008vm} is that the degeneracies of supersymmetric extremal black holes is a functional integral
	on this~AdS$_2$ space defined as
	\begin{equation}
	\label{eq:dmacro-def}
	d_{\rm macro}(Q,P) \= \Big\langle \exp\Bigl[\,q_I\int_{S^1} A^I\Bigr] \Big\rangle_{{\rm EAdS}_2}^{\rm finite} \, .
	\end{equation}
	Here the brackets indicate that one should compute the expectation value of the Wilson line around the 
	Euclidean time circle~$S^1$ for the~$U(1)$ gauge fields~$A^I$ under which the black hole is charged, $q_I$ 
	denotes the corresponding charges of the BH, 
	and the superscript ``finite'' indicates a particular infra-red regularization scheme to deal with the infinite volume 
	of the EAdS$_2$ factor in the near-horizon geometry (see \cite{Sen:2008vm} for details).
	
	To compute the path integral~\eqref{eq:dmacro-def} beyond the leading large-charge approximation, 
	powerful techniques of supersymmetric localization have been employed starting with the work of~\cite{Dabholkar:2010uh}. 
	Localization has been an invaluable tool in the study of partition functions in gauge theories and in many cases has 
	allowed us to reduce a complicated path integral to a much simpler \emph{finite} dimensional integral. A complete 
	review falls outside the scope of the present paper,\footnote{See \cite{Pestun:2016zxk} for an introduction and 
		reviews in the context of field theory.} 
	so we will simply give the final result obtained 
	in~\cite{Dabholkar:2011ec,Gupta:2012cy,Murthy:2013xpa,Murthy:2015yfa,Gupta:2015gga,deWit:2018dix,Jeon:2018kec} 
	for~\eqref{eq:dmacro-def} after localization and for the class of black holes we are interested in.
	For \hbps black hole solutions of an~$\mathcal{N}=2$ supergravity theory with holomorphic prepotential~$F$, 
	we have
	\begin{equation}
	\label{eq:dmacro-loc}
	d_{\rm macro}(Q,P) \= \int_{\mathcal{M}_{\cal Q}}\prod_{I=0}^{n_V}\mathrm{d}\phi^I\,\mu(\phi^I)\,\exp\Bigl[4\pi\,{\rm Im}[F(p^I,\phi^I)] - \pi\,q_I\phi^I\Bigr]\,\bigl(\chi_{\rm V}(p^I,\phi^I)\bigr)^{2 - \frac{1}{12}(n_V + 1)} \, .
	\end{equation}
	The definition of the various quantities entering \eqref{eq:dmacro-loc} are as follows. The integral is over 
	the manifold~$\mathcal{M}_{\cal Q}$, which is characterized by the bosonic field configurations that are 
	supersymmetric with respect to a specific supercharge~${\cal Q}$ preserved by the black hole solution. 
	This manifold is~$(n_V + 1)$-dimensional, where~$n_V$ is the number of abelian vector multiplets under 
	which the black hole is charged, and~$\phi^I$ denote the coordinates on~$\mathcal{M}_{\cal Q}$. 
	The integrand is completely specified by the prepotential~$F(p^I,\phi^I)$ of the theory, which is a 
	homogeneous holomorphic function of its arguments. The associated K\"{a}hler potential~$\chi_{\rm V}(p^I,\phi^I)$ 
	is built out of this prepotential. Finally, we have denoted by~$\mu(\phi^I)$ the measure on $\mathcal{M}_{\cal Q}$,
	which was not obtained from first principles in the above references, but constrained to be a function that contributes~$O(1)$ growth to the 
	entropy when all the charges are scaled to be large.
	
	To apply the formula~\eqref{eq:dmacro-loc} to the~\qbps single-centered black hole solution of the~$\mathcal{N}=4$ theory 
	discussed in Section \ref{sec:exact-entropy}, one consistently truncates the latter theory to an~$\mathcal{N}=2$ theory 
	with~$n_V = 23$ multiplets 
	and prepotential \cite{Shih:2005he}
	\begin{equation}
	\label{eq:prepot-N4}
	F(p^I,\phi^I) \= -\frac{X^1}{X^0}\,X^a\,C_{ab}\,X^b + \frac{1}{2\mathrm{i}\pi}\log\Bigl[\eta^{24}\Bigl(\frac{X^1}{X^0}\Bigr)\Bigr] \, , 
	\quad {\rm with} \quad X^I = \phi^I + \mathrm{i}\,p^I \, ,
	\end{equation}
	where~$C_{ab}$ is the intersection matrix on the middle homology of the internal K3 manifold and~$a,b=2, \ldots, 23$. 
	Using this data and assuming a certain measure on~$\mathcal{M}_{\cal Q}$,~\cite{Murthy:2015zzy,Gomes:2015xcf} 
	showed that the finite dimensional integral~\eqref{eq:dmacro-loc} could be put in the form of a sum of~$I$-Bessel functions 
	indicative of a Rademacher-type expansion for the macroscopic degeneracies of single-centered~\qbps black holes, 
	similar to the exact microscopic formula~\eqref{eq:app-mixed-mock-coeffs}. 
	The above assumption about the measure was essentially a statement of consistency with a certain way of expanding 
	the microscopic formula~\eqref{eq:DVV}, as we now explain.

	\subsection{The measure and Rational Quadratic Divisors}

	The~$z$-integral in the microscopic degeneracy formula~\eqref{eq:DVV} can be performed by calculating 
	residues at the so-called Rational Quadratic Divisors (RQDs) of the Igusa cusp form~$\Phi_{10}$. 
	The leading contribution to the integral comes from the RQD located at~$z=0$~\cite{Dijkgraaf1997d}. 
	Near~$z=0$, the Igusa cusp form behaves as
	\begin{equation}
	\Phi_{10}\tsz \= 4\pi^2\,(2z - \tau - \sigma)^{10}\,z^2\,\eta(\tau)^{24} \,\eta(\sigma)^{24} + \mathcal{O}(z^4) \, .
	\end{equation}
	The remaining integral in~$\tau$ and~$\sigma$ can then be expressed as~\cite{David2006a}
	\begin{equation}
	d_{\frac14}(Q,P) \; \simeq \; (-1)^{\ell + 1}\,\int_{\mathcal{C}_2}\,\frac{\mathrm{d}^2\tau}{\tau_2{}^2}\,e^{-\mathcal{F}(\tau_1,\tau_2)} \, ,
	\end{equation}
	where $\simeq$ indicates that there are subleading contributions coming from other RQDs 
	(additional poles in~$\Phi_{10}^{-1}$), and~$\tau = \tau_1 + \mathrm{i}\tau_2$. 
	The function~$\mathcal{F}(\tau_1,\tau_2)$ is given by
	\begin{align}
	\mathcal{F}(\tau_1,\tau_2) =&\; -\frac{\pi}{\tau_2}\left(n-\ell\tau_{1}+m(\tau_1^2 + \tau_2^2)\right) 
	+ \ln \eta^{24}(\tau_1 + \mathrm{i}\tau_2) + \ln \eta^{24}(-\tau_1+\mathrm{i}\tau_2) + 12\ln(2\tau_2) \nonumber \\
	&\; - \ln\left[\frac{1}{4\pi}\left\{26+\frac{2\pi}{\tau_2}(n-\ell \tau_{1}+m(\tau_1^2 + \tau_2^2))\right\}\right] \, ,
	\end{align}
	and the contour of integration~$\mathcal{C}_2$ is required to pass through the saddle-point 
	of~$\mathcal{F}(\tau_1,\tau_2)$. 
	This way of manipulating the microscopic degeneracy formula corresponds, in physics,
	to calculating the degeneracies of BHs whose magnetic as well as electric charges grow
	at the same rate. In contrast, the expansion studied in Section~\ref{sec:exact-entropy} 
	following~\cite{Dabholkar2012} corresponds to fixing the magnetic charges and letting the electric 
	charges grow (see Equation~\eqref{eq:Fourier-Jacobi}).

	Adding a total derivative term and comparing to the macroscopic 
	localized integral~\eqref{eq:dmacro-loc}, the authors of \cite{Murthy:2015zzy, Gomes:2015xcf} 
	concluded that the measure factor, corresponding to the leading RQD of $\Phi_{10}$ located at~$z=0$, 
	should take the form
	\begin{equation}
	\label{eq:loc-measure}
	\mu(\phi^I) \= m + E_2\Bigl(\frac{\phi^1 + \mathrm{i}p^1}{\phi^0}\Bigr) + E_2\Bigl(-\frac{\phi^1 - \mathrm{i}p^1}{\phi^0}\Bigr) \, ,
	\end{equation}
	where~$E_2$ is the Eisenstein series of weight 2, related to the Dedekind eta function as
	\begin{equation}
	E_2(\tau) \= \frac{1}{2\pi \i}\,\frac{\mathrm{d}}{\mathrm{d}\tau}\log\eta(\tau)^{24} \, .
	\end{equation}
	
	Using the measure~\eqref{eq:loc-measure} in the integral~\eqref{eq:dmacro-loc} leads to an infinite sum of~$I$-Bessel functions 
	coming from integrating term-by-term the series expansions of the prepotential and the measure~\cite{Murthy:2015zzy}. 
	It was noticed in that paper that this infinite sum begins with terms that become smaller up to a point, 
	but that the integrals start diverging after a while. 
	This behavior is characteristic of an asymptotic series, which prompted~\cite{Murthy:2015zzy} 
	to truncate the sum after a finite number of terms. This was achieved using a contour prescription 
	given in~\cite{Gomes:2015xcf}. 
	The end result, after evaluating the integrals, was then
	\begin{align}
	\begin{split}
	\label{eq:dmacro-approx}
	d_{\rm macro}(Q,P) \; \simeq \; &\; 2\pi\sum_{\substack{0\,\leq\,\tilde{\ell}\,\leq\,m \\[.5mm] \widetilde{\Delta} < 0}}\,(\tilde{\ell}
	-2\,\tilde{n})\,d(m+\tilde{n}-\tilde{\ell})\,d(\tilde{n})\,\frac{\cos{\bigl(\pi(m-\tilde{\ell})\ell/m\bigr)}}{\sqrt{m}}\; \times \\
	&\qquad\qquad\quad \times \,\Bigl(\frac{|\widetilde{\Delta}|}{\Delta}\Bigr)^{23/4}\,
	I_{23/2}\Bigl(\frac{\pi}{m}\sqrt{|\widetilde{\Delta}|\Delta}\Bigr) \, ,
	\end{split}
	\end{align}
	where~$d(n)$ is the~$n^{\mathrm{th}}$ Fourier coefficient of the Dedekind eta function as given 
	in~\eqref{eq:Fourier-eta},~$\Delta$ is the usual discriminant~$4mn - \ell^2$, and the~$I$-Bessel function 
	is defined in~\eqref{eq:app-Bessel}. Comparing the above macroscopic result to the Fourier 
	coefficients~\eqref{eq:app-Rad} and~\eqref{eq:app-Klooster}, we see that~$d_{\mathrm{macro}}$ is the 
	first~($k=1$) term in the Rademacher expansion for a Jacobi form of weight~$-10$ upon making the identification
	\begin{equation}
	\label{eq:sugra-wrong-polar}
	c_m(\tilde{n},\tilde{\ell}) \= (\tilde{\ell}-2\,\tilde{n})\,d(m+\tilde{n}-\tilde{\ell})\,d(\tilde{n}) \quad {\rm for} 
	\quad \widetilde{\Delta} = 4m\tilde{n}-\tilde{\ell}{\,}^2 < 0 \, .
	\end{equation}
	This proposal, motivated by the exact computation of a supergravity path integral, already offered a very good 
	numerical agreement with the microscopic data and hinted at an intricate relationship between the Fourier 
	coefficients of a simple modular form (the Dedekind eta function) and those of the more complicated 
	mock Jacobi forms~$\psi_m^{\rm F}$. However, detailed numerical investigations also showed that the 
	formula~\eqref{eq:sugra-wrong-polar} cannot be the complete answer, as evidenced by the small 
	discrepancies between the left- and right-hand sides highlighted in the tables of \cite{Murthy:2015zzy}. 
	
	Since the derivation reviewed above relied on the approximations related to the asymptotic nature of the series, 
	it was already clear that~\eqref{eq:dmacro-approx} is just the beginning of the complete formula. In the 
	rest of the paper, we obtain the correct and exact relationship between the polar terms of~$\psi_m^{\rm F}$ 
	and the Fourier coefficients of~$\eta(\tau)^{-24}$ based on a precise analysis of negative discriminant states 
	in~$\mathcal{N}=4$ string theory, as summarized in our main formula~\eqref{eq:OSV-type}. 
	Therefore, the results of the present paper can be interpreted as giving us the precise way to take 
	into account the subleading RQDs that correct the measure~\eqref{eq:loc-measure}, and truncate 
	the infinite sum of Bessel functions arising from~\eqref{eq:dmacro-loc}.\footnote{The idea of 
		summing up the contributions from all the RQDs of~$\Phi_{10}$ to obtain the exact degeneracy of 
		the dyonic BH was put forward in~\cite{Murthy:2009dq}, but the lack of good technology at the time 
		also led to divergent sums.} 
	\footnote{The conclusions of this paper do not mean that there is not another way to obtain the 
		exact single-centered BH degeneracies after resummation of the residues of the RQDs in a manner 
		consistent with the~$Sp(2,\IZ)$ symmetry of~$\Phi_{10}$. 
		We note, however, that such an enterprise would involve some notion of a ``mock" Siegel form
		that has not been made precise in the mathematical literature to the best of our knowledge.}

	\section{Negative discriminant states and walls of marginal stability}
	\label{sec:neg-disc}

	Now we turn back to our main goal, which is to obtain 
	an analytic formula for the degeneracies of negative discriminant \qbps states $\widetilde{c}_m(n,\ell)$ 
	as defined in \eqref{eq:polar-def} in terms of the coefficients of the Dedekind eta function. 
	In this section we set up the problem in a convenient form after reviewing some facts about negative 
	discriminant states and walls of marginal stability associated to negative discriminant state decays. 
	As reviewed in Section \ref{sec:exact-entropy}, we are interested in counting
	the number of negative discriminant states in the region~$\mathcal{R}$,
	which correspond to bound states of two~\hbps states. 
	Following~\cite{Sen2011}, a convenient way to do so is to count how bound states appear or 
	decay as we move around the moduli space parameterized by~$\Sigma$ discussed around~\eqref{eq:defZSigma}.
	When we cross a wall of marginal stability, bound states appear or decay and contribute to the degeneracies 
	of all negative discriminant states contained in~$\psi^\mathrm{F}_m$. We now review the structure of these 
	walls of marginal stability, referring the reader to~\cite{Sen2007,Sen2011} for more details.
	
	\subsection{Walls of marginal stability: Notation}
	\label{sec:walls}
	
	\begin{enumerate}
		\item In the $\Sigma$ upper half-plane, the walls of marginal stability are of two types \cite{Sen2007}: 
		
		\begin{enumerate}
			\item Semi-circles connecting two rational points $p/r$ and $q/s$ such that $ps-qr=1$. 
			We denote these walls as \textit{S-walls}.
			
			\item Straight lines connecting $\mathrm{i}\infty$ to an integer. These can be thought of as special cases of  
			the above expressions when~$r=0$ and~$p=s=1$, or when~$s=0$ and~$q = -r = 1$. We denote these walls as \textit{T-walls}.
		\end{enumerate}
		
		\begin{figure}[h]
			\centering
			\includegraphics[width=.7\textwidth]{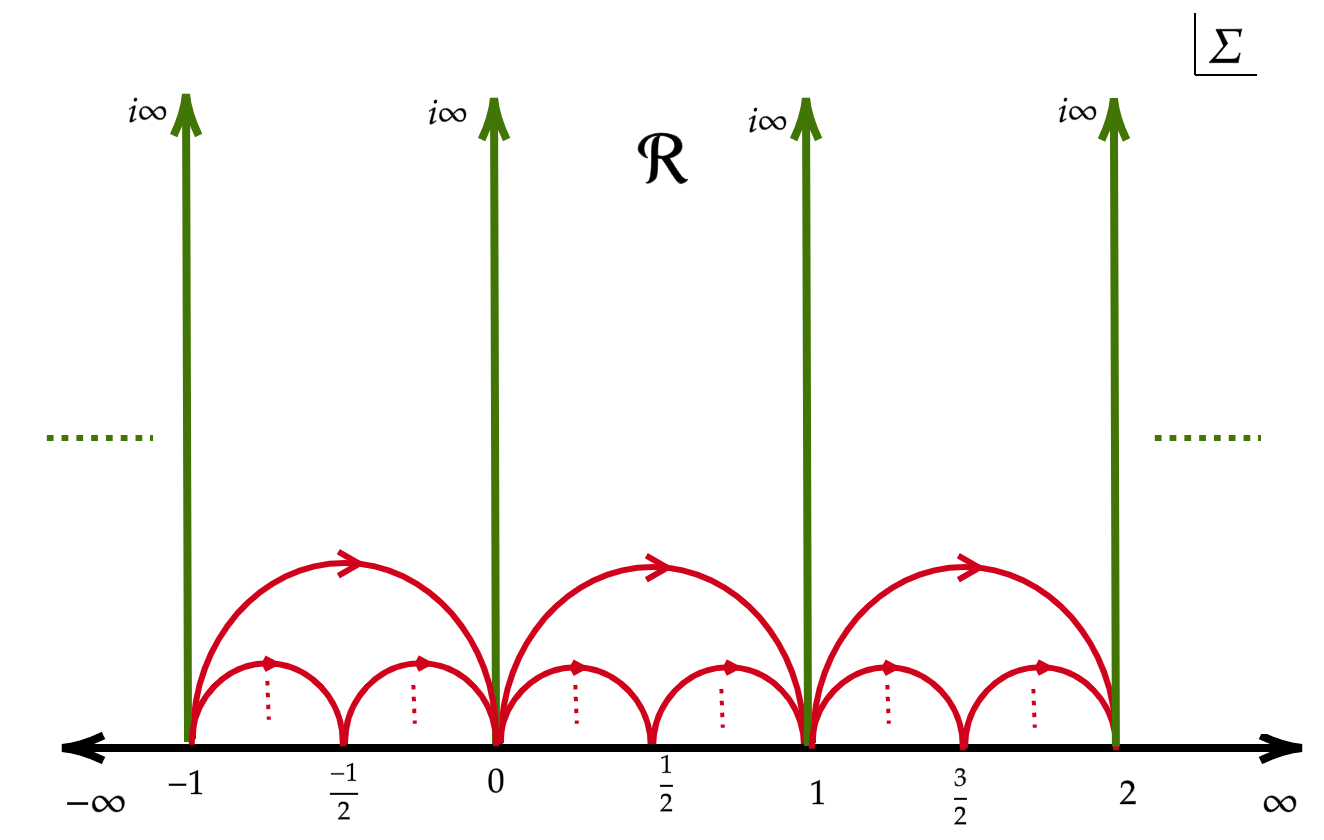}
			\caption{ Structure of T-walls (green) and S-walls (red) in the upper half-plane.}
			\label{fig:twalls}
		\end{figure}
		To any T- or S-wall we associate the following matrix, 
		\begin{align}
		\label{eq:wall}
		\gamma \= \wmat{p}{q}{r}{s} \in PSL(2,\mathbb{Z}) \, .
		\end{align}
		
		\item Given an initial charge vector $ \displaystyle (n,\ell,m) =  (Q^2/2,\,Q \cdot P,\,P^2/2)$, there is an associated charge 
		breakdown at a wall $ \displaystyle \gamma  $ of the form \eqref{eq:wall}, given by
		\begin{align}
		\label{eq:chargebreak}
		\cvec{Q}{P} \longrightarrow \cvec{p (sQ - qP)}{r (sQ -qP)} + \cvec{q(-rQ + pP)}{s(-rQ + pP)} \, ,
		\end{align}
		and which corresponds to a \qbps BH decaying into two \hbps centers. The charges of the two centers are 
		given by~$\gamma \cdot \cvec{Q_\gamma}{0}$ and~$\gamma \cdot \cvec{0}{P_\gamma}$, with 
		\begin{equation}
		\cvec{Q_\gamma}{P_\gamma} = \gamma^{-1} \cdot \cvec{Q}{P} \, , 
		\end{equation}
		which shows that after the breakdown one center is purely electric while the other is purely magnetic in the new frame. 
		We define $(n_\gamma, \ell_\gamma, m_\gamma) = (Q_\gamma^2/2,\,Q_\gamma\cdot P_\gamma,\,P_\gamma^2/2)$, 
		which are given explicitly by 
		\begin{align}
		\label{eq:transformedcharges}
		\begin{split}
		\nw &\=  s^2 n +  q^2 m - s q \ell \, , \\
		\lw &\= - 2s r n - 2 p q m + \ell( ps + qr) \, , \\
		\mw &\= r^2 n + p^2 m - p r \ell \, .
		\end{split}
		\end{align}
		
		\item The set of matrices that characterize the walls in the $\Sigma$ upper half-plane can be divided into subsets that satisfy the following properties:
		\begin{align}
		\label{eq:wallcategory}
		\begin{split}
		&\Gamma_S^+ \defeq \left \lbrace \gamma = \wmat{p}{q}{r}{s}\in PSL(2,\mathbb{Z}) \ \Big \vert \  r > 0, \, s > 0   \right \rbrace \, , \\
		&\Gamma_S^- \defeq \left \lbrace \gamma = \wmat{p}{q}{r}{s}\in PSL(2,\mathbb{Z}) \ \Big \vert \ r > 0, \, s < 0   \right \rbrace \, , \\
		& \Gamma_T \defeq \left \lbrace \gamma = \wmat{p}{q}{r}{s}\in PSL(2,\mathbb{Z}) \ \Big \vert \  r s = 0 \right \rbrace \, .
		\end{split}
		\end{align}
		Because the above matrices have unit determinant, the walls in $\Gamma_S^+$ have $p/r > q/s$, and the walls in $\Gamma_S^-$ have $p/r < q/s$.
		We denote by~$\Gamma_S = \Gamma_S^-  \cup \Gamma_S^+$ the full set of S-walls. Notice that~$PSL(2,\mathbb{Z}) = \Gamma_S \cup \Gamma_T$.
		
		\item We define the \emph{orientation} of a wall $\gamma$ to be $ \displaystyle q/s \rightarrow p/r $. With respect to this orientation,
		a bound state of \hbps states exists in the chamber to the right of the wall if $ \displaystyle \ell_\gamma < 0 $, and in the chamber to the 
		left of the wall if $\ell_\gamma > 0$ \cite{Sen2011}.
		
		\item The attractor region $ \displaystyle \mathcal{R} $ in \eqref{eq:R-region} is the region of the $\Sigma$ upper 
		half-plane bounded by 
		the T-walls $0 \rightarrow \mathrm{i}\infty$, $1\rightarrow \mathrm{i}\infty$ and the semi-circular S-wall $0 \rightarrow 1$. 
		(Note that, because of the negative sign in the contour~\eqref{eq:moduli-contour}, the 
		attractor region~$-\text{Im}(\tau)<\text{Im}(z)\le 0$ as in~\eqref{eq:R-region} maps to~$0\le\text{Re}(\Sigma)<1$.) 
		We will be 
		interested in the degeneracies of negative discriminant states in this region, as reviewed in Section~\ref{sec:exact-entropy}.
	\end{enumerate}
	
	From Point 4 combined with~\eqref{eq:transformedcharges}, it is clear that none of the T-walls contribute in the region~$\mathcal{R}$. For example, when~$r=0$,~$\gamma = \begin{psmallmatrix} 1 & q \\[.2mm] 0 & 1 \end{psmallmatrix}$ which means~$(n_\gamma,\ell_\gamma,m_\gamma) = (n+q^2 m - q\ell,\ell - 2qm,m)$. Recalling that we can restrict ourselves to~$0 \leq \ell <2m$, this shows that when $q \geq 0$, the above T-walls contribute to the right of the region~$\mathcal{R}$. This is consistent with our analysis of Section~\ref{sec:exact-entropy}, where we showed that~$\psi_m^{\mathrm{P}}$, which captures all the T-walls, actually has vanishing Fourier coefficients in the region~$\mathcal{R}$.
	
	For the S-walls, there exists a map between the sets $\Gamma_S^+$ and $\Gamma_S^-$, given by the right multiplication of an element of $\Gamma_S^+$ by the matrix 
	\begin{equation}
	\label{eq:Strafo}
	\tilde S \= \wmat{0}{-1}{1}{0} \, .
	\end{equation} 
	This map reverses the orientation of the wall and flips the sign of $\ell_\gamma$. 
	Furthermore,~$\tilde S$ squares to~$-I$, which means that it is an involution in~$PSL(2,\IZ)$.
	Therefore we can focus only on elements of $\Gamma_S^+$ when discussing the details 
	of negative discriminant states breakdowns across walls of marginal stability.

	\subsection{Towards a formula for black hole degeneracies}
	\label{sec:towards}

	Upon crossing a wall of marginal stability, the index jumps by an amount controlled by the generating function of each of the associated \hbps centers. The latter is given by the inverse of $\eta(\tau)^{24}$, whose Fourier coefficients are given by the partition function
	into~$24$ colors~$p_{24}(n)$ (cf.~Equation~\eqref{eq:Fourier-eta}). 
	Summing up all possible decays across the S-walls leads to the following counting 
	formula for negative discriminant states 
	living in the region $\mathcal{R}$:
	\begin{equation}
	\label{eq:deg1}
	\frac12 \sum_{\substack{\gamma\in\Gamma_S }}\,
	(-1)^{\ell_\gamma+1}\,\theta(\gamma,\mathcal{R})\,\vert\ell_\gamma\vert\,d(m_\gamma)\,d(n_\gamma) \, ,
	\end{equation}
	where
	the function $\theta(\gamma,\mathcal{R})$ is a step-function giving 1 if the bound state exists on the same side 
	of the wall~$\gamma$ bounding $\mathcal{R}$ and 0 otherwise. Formally, it is defined as follows
	\begin{align}
	\label{eq:thetastep}
	\theta(\gamma, \mathcal{R}) \= \Bigg \vert  \frac{\mathcal{O}(\gamma, \mathcal{R}) + \sgn(\ell_\gamma)}{2}\Bigg \vert \, , 
	\qquad \mathcal{O}(\gamma, \mathcal{R}) \= 
	\begin{cases}
	+1, \ \ \gamma \in \Gamma_S^+\\
	-1, \ \ \gamma \in \Gamma_S^- 
	\end{cases} \, .
	\end{align}
	On one hand this sum can be written in a more covariant manner by extending it to a sum over all matrices in~$PSL(2,\IZ)$,
	\begin{equation}
	\label{eq:deg1T}
	\frac12 \sum_{\substack{\gamma\in PSL(2,\IZ) }}\,
	(-1)^{\ell_\gamma+1}\,\theta(\gamma,\mathcal{R})\,\vert\ell_\gamma\vert\,d(m_\gamma)\,d(n_\gamma) \, ,
	\end{equation}
	by extending the~$\theta$ function to all of~$PSL(2,\IZ)$ via 
	\be \label{eq:thetastepR}
	\theta(\gamma, \mathcal{R})\=0 \,, \qquad \gamma \in \Gamma_T \,,
	\ee
	because the T-walls do not contribute in the region~$\mathcal{R}$ as we saw above.
	On the other hand, the sum~\eqref{eq:thetastep} can also be written as a sum over a smaller set as follows. Note that 
	the summand in equation \eqref{eq:deg1} is invariant under a transformation by the matrix~$\tilde S$ 
	given in equation \eqref{eq:Strafo} 
	because~$n_{\gamma\tilde S} = m_\gamma$, $m_{\gamma\tilde S } = n_\gamma$, 
	$\ell_{\gamma\tilde S} = -\ell_\gamma$ and~$\tilde S$ 
	exchanges~$\Gamma_S^+$ and~$\Gamma_S^-$. 
	This means that the contributions from the sum over~$\gamma \in \Gamma_S^+$ and~$\gamma \in \Gamma_S^-$ 
	are equal and we can sum over $\Gamma_S^+$ only. So, we can alternatively write~\eqref{eq:deg1} as
	\begin{equation}
	\label{eq:deg}
	\sum_{\substack{\gamma\in\Gamma_S^+ }}\,(-1)^{\ell_\gamma+1}\,
	\bigg \vert  \frac{1 + \sgn(\ell_\gamma)}{2}\bigg \vert\; \vert\ell_\gamma\vert\;d(m_\gamma)\;d(n_\gamma) \, .
	\end{equation}

	\subsection{A subtlety from bound state metamorphosis}

	While accounting for all the negative discriminant states in the region~$\mathcal{R}$, there is a further subtlety that needs to be taken 
	into account due to a phenomenon 
	known as bound state metamorphosis (BSM)~\cite{Andriyash:2010yf,Chowdhury:2012jq, Sen2011}. BSM stems from the fact that when one or 
	both \hbps centers making up a \qbps bound state carry the lowest possible charge invariant (that is, 
	when~$n_\gamma = -1$,~$m_\gamma = -1$ or~$n_\gamma = m_\gamma = -1$ for a given wall~$\gamma$), two or more bound states 
	must be identified following a precise set of rules to avoid overcounting in the index~\eqref{eq:deg}. 
	Thus we can write the set of all contributing walls as the quotient 
	\be\label{defgset}
	\gbsm \= PSL(2,\IZ)/\text{BSM} \,,
	\ee
	and write the polar degeneracies~\eqref{eq:polar-def} as
	\begin{equation}
	\label{eq:deg1again}
	\widetilde{c}_m(n,\ell)  \= \frac12 \sum_{\substack{\gamma\in \gbsm }}\,
	(-1)^{\ell_\gamma+1}\,\Theta(\gamma)\,\vert\ell_\gamma\vert\,d(m_\gamma)\,d(n_\gamma) \,.
	\end{equation}
	Here we have to introduce a new function~$\Theta(\gamma)$ which generalizes the 
	function~$\theta(\gamma,\mathcal{R})$ defined above to take into account the phenomenon of BSM so that it 
	is defined on the coset~$\gbsm$. We will be in a position to give a proper definition after a discussion of BSM
	in the following sections.
	We can also present this formula as a sum over the set~$\Gamma_S$ or~$\Gamma_S^+$ 
	modulo the identifications due to BSM for the reasons discussed above (T-walls do not contribute in~$\mathcal{R}$, 
	and~$\tilde S$ gives a map between~$\Gamma^-_S$ and~$\Gamma_S^+$): 
	\begin{equation}
	\label{eq:degagain}
	\widetilde{c}_m(n,\ell)  \= \sum_{\substack{\gamma\in\Gamma_S^+/\text{BSM} }}\,(-1)^{\ell_\gamma+1}\,
	\Theta(\gamma) \, \vert\ell_\gamma\vert\;d(m_\gamma)\;d(n_\gamma) \, .
	\end{equation}
	
	Given that the left-hand side of this formula is finite, it is reasonable to expect that given a set of 
	initial charges $ \displaystyle (n,\ell,m) $, only a finite subset of walls of marginal stability 
	gives a non-zero contribution to the above sums. 
	This expectation turns out to be correct and we can write the final formula as a sum over the~\emph{finite} set~$\gset$ 
	\begin{equation}
	\label{eq:degfinite}
	\widetilde{c}_m(n,\ell)  \= \sum_{\substack{\gamma\in\gset}}\,(-1)^{\ell_\gamma+1}\,
	\vert\ell_\gamma\vert\;d(m_\gamma)\;d(n_\gamma) \, .
	\end{equation}
	Our goal in the following sections is to now fully characterize the 
	subset $\gset$ and show that it contains a finite number of elements for a given charge vector~$(n,\ell,m)$.

	It will be convenient to split the characterization of the set $\gset$ depending on whether BSM does not or does occur. 
	This will be the subject of sections \ref{sec:negwomet} and \ref{sec:metamorphosis}, respectively.
	Before initiating the study of the finiteness of~$\gset$, we recall from the discussion below~\eqref{eq:DMZ-contour} 
	that we can restrict the charge vector to be such that $0 \leq \ell < 2m$. 
	In addition, $\psi_m^\mathrm{F}$ has even weight so there is a reflection symmetry $\ell \rightarrow - \ell$ which allows us to restrict ourselves to the case\footnote{Note that even though $\psi_m^\mathrm{F}$ is not modular but only mock modular, both its completion $\widehat{\psi}_m^\mathrm{F}$ and its shadow, and therefore $\psi_m^\mathrm{F}$ itself, enjoy this $\ell \rightarrow -\ell$ symmetry \cite{Dabholkar2012}.} 
	$\ell \in \{0,\ldots,m\}$. 
	The index $m$ runs from $-1$ to $+\infty$ in the expansion \eqref{eq:Fourier-Jacobi}. 
	For~$m=-1$ and~$m=0$, there is no macroscopic BH as explained in Section \ref{sec:exact-entropy}, and therefore 
	we only study~$m>0$ in the following. 
	Thus our goal is to study the set~$\gset$ of walls of marginal stability 
	for a charge vector~$(n,\ell,m)$ that satisfies~$m>0$,~$n\geq -1$ and~$0\leq \ell\leq m$.

	\section{Negative discriminant states without metamorphosis}
	\label{sec:negwomet}
	
	In this section, we begin to characterize~$\gset$. 
	For the time being we ignore the phenomenon of BSM (which will be the subject of the next section), 
	and show that the contribution to~$\gset$ in this case is finite.  
	Accordingly, in order to identify the walls that contribute to the polar degeneracies~$\widetilde{c}_m(n,\ell)$, 
	we study the system of 
	inequalities~$m_\gamma \geq 0$ and~$n_\gamma \geq 0$ defined in~\eqref{eq:transformedcharges} 
	for a given charge vector~$(n,\ell,m)$ 
	such that~$ \Delta = 4mn - \ell^2 < 0$ and~$0 \leq \ell \leq m$.
	As explained above, we focus on walls in~$\Gamma_S^+ \subset PSL(2,\IZ)$, 
	which allows us to choose~$r,s>0$ in the following. 
	The condition~$m_\gamma \geq 0$ then amounts to
	\begin{equation}
	m_\gamma = m\,\Bigl(\frac{p}{r}\Bigr)^2 - \ell\,\Bigl(\frac{p}{r}\Bigr) + n \geq 0 \, .
	\end{equation}
	The first equality defines a parabola in the $(p/r,y=m_\gamma)-$plane and the condition $m>0$ means that the inequality has two branches:
	\begin{equation}
	\frac{p}{r} \geq \frac{\ell + \sqrt{|\Delta|}}{2m} \quad \text{or} \quad \frac{p}{r} \leq \frac{\ell - \sqrt{|\Delta|}}{2m} \, .
	\end{equation}
	We will call these positive and negative ``runaway branches'' since $p/r$ is unbounded from above or from below, respectively. The condition $n_\gamma \geq 0$ amounts to
	\begin{equation}
	\label{eq:ng-pos}
	n_\gamma = m\,\Bigl(\frac{q}{s}\Bigr)^2 - \ell\,\Bigl(\frac{q}{s}\Bigr) + n \geq 0 \, .
	\end{equation}
	Moreover, using that the determinant of $\gamma$ must be equal to one, we have
	\begin{equation}
	\label{eq:unit-det}
	\frac{q}{s} = \frac{p}{r} - \frac{1}{rs} \, ,
	\end{equation}
	and so the first equality in~\eqref{eq:ng-pos} can also be seen as a parabola in the~$(p/r,y=n_\gamma)-$plane, shifted by~$1/(rs)$ compared to the first parabola. The condition $n_\gamma \geq 0$ also has a positive and negative runaway branch,
	\begin{equation}
	\frac{p}{r} \geq \frac{\ell + \sqrt{|\Delta|}}{2m} + \frac{1}{rs} \quad \text{or} \quad \frac{p}{r} \leq \frac{\ell - \sqrt{|\Delta|}}{2m} + \frac{1}{rs} \, .
	\end{equation}
	
	Recall from Section~\ref{sec:neg-disc} that we focus on $\Gamma_S^+$ walls which corresponds to $rs > 0$. We split the argument in two cases. Considering 
	\begin{equation}
	\frac{\ell - \sqrt{|\Delta|}}{2m} + \frac{1}{rs} < \frac{\ell + \sqrt{|\Delta|}}{2m} \, ,
	\end{equation}
	the smaller intercept with the~$p/r$ axis of the shifted~$n_\gamma$-parabola  is smaller than the larger intercept with the~$p/r$ axis of the~$m_\gamma$-parabola. This situation is illustrated in Figure~\ref{fig:runaway}.
	\begin{figure}[h]
		\centering
		\begin{subfigure}[b]{0.8\textwidth}
			\includegraphics[width=\textwidth]{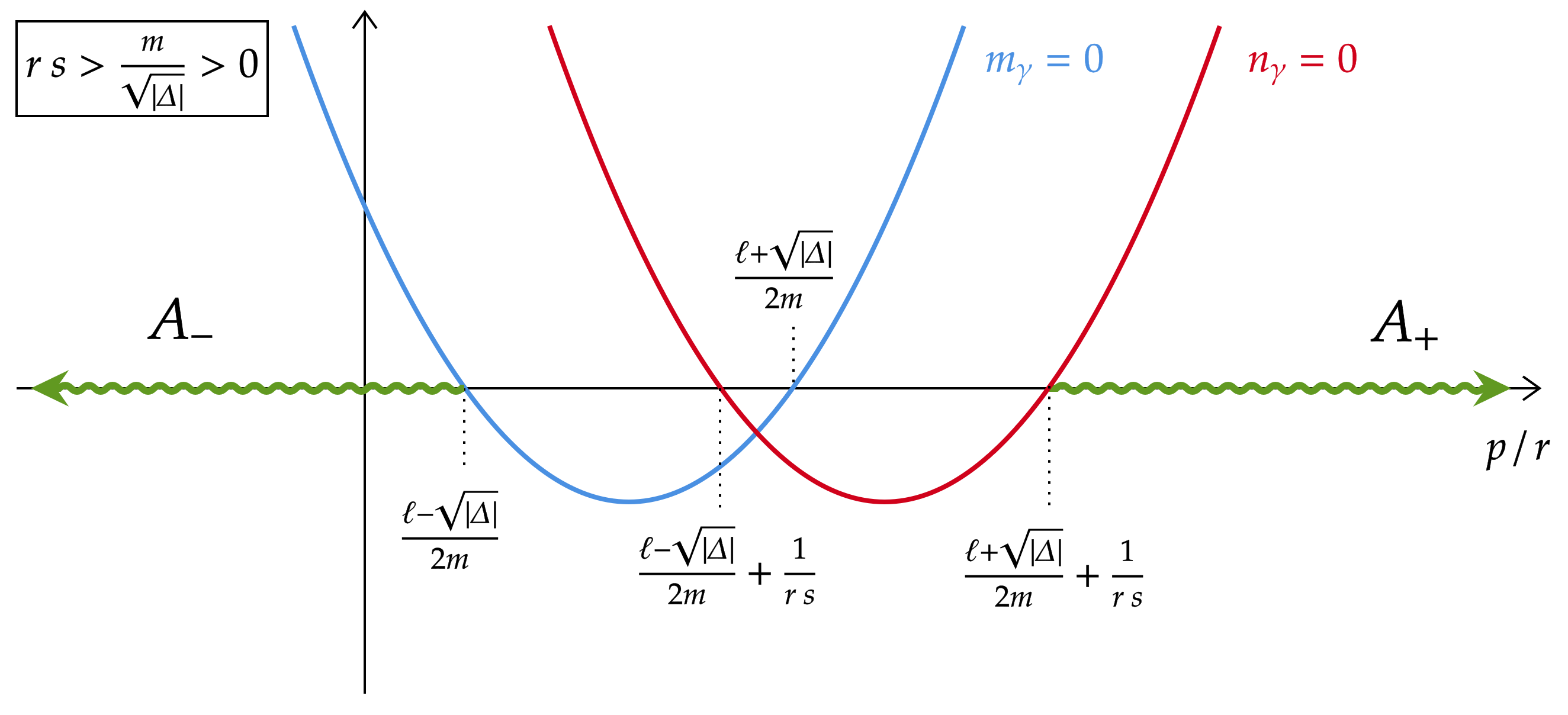}
			\caption{The two runaway branches $A_\pm$ in the case where $ r s > \frac{m}{\sqrt{\Delta}} > 0 $ }
			\label{fig:runaway}
		\end{subfigure}
		\\\vspace{5mm}
		\begin{subfigure}[b]{0.8\textwidth}
			\includegraphics[width=\textwidth]{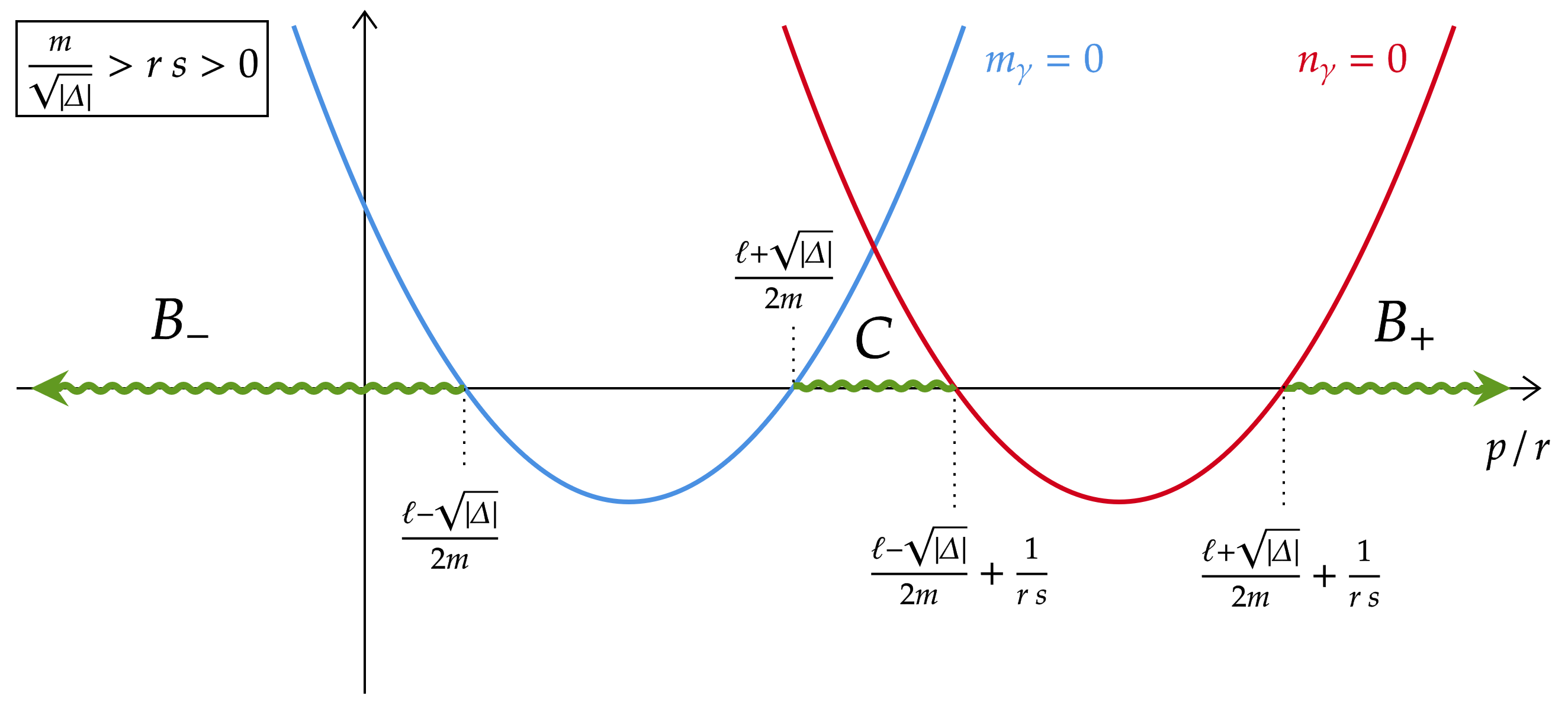}
			\caption{The runaway branches $B_\pm$ and the bounded branch $C$ when $ \frac{m}{\sqrt{\Delta}} > r s > 0 $ }
			\label{fig:bounded}
		\end{subfigure}
		\caption{The regions where $m_\gamma \geq 0$ and $n_\gamma \geq 0$ for $r s > 0$, denoted in green}
	\end{figure}
	In this case, requiring both inequalities~$m_\gamma \geq 0$ and~$n_\gamma \geq 0$ implies that
	\begin{equation}
	\label{eq:Ap}
	\frac{p}{r} \geq \frac{\ell + \sqrt{|\Delta|}}{2m} + \frac{1}{rs} \quad \text{and} \quad r s > \frac{m}{\sqrt{|\Delta|}} > 0 \, ,
	\end{equation}
	on the positive runaway branch which we denote $A_+$, or
	\begin{equation}
	\label{eq:Am}
	\frac{p}{r} \leq \frac{\ell - \sqrt{|\Delta|}}{2m} \quad \text{and} \quad r s > \frac{m}{\sqrt{|\Delta|}} > 0 \, ,
	\end{equation}
	on the negative runaway branch which we denote $A_-$. 
	For a given set of $(n,\ell,m)$, the conditions~\eqref{eq:Ap} have solutions over the integers for $(p,r,s)$. However, if we supplement this system with the condition that $\ell_\gamma > 0$ (since this is a $\Gamma_S^+$ wall this condition is necessary to have a non-zero contribution owing to the~$\Theta$ function in \eqref{eq:degagain}), then there are no 
	solutions for $(p,r,s)$. Indeed these conditions imply that
	\begin{equation}
	\begin{split}
	0 <&\; 2\,\frac{r}{s}\,n_\gamma + \ell_\gamma = \ell - 2\,\frac{q}{s}\,m =  \ell-2\lp \frac{p}{r}-\frac{1}{rs} \rp m \\ 
	\Longrightarrow \;\; \frac{p}{r} <&\; \frac{\ell}{2m} + \frac{1}{rs} \, , 
	\end{split}
	\end{equation}
	which is in contradiction with the first equality of~\eqref{eq:Ap}.
	Similarly, supplementing the branch~$A_-$ 
	by the condition~$\ell_\gamma > 0$, there are no solutions 
	for $(p,r,s)$. This can be seen by showing that the inequality $2 \frac{s}{r} m_\gamma + \ell_\gamma >0$ is in contradiction with the first inequality of~\eqref{eq:Am}. \\
	
	The other case to consider is when
	\begin{equation}
	\frac{\ell - \sqrt{|\Delta|}}{2m} + \frac{1}{rs} \geq \frac{\ell + \sqrt{|\Delta|}}{2m} \, .
	\end{equation}
	This means that the smaller intercept of the shifted parabola is larger than or equal to the larger intercept of the original one, as illustrated in Figure \ref{fig:bounded}. In this case, we still have the usual runaway branches which we call $B_+$ and $B_-$, but in addition a new branch of solutions for $p/r$ opens up, which we call the bounded branch $C$,
	\begin{equation}
	\label{eq:C-branch}
	\frac{\ell + \sqrt{|\Delta|}}{2m} \leq \frac{p}{r} \leq \frac{\ell - \sqrt{|\Delta|}}{2m} + \frac{1}{rs} \quad \text{and} \quad \frac{m}{\sqrt{|\Delta|}} \geq r s > 0 \, .
	\end{equation}
	Once again, adding the condition that $\ell_\gamma > 0$ suffices to show that there are no integer solutions $(p,r,s)$ to the system of inequalities characterizing the runaway branches $B_\pm$. The proof of this is identical to the one above for $A_\pm$. There are now however solutions for the $C$ branch. Observe that on this branch we also have a condition on the original charges: since~$r$ and~$s$ are integers,~$r s \geq 1$ and so the second condition in \eqref{eq:C-branch} demands that $m \geq \sqrt{|\Delta|}$. Including the inequality $\ell_\gamma > 0$, we obtain the following system for potential walls without BSM contributing to the polar coefficients: 
	\begin{equation}
	\label{eq:master}
	\begin{cases} 
	&\frac{\ell + \sqrt{|\Delta|}}{2m} \leq \frac{p}{r} \leq \frac{\ell - \sqrt{|\Delta|}}{2m} + \frac{1}{rs} \\[1mm]
	&\quad \frac{m}{\sqrt{|\Delta|}} \geq r s > 0 \\[1mm]
	&-2n r s - 2m p q + \ell (p s + q r) > 0
	\end{cases} \, .
	\end{equation}
	
	To analyze this system, we start by using the unit determinant condition to eliminate $q$, and then express everything in terms of the variables
	\begin{equation}
	P := p s \, , \quad R := r s \, .
	\end{equation}
	Then \eqref{eq:master} takes the form of a system of inequalities on two variables $(P,R)$:
	\begin{equation}
	\label{eq:master-PR}
	\begin{cases}
	&\frac{\ell + \sqrt{|\Delta|}}{2m} \leq \frac{P}{R} \leq \frac{\ell - \sqrt{|\Delta|}}{2m} + \frac{1}{R} \\[1mm]
	&\quad \frac{m}{\sqrt{|\Delta|}} \geq R > 0 \\[1mm]
	&- 2 n R  - 2 m \frac{P}{R}(P-1) + \ell (2P-1) > 0
	\end{cases} \, .
	\end{equation}
	We can analyze this system on a case-by-case basis, depending on the original charges $(n,\ell,m)$. Recall that we must only consider $4mn-\ell^2 < 0$, $m>0$ and $0 \leq \ell \leq m$.
	
	\begin{enumerate}
		\item
		\textbf{\underline{Case 1}}: $m > 0$, $n=-1$ \\
		In this case, there are no integer solutions to \eqref{eq:master-PR}. We see from equations \eqref{eq:transformedcharges} that $n=-1$ and 
		$n_\gamma \geq0$, $m_\gamma\geq 0$ require that $p,q\neq 0$. The left-hand-side of the first inequality in equation \eqref{eq:master-PR} implies 
		that $P/R>0$, which then implies that $p>0$ since $r,s>0$. The determinant condition $ps-qr=1$ then requires that $q>0$ as well. However, this is a 
		contradiction since the right-hand-side of the first inequality in equation \eqref{eq:master-PR} can be rewritten as
		\be
		\frac{q}{s} = \frac{P}{R}-\frac{1}{R} \leq \frac{\ell - \sqrt{|\Delta|}}{2m} < 0\,.
		\ee
		Here we used that for $m>0, n=-1$ we have $ \sqrt{|\Delta|}=\sqrt{\ell^2+4m}>\ell$.
		
		\item
		\textbf{\underline{Case 2}}: $m > 0$, $n=0$, and $\ell > 0$ \\
		In this case we find solutions given by 
		\begin{equation}
		P = 1 \, , \quad 0< R \leq \frac{m}{\ell} \, .
		\end{equation}
		Translating back to the original $(p,q,r,s)$ variables, this yields matrices of the form
		\begin{equation}
		\label{eq:walls-noBSM-1}
		\begin{pmatrix} 1 & 0 \\ r & 1 \end{pmatrix} \, , \quad \text{with} \quad 0 < r \leq \frac{m}{\ell} \, .
		\end{equation}
		Note that all entries in the above matrix are bounded from above by $m$. As we will see below, such $m$-dependent bounds always arise when considering the set of contributing walls $\gset$.
		
		\item
		\textbf{\underline{Case 3}}: $m > 0$, $n > 0$ and $\ell > 0$ \\
		This case is slightly more involved. First, notice from the left-hand-side of the first inequality in equation \eqref{eq:master-PR} that $P=ps>0$. Therefore we split the discussion depending on whether $P = 1$ or $P > 1$.
		\begin{enumerate}
			\item 
			\textbf{Case 3a: $P = 1$} \\
			In this case the inequalities \eqref{eq:master-PR} with $P=1$ impose
			\begin{equation}
			P = 1 \, , \quad 0 < R \leq \frac{\ell - \sqrt{|\Delta|}}{2n} \, .
			\end{equation}
			In the variables $(p,q,r,s)$, we therefore have a non-zero contribution to 
			the polar coefficients from matrices of the form:
			\begin{equation}
			\label{eq:walls-noBSM-2}
			\begin{pmatrix} 1 & 0 \\ r & 1 \end{pmatrix} \, , \quad \text{with} \quad 0 < r \leq \frac{\ell - \sqrt{|\Delta|}}{2n} \leq \frac{m-1}{2n} \, .
			\end{equation}
			Again, note that all entries in the above matrix are smaller than $m$. 
			
			\item 
			\textbf{Case 3b: $P > 1$} \\
			In this case, the inequalities \eqref{eq:master-PR} yield the following bounds on $P$ and $R$:
			\begin{equation}
			\label{eq:walls-noBSM-3}
			1 < P \leq \frac{1}{2}\Bigl(1+\frac{\ell}{\sqrt{|\Delta|}}\Bigr) \, , \qquad \frac{\ell + \sqrt{|\Delta|}}{2n}\,(P-1) \leq R \leq \frac{\ell - \sqrt{|\Delta|}}{2n}\,P \, .
			\end{equation}
			The corresponding walls do not start at $q/s = 0$ but instead are strictly inside the largest semi-circular S-wall 0 $\rightarrow$ 1. 
			
			Note that we can again get an $m$-dependent upper bound on $P$ by using the fact that $\ell \leq m$ and $\sqrt{|\Delta|}\geq 1$. We can also use this upper bound on $P$ in the upper bound on $R$ directly to obtain
			\be
			P \leq \frac{m+1}{2} \, , \qquad R \leq \frac{m}{ \sqrt{|\Delta|}} \, .
			\ee
			The above implies that all matrix entries of $\gamma \in PSL(2,\mathbb{Z})$ satisfying \eqref{eq:walls-noBSM-3} are bounded from above by $m$.
		\end{enumerate}
	\end{enumerate}
	
	This exhausts all possible cases for contributions without BSM: 
	the conditions~\eqref{eq:walls-noBSM-1},~\eqref{eq:walls-noBSM-2} and~\eqref{eq:walls-noBSM-3} with~$n\geq 0$ and~$\ell,m > 0$ fully characterize the set~$\gset$ in this case. By inspection, this set has a finite number of elements. Observe that all the walls giving a non-zero contribution to~\eqref{eq:degfinite} have 
	entries bounded from above by~$m$.

	\section{Effects of black hole bound state metamorphosis}
	\label{sec:metamorphosis}
	
	We now turn to identifying the walls of marginal stability for which BSM is relevant. We study the problem systematically in three different cases viz., magnetic-, electric-, and dyonic-metamorphosis, corresponding to~$m_\gamma = -1$,~$n_\gamma = -1$ and~$m_\gamma = n_\gamma = -1$, respectively. 
	
	\subsection{Magnetic metamorphosis case: $ \displaystyle m_\gamma = -1, n_\gamma \geq 0 $ }
	\label{sec:mag-BSM}
	
	As in Section~\ref{sec:negwomet}, we start with a charge vector~$(n,\ell,m)$ such that~$4mn-\ell^2 < 0$ and~$0 \leq \ell \leq m$. However, we are now interested in the walls~$\gamma$ for which~$ \displaystyle m_\gamma = -1$ and~$ \displaystyle n_\gamma\geq 0 $. 
	The idea behind magnetic metamorphosis is that when there is a wall~$ \displaystyle \gamma $ such that~$ \displaystyle m_\gamma = -1, $ then there is another wall~$ \displaystyle \t \gamma $ which has the exact same contribution to the index as~$ \displaystyle \gamma $. Furthermore, one needs to implement a precise prescription to properly account for such walls, and avoid overcounting in the polar degeneracies, as shown in~\cite{Sen2011,Chowdhury:2012jq}. 
	It will be useful to explicitly review some details of this phenomenon. To do so, we begin with the following definition:
	\begin{definition}
		\label{def:mag-meta-wall}
		For a wall $ \displaystyle \gamma $ with $ \displaystyle m_\gamma = -1, n_\gamma \geq 0 $, we define its \textit{metamorphic dual} as
		\begin{equation}
		\label{eq:mag-BSM-meta}
		\t \gamma \defeq \gamma \d \wmat{1}{-\ell_\gamma}{0}{1} \, .
		\end{equation}
	\end{definition}
	With this definition, the prescription found in~\cite{Sen2011,Chowdhury:2012jq} to properly account for magnetic-BSM can be summarized as follows (we refer the reader to the references just mentioned for a physical justification of this):
	\begin{quotation}
		\ndt A wall~$\gamma$ at which magnetic-BSM occurs contributes to the polar coefficients~$\widetilde{c}_m(n,\ell)$ 
		if and only if~$\gamma$ and its metamorphic 
		dual~$\tilde{\gamma}$ both contribute in~$\mathcal{R}$. If so, the contributions of~$\gamma$ and~$\tilde{\gamma}$ should be counted only once.
	\end{quotation}
	A necessary condition for the second part of the prescription is that both~$\gamma$ and~$\tilde{\gamma}$ have the same index contribution, as we now review. First, from Definition \ref{def:mag-meta-wall}, it is easy to see that a wall $ \displaystyle \gamma $ and its metamorphic dual $ \displaystyle \t \gamma $ have the same end point $p/r$,
	\begin{align}
	\label{eq:gamma-gammatilde}
	\gamma = \wmat{p}{q}{r}{s} \;\; \Longleftrightarrow \;\; \t \gamma = \wmat{p \ \ }{- p \ell_\gamma + q}{r\ \ }{-r \ell_\gamma + s} \, .
	\end{align} 
	With this, we prove the following statement: 
	\begin{proposition}
		\label{claim:mag-BSM}
		For a given set of charges~$(n,\ell,m)$, the wall~$ \displaystyle \t \gamma $ has the same index contribution as~$ \displaystyle \gamma $ to the 
		polar coefficients~$\widetilde{c}_m(n,\ell)$.
	\end{proposition}
	\begin{proof}
		First consider the wall $ \displaystyle \gamma $ for which the electric and magentic centers are
		\begin{align}
		Q_\gamma = sQ- q P \, , \quad P_\gamma = -rQ + pP \, .
		\end{align}
		Thus, we have
		\begin{align}
		\ell_\gamma = Q_\gamma \d P_\gamma = -sr Q^2 - qp P^2 + (sp + qr)\,Q\d P.
		\end{align}
		A similar calculation for $ \displaystyle \t \gamma $ shows that $Q_{\t \gamma} = Q_\gamma + \ell_\gamma P_\gamma$ and $P_{\t \gamma} = P_\gamma$, as well as
		\begin{align}
		\label{eq:ell-minusell}
		\ell_{\t \gamma} = \left ( Q_\gamma + \ell_\gamma P_\gamma \right ) \d P_\gamma = \ell_\gamma +  \ell_\gamma P_\gamma^2 = - \ell_\gamma \, ,
		\end{align}
		where 
		in the last equality 
		we have made use of the fact that $ \displaystyle m_\gamma = P_\gamma^2/2 = -1 $. 
		We also note that the above considerations imply
		\begin{align}
		\label{eq:qpeqq2}
		m_\gamma = m_{\t \gamma} = -1 \, , \quad n_{\t \gamma} = \left (n_\gamma + \ell_\gamma^2 m_\gamma +  \ell_\gamma^2 \right )   = n_\gamma \, .
		\end{align}
		From~\eqref{eq:ell-minusell},~\eqref{eq:qpeqq2} and the fact that a given bound state of charges~$(n',\ell',m')$ has indexed degeneracy~$(-1)^{\ell' + 1}|\ell'|d(n')d(m')$, we conclude that a wall~$\gamma$ and its metamorphic image~$\t \gamma$ contribute equally to the negative discriminant degeneracies.
	\end{proof}
	
	We now illustrate the potential non-zero contributions to the formula~\eqref{eq:degagain} from magnetic-BSM walls, implementing the above prescription. Recall that the walls we are summing over in the index formula are in $\Gamma_S^+$ and are thus oriented from left to right (they have $p/r > q/s$). From~\eqref{eq:gamma-gammatilde}, 
	we then have the following possible configurations for the wall $\gamma$ and its dual:
	\begin{enumerate}
		\item  
		\textbf{Case}~$ \displaystyle \frac{p}{r}  > \frac{-p\ell_\gamma+ q}{-r\ell_\gamma + s} > \frac{q}{s}$, as shown in Figure~\ref{fig:mmeta1}
		\begin{figure}[h]
			\centering
			\begin{subfigure}[b]{0.49\textwidth}
				\includegraphics[width=\textwidth]{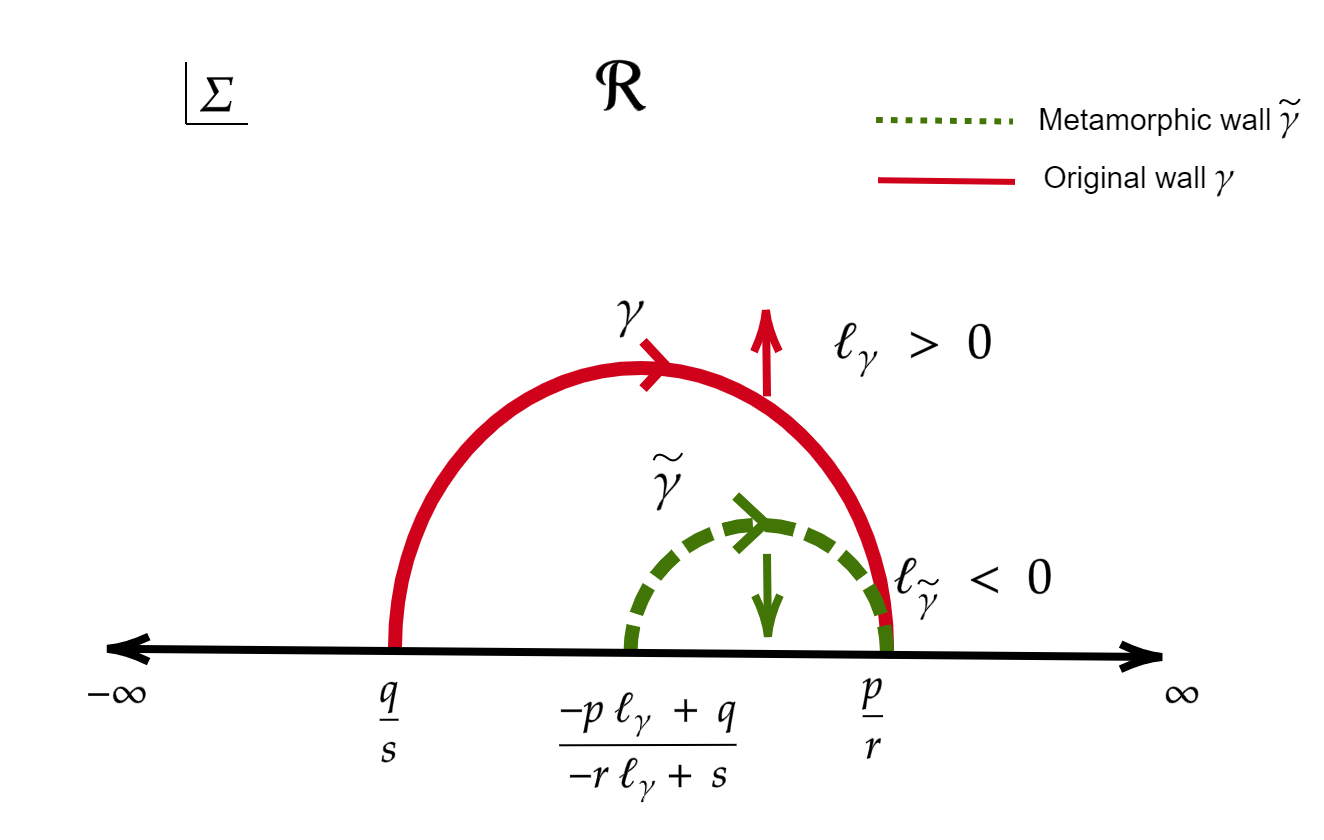}
				\caption{ Case $ \displaystyle \frac{p}{r} > \frac{- p\ell_\gamma + q}{-r \ell_\gamma + s} > \frac{q}{s} \, , \;\; \ell_\gamma = - \ell_{\t \gamma} > 0$ }
				\label{fig:mmeta1a}
			\end{subfigure}
			\begin{subfigure}[b]{0.49\textwidth}
				\includegraphics[width=\textwidth]{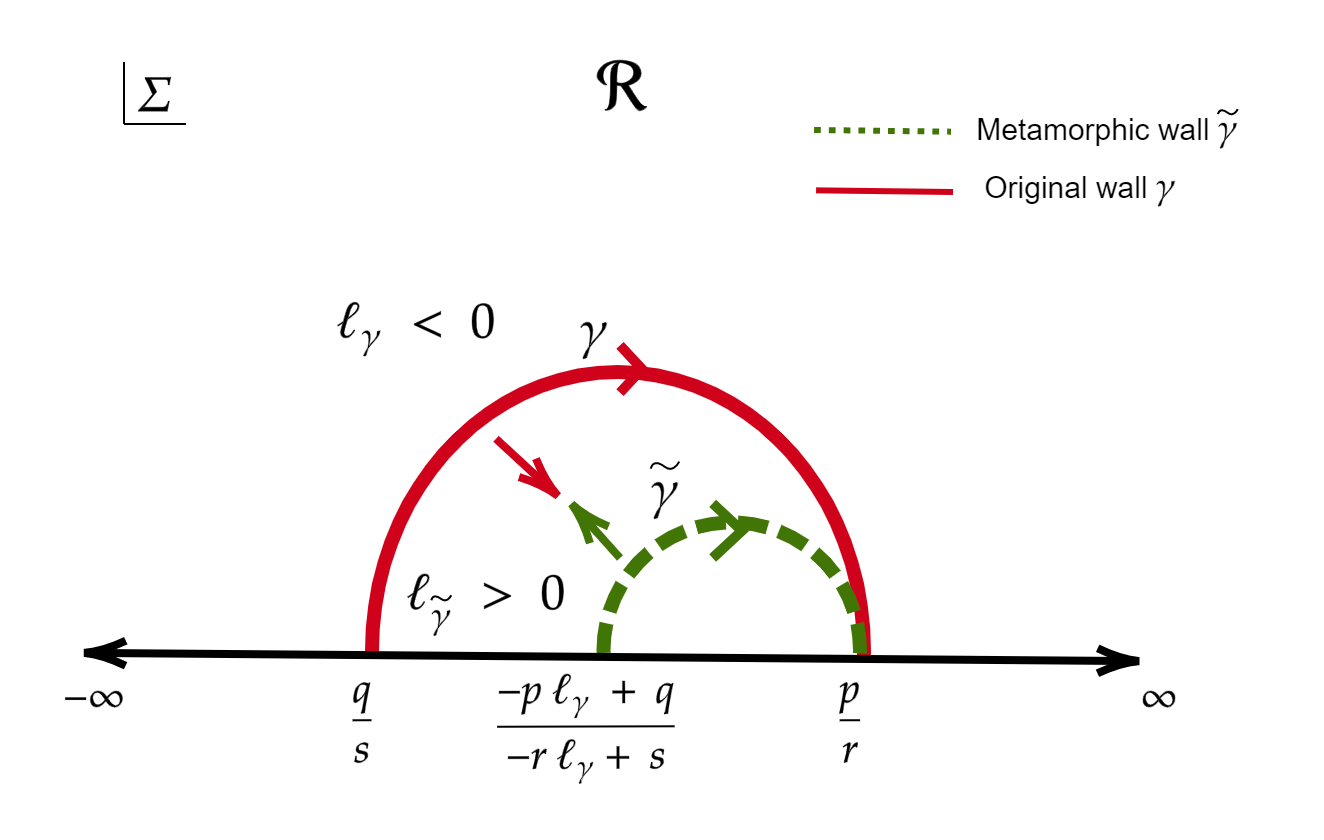}
				\caption{ Case $ \displaystyle \frac{p}{r} > \frac{- p\ell_\gamma + q}{-r \ell_\gamma + s} > \frac{q}{s} \, , \;\; \ell_\gamma = - \ell_{\t \gamma} < 0$ }
				\label{fig:mmeta1b}
			\end{subfigure}
			\caption{Metamorphosis for $\displaystyle \frac{p}{r} > \frac{- p\ell_\gamma + q}{-r \ell_\gamma + s} > \frac{q}{s}$}
			\label{fig:mmeta1}
		\end{figure}
		\begin{enumerate}
			\item     
			The situation~$\ell_\gamma = - \ell_{\t \gamma} > 0$ 
			shown in Figure~\ref{fig:mmeta1a} leads to a contradiction as follows. 
			If $-r\ell_\gamma+s>0$, then we find from $ \displaystyle  \frac{- p\ell_\gamma + q}{-r \ell_\gamma + s} > \frac{q}{s}$ 
			that $- p s \ell_\gamma + qs > -rq \ell_\gamma + qs$ which implies $0 > \ell_\gamma (ps - qr) = \ell_\gamma$ (because $\gamma \in PSL(2,\mathbb{Z})$), 
			which contradicts the assumption that $ \displaystyle  \ell_\gamma > 0$. For $-r\ell_\gamma+s<0$, we find 
			from $ \displaystyle \frac{p}{r} > \frac{- p\ell_\gamma + q}{-r \ell_\gamma + s} $ that $- p r \ell_\gamma + ps < -p r \ell_\gamma + qr$ which 
			leads to the contradiction $1=ps - qr <0$. Therefore this scenario does not occur.
			
			\item     
			The situation $ \displaystyle \ell_\gamma = - \ell_{\t \gamma} < 0$ as seen in Figure~\ref{fig:mmeta1b} 
			does occur but 
			does not contribute to the index computed in the region~$ \displaystyle \mathcal{R}$ owing to the BSM prescription presented below Definition~\ref{def:mag-meta-wall}.
		\end{enumerate}
		
		\item 
		\textbf{Case}~$ \displaystyle \frac{p}{r} > \frac{q}{s} > \frac{-p\ell_\gamma+ q}{-r\ell_\gamma + s} $, as shown in Figure~\ref{fig:mmeta2}
		\begin{figure}[h]
			\centering
			\begin{subfigure}[b]{0.49\textwidth}
				\includegraphics[width=\textwidth]{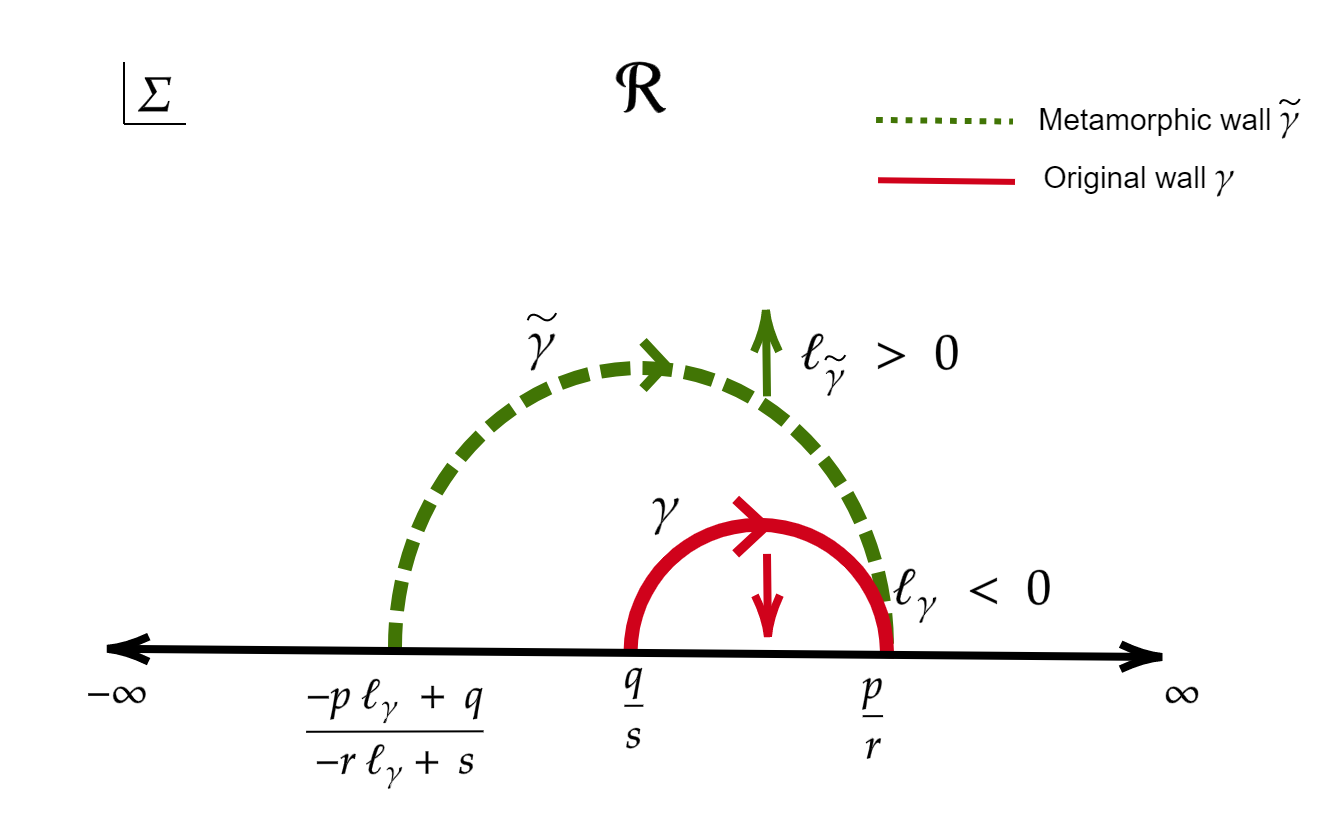}
				\caption{ Case $ \displaystyle \frac{p}{r}  > \frac{q}{s} > \frac{- p\ell_\gamma + q}{-r \ell_\gamma + s} \, , \;\; \ell_\gamma = -\ell_{\t \gamma} < 0$}
				\label{fig:mmeta2a}
			\end{subfigure}
			\begin{subfigure}[b]{0.49\textwidth}
				\includegraphics[width=\textwidth]{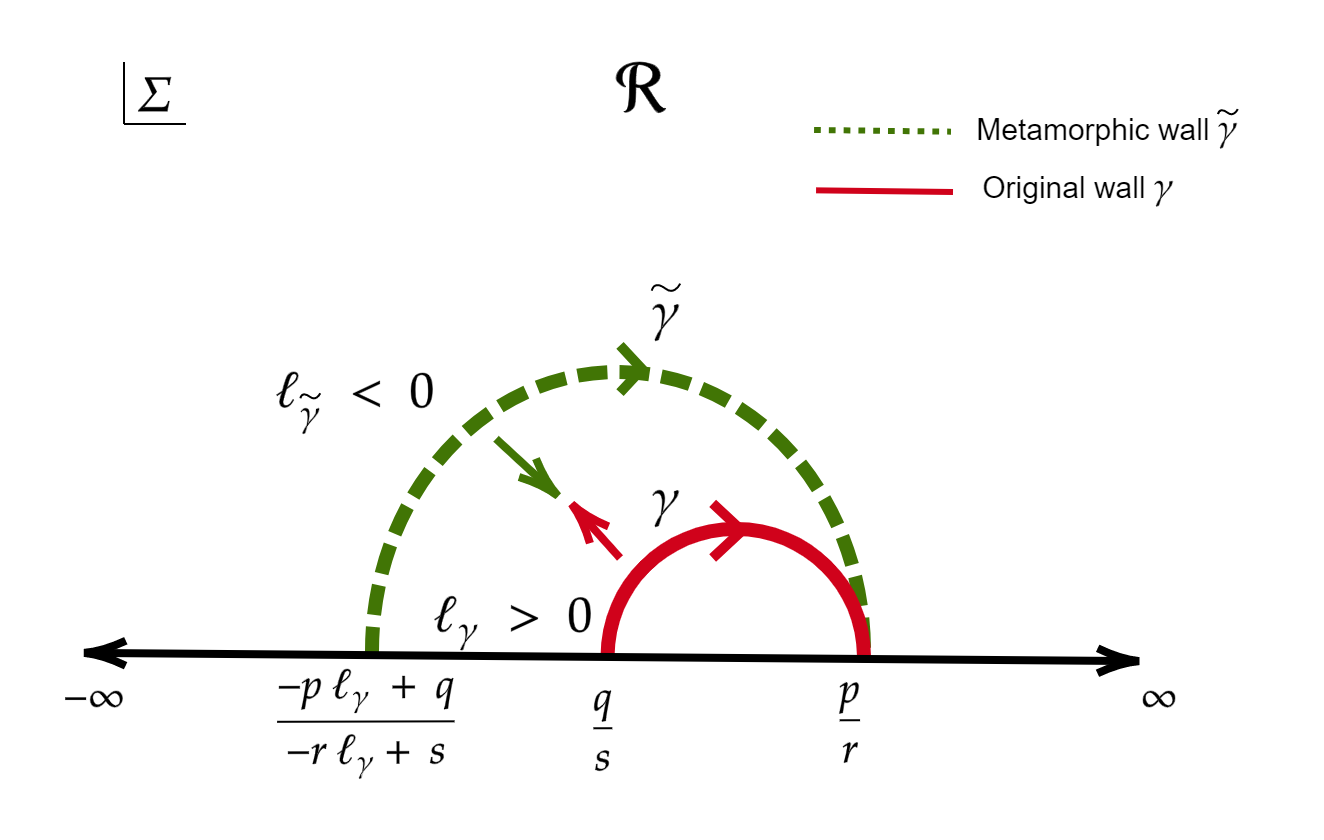}
				\caption{ Case $ \displaystyle \frac{p}{r} > \frac{q}{s} > \frac{- p\ell_\gamma + q}{-r \ell_\gamma + s} \, , \;\; \ell_\gamma = -\ell_{\t \gamma} >0$}
				\label{fig:mmeta2b}
			\end{subfigure}
			\caption{Metamorphosis for $\displaystyle \frac{p}{r}> \frac{q}{s} >  \frac{- p\ell_\gamma + q}{-r \ell_\gamma + s} $}
			\label{fig:mmeta2}
		\end{figure}
		\begin{enumerate}
			\item 
			For the case~$\ell_\gamma = -\ell_{\t \gamma} < 0$ as in Figure~\ref{fig:mmeta2a}, we run into a contradiction analogous to the one of Figure~\ref{fig:mmeta1a}.
			
			\item 
			The case of $ \displaystyle \ell_\gamma = -\ell_{\t \gamma} > 0$ as in Figure~\ref{fig:mmeta2b} does again 
			occur but does not contribute to the black hole degeneracy in the region $ \displaystyle \mathcal{R} $ owing to the BSM prescription.
		\end{enumerate}
		
		\item  
		\textbf{Case} $ \displaystyle   \frac{-p\ell_\gamma+ q}{-r\ell_\gamma + s } >\frac{p}{r} > \frac{q}{s}$ as shown in Figure~\ref{fig:mmeta3}
		\begin{figure}[h]
			\centering
			\begin{subfigure}[b]{0.49\textwidth}
				\includegraphics[width=\textwidth]{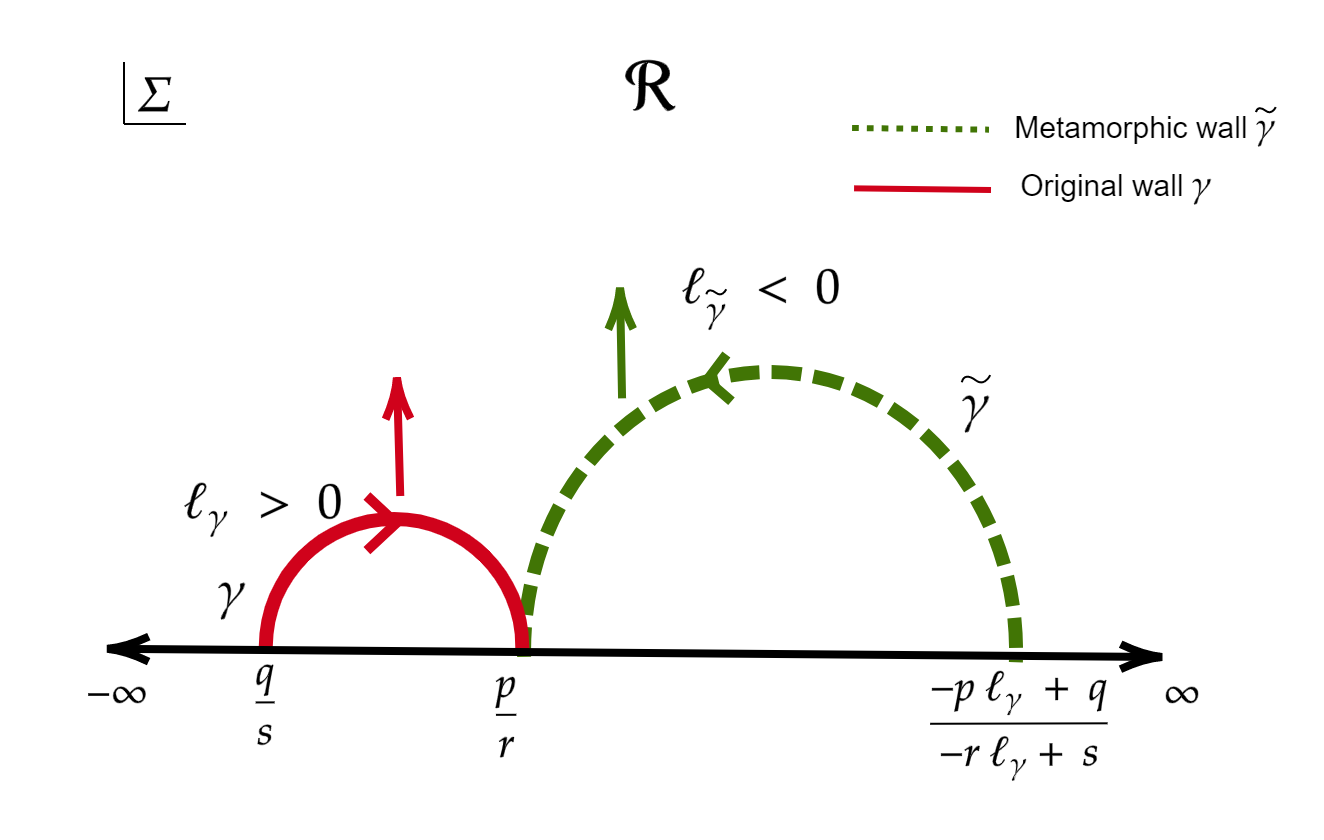}
				\caption{ Case $ \displaystyle \frac{- p\ell_\gamma + q}{-r \ell_\gamma + s} > \frac{p}{r} > \frac{q}{s} \, , \;\; \ell_\gamma = - \ell_{\t \gamma} > 0$}
				\label{fig:mmeta3a}
			\end{subfigure}
			\begin{subfigure}[b]{0.49 \textwidth}
				\includegraphics[width=\textwidth]{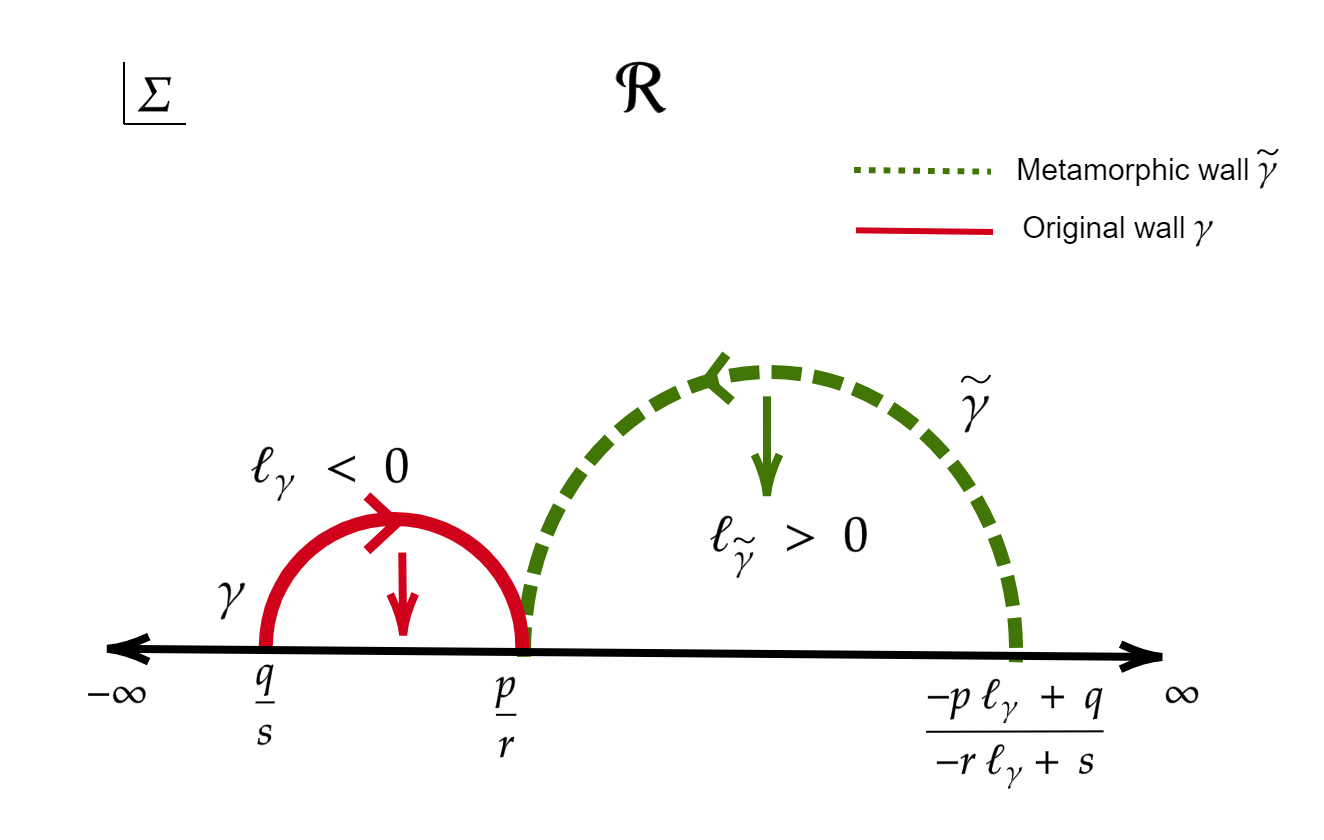}
				\caption{ Case $ \displaystyle  \frac{- p\ell_\gamma + q}{-r \ell_\gamma + s} > \frac{p}{r} > \frac{q}{s} \, , \;\; \ell_\gamma = -\ell_{\t  \gamma} < 0$}
				\label{fig:mmeta3b}
			\end{subfigure}
			\caption{Metamorphosis for $\displaystyle   \frac{- p\ell_\gamma + q}{-r \ell_\gamma + s} > \frac{p}{r}>  \frac{q}{s} $}
			\label{fig:mmeta3}
		\end{figure}
		\begin{enumerate}
			\item 
			For the case as in Figure~\ref{fig:mmeta3a}, there will be a contribution to the index in the region~$ \displaystyle \mathcal{R} $ from the 
			walls~$\gamma$ and~$\tilde{\gamma}$, in accordance with the BSM prescription. Furthermore, 
			Proposition~\ref{claim:mag-BSM} shows that both contributions are equal, and the prescription states that they must be identified to avoid overcounting.
			
			\item 
			The case shown in Figure~\ref{fig:mmeta3b} does not contribute to the black hole degeneracy in~$ \displaystyle \mathcal{R} $ 
			since neither~$\gamma$ nor~$\tilde{\gamma}$ contribute in~$\mathcal{R}$.
		\end{enumerate}
	\end{enumerate}
	
	In summary, we have shown that the only magnetic-BSM walls that give a non-trivial contribution to~\eqref{eq:degagain} must satisfy~$m_\gamma = -1$,~$n_\gamma \geq 0$,~$\ell_\gamma > 0$, as well as~$ \displaystyle \frac{-p\ell_\gamma+ q}{-r\ell_\gamma + s } > \displaystyle \frac{p}{r}$. Observe now that if~$-r\ell_\gamma + s >0$, then we can rewrite the latter inequality as~$0>ps-qr=1$ which is a contradiction. Thus, we find that the walls giving a non-trivial contribution to the index~\eqref{eq:degagain} must have~$-r\ell_\gamma + s <0$, which implies~$\ell_\gamma > s/r$ and therefore is stronger than~$\ell_\gamma >0$. \\
	
	Upon eliminating $q$ using the condition that the walls are in $PSL(2,\mathbb{Z})$, we can write the three conditions~$m_\gamma = -1$, $n_\gamma \geq 0$ and $\ell_\gamma > s/r$ as
	\begin{equation}
	\label{eq:master-mag}
	\begin{split}
	m\Bigl(\frac{p s - 1}{r}\Bigr)^2 - \ell\Bigl(\frac{p s - 1}{r}\Bigr)s + n s^2 \geq 0 \, , \\
	m p^2 - \ell p r + n r^2 = -1 \, , \\
	- 2 n r s  - 2 m \frac{p}{r}(p s-1) + \ell (2p s - 1) > \frac{s}{r} \, .
	\end{split}
	\end{equation}
	We split the discussion in various cases depending on the values of the charges~$(n,\ell,m)$, subject to the conditions $4mn-\ell^2 < 0$, $m>0$ and $0 \leq \ell \leq m$. We further focus on the~$\Gamma_S^+$ walls that have 
	$r,s>0$.
	
	\begin{enumerate}
		\item
		\textbf{\underline{Case 1}}: $m>0$, $n=-1$ \\
		In this case we solve for $r$ in \eqref{eq:master-mag} and obtain two solutions
		\begin{equation}
		r_\pm = \frac12\Bigl(\pm \sqrt{p^2|\Delta| + 4}-\ell p \Bigr) \, .
		\end{equation}
		Since $r_-$ is negative we can discard it and focus on the $r_+$ solution. Inserting this in the inequalities $n_\gamma \geq 0$ and $\ell_\gamma > s/r$, 
		we obtain the following inequalities on $s$, 
		\begin{equation}
		\text{max}\Bigl[\frac12\Bigl(\ell\sqrt{p^2|\Delta| + 4}-p|\Delta|\Bigr),0\Bigr] < s \leq \frac14\lp \ell+ \sqrt{|\Delta|}\rp \lp\sqrt{p^2|\Delta| + 4} - p \sqrt{|\Delta|}\rp \, ,
		\end{equation}
		where we have taken into account the fact that we are only interested in solutions with~$s>0$. 
		Clearly the right-hand side must be greater or equal to one for this to have solutions in~$\mathbb{Z}$, which in turn 
		translates to an upper bound on~$p$, given below. A lower bound on~$p$ arises because the lower bound on~$s$ 
		will become~$p$-dependent for sufficiently small~$p$ (certainly for~$p\leq 0$). In that case we know that the lower 
		bound~$\frac12\Bigl(\ell\sqrt{p^2|\Delta| + 4}-p|\Delta|\Bigr)=\frac12\Bigl(\ell (2r_++\ell p) -p|\Delta|\Bigr)$ is an 
		integer or a half-integer. Since~$s$ has to be strictly larger than this lower bound but smaller than the upper bound, we 
		find that the gap between the upper and lower bound on~$s$ has to be at least 1/2, which leads to a lower bound on~$p$. 
		The final range one obtains is
		\begin{equation}
		-1+\frac{1}{4m} \lp \frac{\ell (1+4m)}{\sqrt{|\Delta|}+1}\rp\leq p \leq \frac{1}{2} + \frac{1}{2m}\lp \frac{\ell(m+1)}{\sqrt{|\Delta|}} - 1\rp \, .
		\end{equation}
		Since~$0 \leq \ell \leq m$, the upper bound on~$p$ is maximized by taking~$\ell=m$, in which 
		case one obtains~$p\leq \frac12\Bigl( 1-\frac{1}{m}+\frac{m+1}{\sqrt{m(m+4)}}\Bigr) <1$, while the 
		lower bound on~$p$ trivially implies that~$p\geq 0$. So, we actually find that there are only solutions 
		with~$p=0$, which then implies that $r=r_+=1$ and $q=(ps-1)/r=-1$. The range for~$s$ simplifies 
		substantially and the only matrices that contribute in this case are
		\begin{equation}
		\label{eq:walls-magBSM-1}
		\begin{pmatrix} 0 & -1 \\ 1 & s \end{pmatrix}  \, , \quad \text{with} \quad \ell < s \leq\frac12\lp \ell+ \sqrt{|\Delta|}\rp<m+1 \, .
		\end{equation}
		Here we have used that $0\leq \ell \leq m$ to get a simple $m$-dependent upper bound on $s$.
		
		\item 
		\textbf{\underline{Case 2}}: $m>0$, $n = 0$ and $\ell > 0$ \\
		In this case, the system \eqref{eq:master-mag} imposes
		\begin{equation}
		r = \frac{1+mp^2}{\ell p} \, .
		\end{equation}
		Requiring $r>0$ to be an integer fixes $p= 1$ and $\ell\,|\,(m+1)$. Then the condition~$\ell_\gamma> s/r$ is 
		automatically satisfied for $s>0$, while the condition $m_\gamma \geq 0$ requires $s=1$. Thus, the matrices satisfying \eqref{eq:master-mag} are of the form
		\begin{equation}
		\label{eq:walls-magBSM-2}
		\begin{pmatrix} 1 & 0 \\ \frac{m+1}{\ell} & 1 \end{pmatrix}  \, , \quad \text{with} \quad \ell\,|\,(m+1) \, .
		\end{equation}
		This set of matrices has entries that are trivially bounded from above by $(m+1)/\ell$. Among all matrices 
		that contribute to the index~\eqref{eq:degagain}, we obtain here the maximal entry~$m+1$ for~$\ell=1$.
		
		\item
		\textbf{\underline{Case 3}}: $m>0$, $n > 0$ and $\ell > 0$ \\
		In this case we solve for $r$ using $m_\gamma = -1$ and obtain two solutions
		\begin{equation}
		\label{eq:r-sol-3}
		r_\pm = \frac{1}{2n}\Bigl(\ell p \pm \sqrt{p^2|\Delta| - 4n}\Bigr) \, .
		\end{equation}
		Note that sign($p) = $ sign($r_\pm)$, so our restriction to $r>0$ implies in this case $p>0$. 
		The reality of the square root in $r_\pm$ actually implies a stronger lower bound on $p$, 
		\begin{equation}
		\label{eq:p-lower}
		2\sqrt{\frac{n}{|\Delta|}} \leq p \, .
		\end{equation}
		Together with the conditions $n_\gamma \geq 0$ and $\ell_\gamma > s/r$, we find
		\begin{equation}
		\label{eq:s-window}
		\text{max}\Bigl[\frac1{2n}\Bigl(p|\Delta| \pm \ell\sqrt{p^2|\Delta| - 4n}\Bigr),0\Bigr] < s \leq \frac1{4n}\lp \ell+ \sqrt{|\Delta|} \rp \lp p\sqrt{|\Delta|}\pm \sqrt{p^2|\Delta| - 4n}\rp \, ,
		\end{equation}
		where the upper and lower signs are for~$r=r_+$ and~$r=r_-$, respectively. We should require that the right-hand side of the above 
		equation be greater than one to have integer solutions for~$s$. 	
		For the lower sign, this criterion yields an upper bound on~$p$,
		\begin{equation}
		\label{eq:p-upper}
		p \leq \frac{1}{2} + \frac{1}{2m}\lp \frac{\ell(m+1)}{\sqrt{|\Delta|}} - 1\rp \, ,
		\end{equation}
		which shows that there is a finite number of walls with~$r=r_-$ contributing to~\eqref{eq:degagain}. 
		Furthermore, the wall matrix entries are again bounded by simple $m$-dependent functions, as follows. 
		For~$p$ as in~\eqref{eq:p-upper} we notice that the upper bound is maximized 
		for~$\ell=m$ and~$|\Delta|=1$,\footnote{These values cannot actually be obtained, so there is a slightly stronger but more complicated bound.} 
		in which case one finds~$p<\frac12 \lp 2+m-\frac1m\rp<1+\frac{m}{2}$. Similarly, we can derive an~$m$-dependent upper 
		bound on~$s>0$ as follows:~$p \sqrt{|\Delta|}-\sqrt{p^2|\Delta| - 4n}$ is a monotonically decreasing function of~$p$ and 
		therefore maximal when~$p$ is at its lower bound $2\sqrt{n/|\Delta|}$ from~\eqref{eq:p-lower}. Taking into 
		account that~$1\leq n$ and $\sqrt{|\Delta|}<\ell \leq m$, we then find~$s\leq \frac{\ell+\sqrt{|\Delta|}}{2\sqrt{n}}<m$. 
		Using the upper bound on~$s$ we likewise find an upper bound on the remaining 
		entry,~$0\leq q=(ps-1)/r_- \leq \lp p\sqrt{|\Delta|}-\sqrt{p^2|\Delta|-4n}\rp/2 < \sqrt{m}/2$. \\
		For the upper sign, which corresponds to picking~$r=r_+$ in~\eqref{eq:r-sol-3}, requiring that the right-hand side 
		of~\eqref{eq:s-window} be greater than one does not yield additional constraints on~$p$. In this case, we thus only have the lower bound~\eqref{eq:p-lower}. 
		However, numerical investigations up to~$m=30$ show that the set of walls with~$r=r_+$ is finite, and in fact 
		consists of only a single element for a given value of~$m,n,\ell > 0$. Furthermore, the entries of the matrix associated 
		to such walls are always strictly less than~$m$. It seems that imposing integrality of the matrix entries on top of the above 
		conditions severely restricts the contributing walls with~$r=r_+$, although we have not managed to show this analytically. 
		We leave this as an interesting problem for the future.
	\end{enumerate}
	
	The above analysis shows that the set of walls at which magnetic-BSM occurs and that give a non-trivial contribution to the index~\eqref{eq:degagain} is finite with entries bounded from above by~$m+1$ (see Equation~\eqref{eq:walls-magBSM-2}). Aside from the case with~$m,n,\ell > 0$ and~$r = r_+$, we were able to show this analytically. Nevertheless, our numerical investigations have shown that the same conclusion holds for the latter walls. Some more details are presented in Appendix~\ref{sec:finiteness}.

	\subsection{Electric metamorphosis case: $ \displaystyle n_\gamma = -1, m_\gamma \geq 0 $ }
	\label{sec:el-BSM}

	Having expounded the details of the magnetic metamorphosis case in the previous subsection, we can make use of these results to work out the electric metamorphosis case at almost no extra cost. This follows from combining Proposition~\ref{claim:mag-BSM} with the observation below~\eqref{eq:Strafo}, which shows that acting with $\tilde{S}$ on the metamorphic dual $\tilde{\gamma}$~\eqref{eq:mag-BSM-meta} of a wall $\gamma$ with~$m_\gamma = -1$ produces a wall with~$n_\gamma = -1$ and with the same orientation as that of~$\gamma$. Indeed, the charges associated with the wall~$\tilde{\gamma} \cdot \tilde{S}$ are given by
	\begin{equation}
	\label{eq:from-mag-to-el}
	(n_{\tilde{\gamma}\tilde{S}},\,\ell_{\tilde{\gamma}\tilde{S}},\,m_{\tilde{\gamma}\tilde{S}}) 
	\= (m_\gamma,\,\ell_\gamma,\,n_\gamma) \= (-1,\,\ell_\gamma,\,n_\gamma) \, . 
	\end{equation}
	As in the magnetic-BSM case, for a wall~$\gamma$ such that electric metamorphosis occurs 
	there is another wall~$\tilde{\gamma}$ 
	which gives the same contribution to the index: 
	\begin{definition}
		\label{def:el-meta-wall}
		For a wall $ \displaystyle \gamma $ with $ \displaystyle n_\gamma = -1, \ m_\gamma \geq  0 $, we define its metamorphic dual as 
		\begin{equation}
		\label{eq:el-BSM-meta}
		\t \gamma \= \gamma \d \wmat{1}{0}{-\ell_\gamma}{1} \, .
		\end{equation}
	\end{definition}
	We must then employ a prescription analogous to the one presented below Definition~\ref{def:mag-meta-wall} for electric-BSM contributions to avoid overcounting. This is again necessary since an electric-BSM wall and its metamorphic dual~\eqref{eq:el-BSM-meta} have the same contribution to the index:
	\begin{proposition}
		\label{prop:el-BSM}
		For a given set of charges~$(n,\ell,m)$, the wall $ \displaystyle \t \gamma $ has the same index 
		contribution as $ \displaystyle \gamma $ to the 
		polar coefficients~$\widetilde{c}_m(n,\ell)$.
	\end{proposition}
	\begin{proof}
		This is proven completely analogously to the proof of Proposition~\ref{claim:mag-BSM}.
	\end{proof}
	
		Furthermore we recall that, as explained below~\eqref{eq:thetastepR}, the summand of the counting formula for negative discriminant states is invariant under an~$\tilde{S}$-transformation. From~\eqref{eq:from-mag-to-el}, it is clear that the electric-BSM wall~$\tilde{\gamma}\cdot\tilde{S}$ gives the same contribution to the polar coefficients~$\widetilde{c}_m(n,\ell)$ as the magnetic-BSM wall~$\gamma$. Thus, in addition to the above prescription that requires us to identify an electric-BSM wall with its metamorphic dual, we also 
need to identify the contribution of electric-BSM walls with the contribution of magnetic-BSM walls to avoid further overcounting. The BSM prescription in the case of magnetic or electric walls therefore identifies four contributions together for a given set of charges~$(n,\ell,m)$.
	
	From Definition \ref{def:el-meta-wall}, it is easy to see that a wall $ \displaystyle \gamma $ and its metamorphic dual $ \displaystyle \t \gamma $ have the same starting point $q/s$,
	\begin{align}
	\gamma = \wmat{p}{q}{r}{s} \;\; \Longleftrightarrow \;\; \t \gamma = \wmat{- q \ell_\gamma + p}{\ \ q}{-s \ell_\gamma + r}{\ \ s} \, .
	\end{align} 
	Given this and the fact that we look for walls in~$\Gamma_S^+$ (with $p/r > q/s$), we have the following possible configurations for the electric-BSM wall $\gamma$ and its dual:
	\begin{enumerate}
		\item  
		\textbf{Case} $ \displaystyle \frac{p}{r}  > \frac{-q\ell_\gamma+ p}{-s\ell_\gamma + r} > \frac{q}{s}$, as shown in Figure~\ref{fig:nmeta1}
		\begin{figure}[h]
			\centering
			\begin{subfigure}[b]{0.49\textwidth}
				\includegraphics[width=\textwidth]{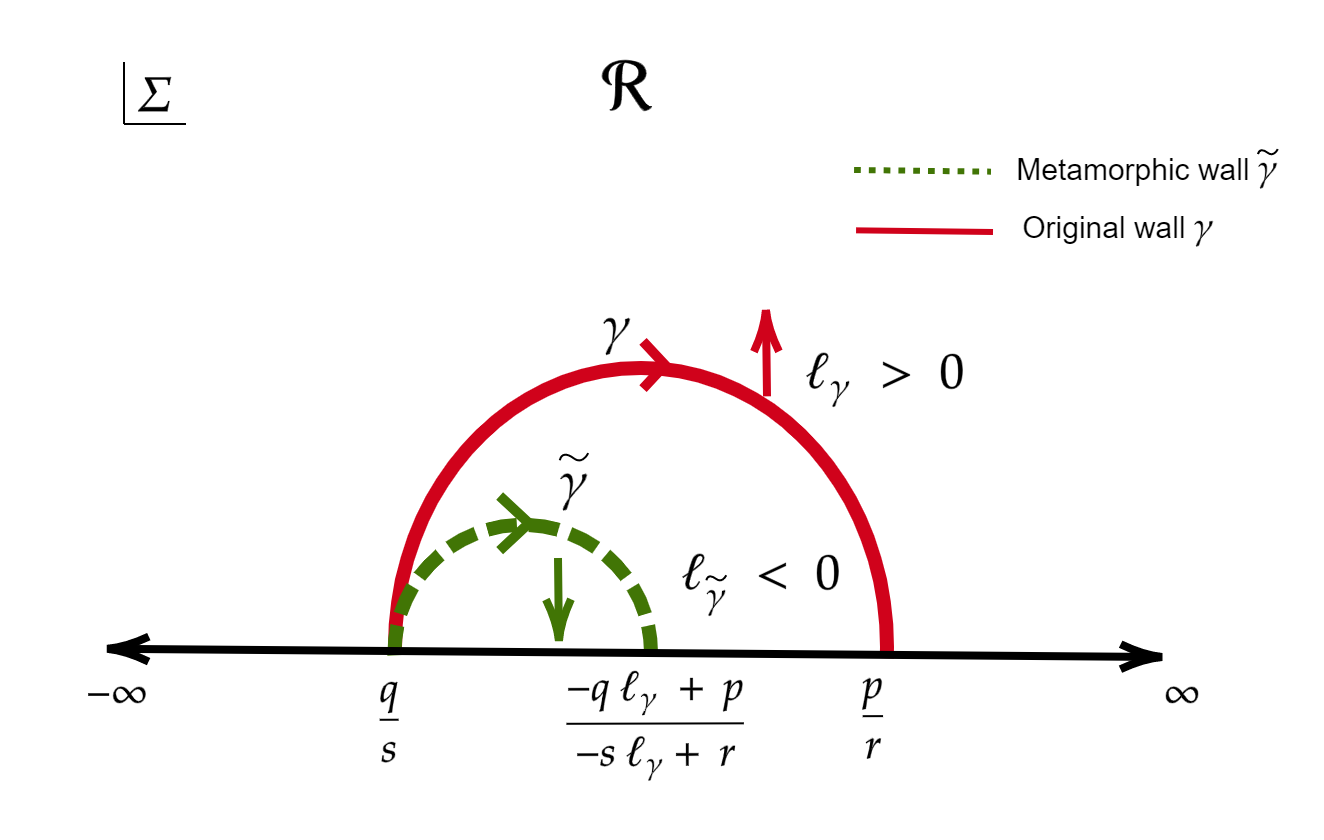}
				\caption{ Case $ \displaystyle \frac{p}{r} > \frac{- q\ell_\gamma + p}{-s \ell_\gamma + r} > \frac{q}{s} \, , \;\; \ell_\gamma = -\ell_{\t \gamma} > 0$ }
				\label{fig:nmeta1a}
			\end{subfigure}
			\begin{subfigure}[b]{0.49\textwidth}
				\includegraphics[width=\textwidth]{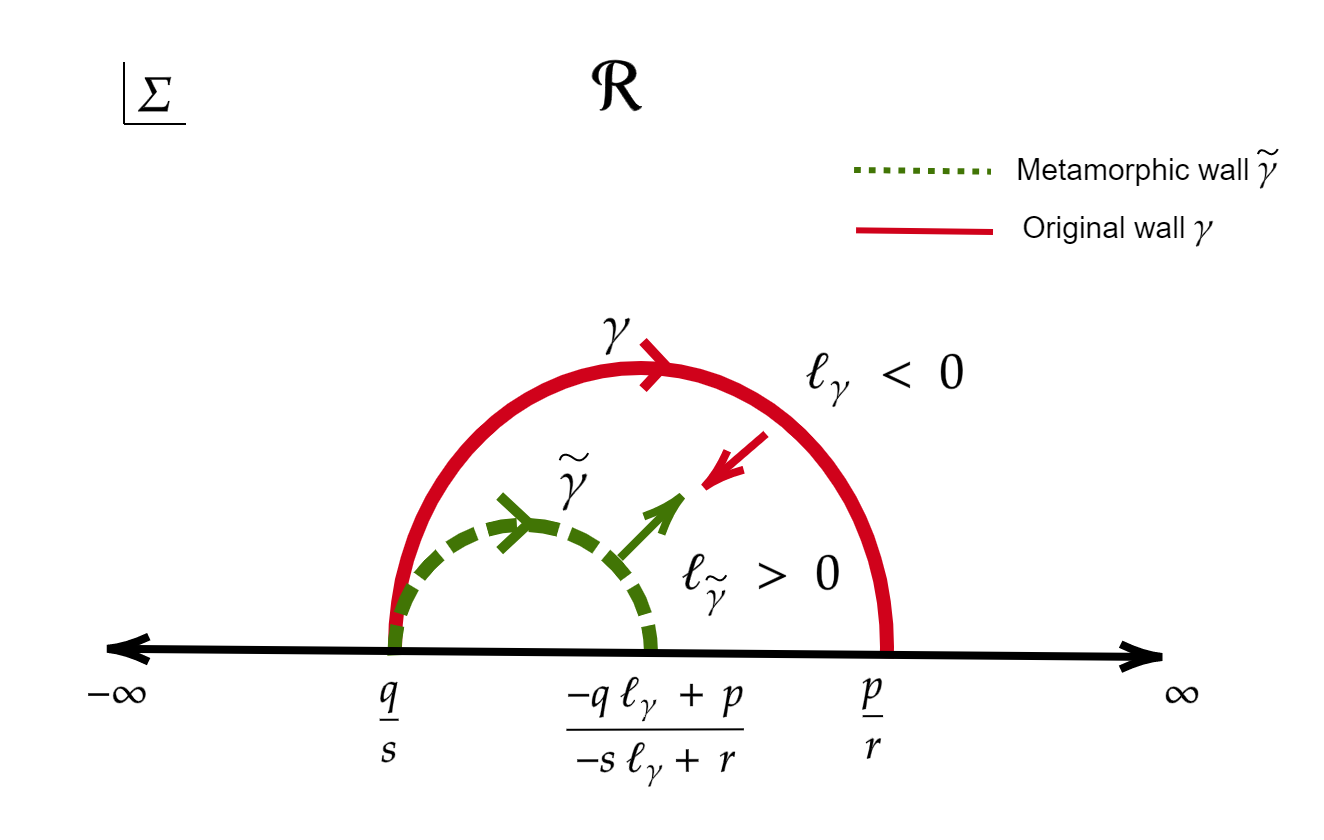}
				\caption{ Case $ \displaystyle \frac{p}{r} > \frac{- q\ell_\gamma + p}{-s \ell_\gamma + r} > \frac{q}{s} \, , \;\; \ell_\gamma = -\ell_{\t \gamma} < 0$}
				\label{fig:nmeta1b}
			\end{subfigure}
			\caption{Metamorphosis for $ \displaystyle \frac{p}{r} > \frac{- q\ell_\gamma + p}{-s \ell_\gamma + r} > \frac{q}{s}$}\label{fig:nmeta1}
		\end{figure}
		
		\begin{enumerate}
			\item     
			For the situation~$\ell_\gamma = -\ell_{\t \gamma} > 0$ as in Figure~\ref{fig:nmeta1a}, one can 
			show that this configuration leads to a contradiction, analogous to the magnetic-BSM case 
			of Figure~\ref{fig:mmeta1a}. Therefore this scenario does not occur.
			
			\item     
			The case~$\displaystyle \ell_\gamma = -\ell_{\t \gamma} < 0$ as seen in Figure~\ref{fig:nmeta1b} 
			does occur but does not contribute to the index in the region~$ \displaystyle \mathcal{R} $ owing to the BSM prescription.
		\end{enumerate}
		
		\item 
		\textbf{Case} $ \displaystyle \frac{-q\ell_\gamma+ p}{-s\ell_\gamma + r} >  \frac{p}{r}  > \frac{q}{s}$, as shown in Figure~\ref{fig:nmeta2}
		\begin{figure}[h]
			\centering
			\begin{subfigure}[b]{0.49\textwidth}
				\includegraphics[width=\textwidth]{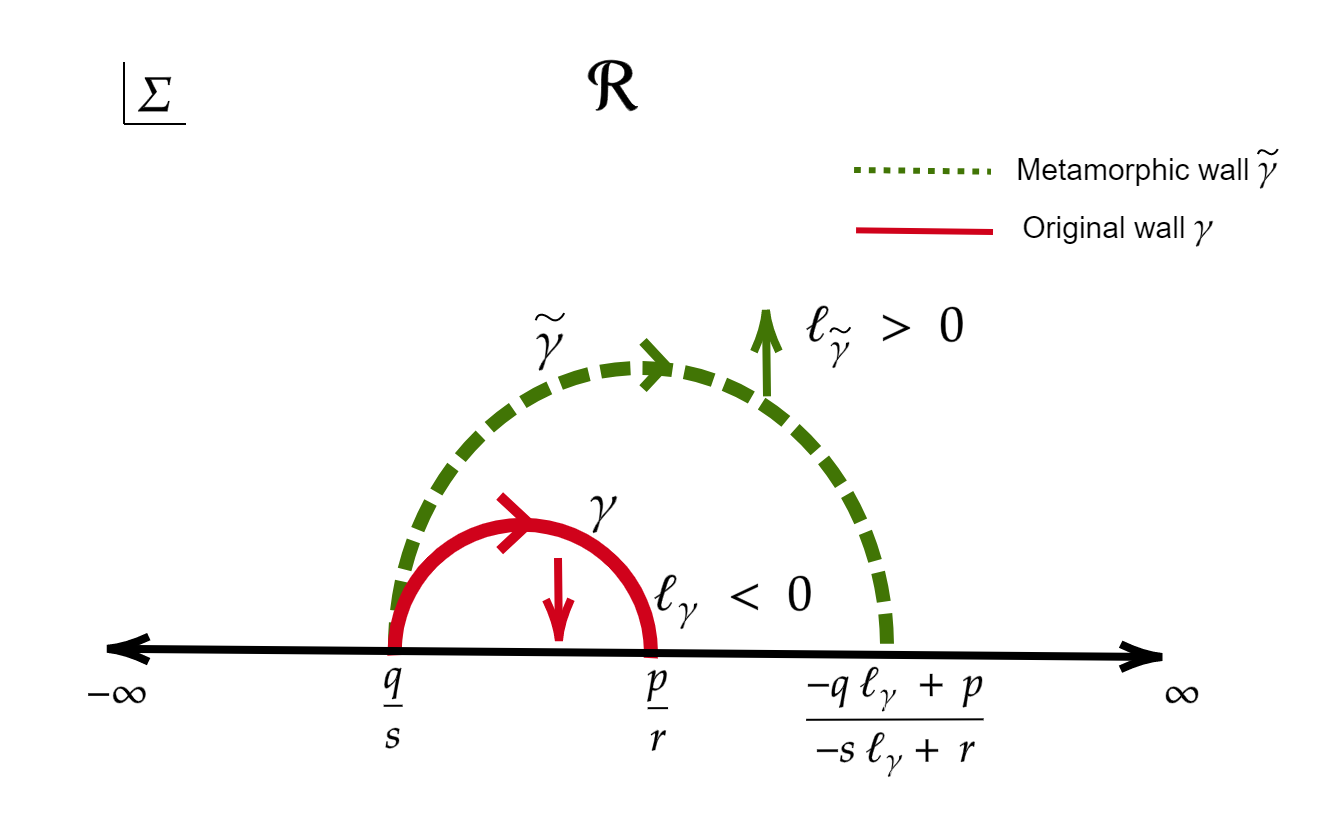}
				\caption{ Case $ \displaystyle \frac{- q\ell_\gamma + p}{-s \ell_\gamma + r} > \frac{p}{r} >  \frac{q}{s} \, , \;\; \ell_\gamma = -\ell_{\t \gamma} > 0$}
				\label{fig:nmeta2a}
			\end{subfigure}
			\begin{subfigure}[b]{0.49\textwidth}
				\includegraphics[width=\textwidth]{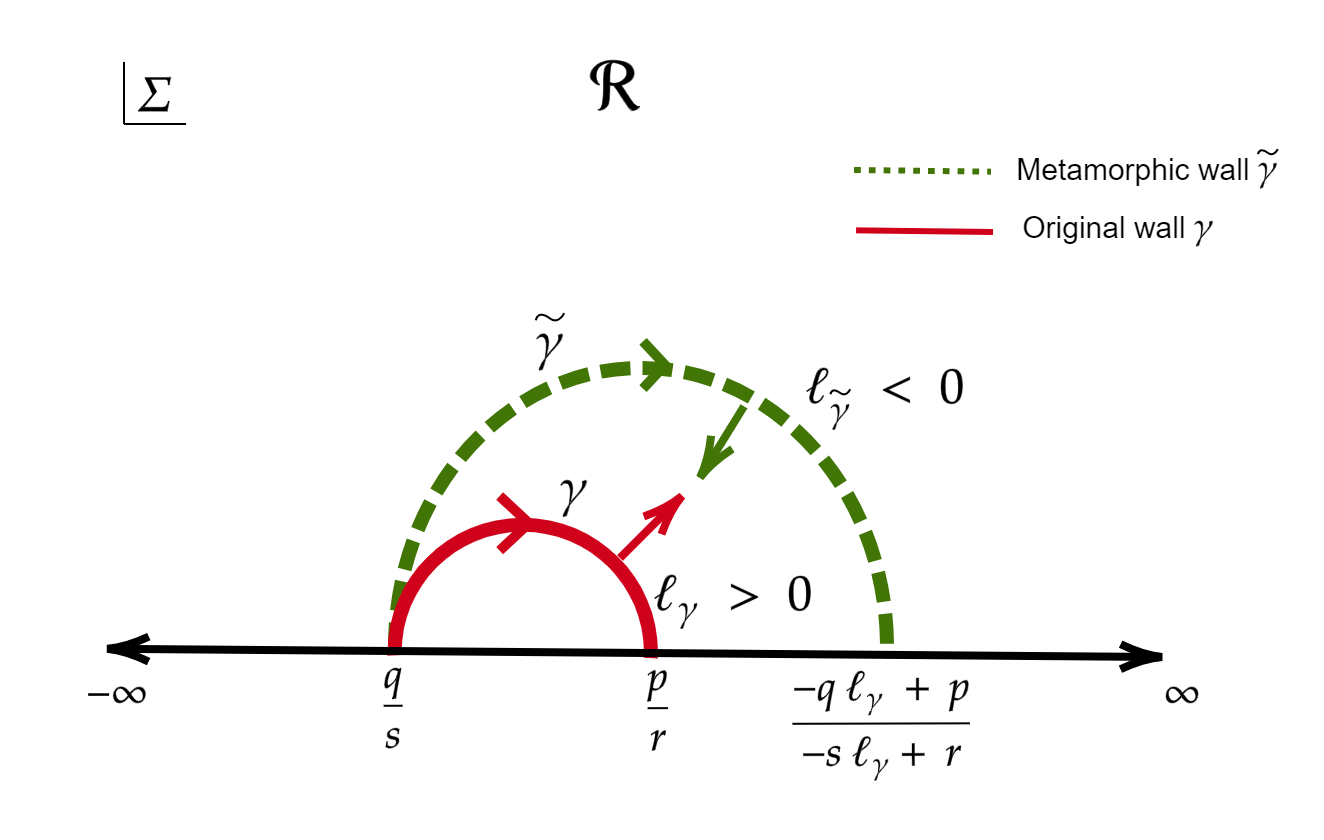}
				\caption{ Case $ \displaystyle \frac{- q\ell_\gamma + p}{-s \ell_\gamma + r} > \frac{p}{r} >  \frac{q}{s} \, , \;\; \ell_\gamma = -\ell_{\t \gamma} < 0$}
				\label{fig:nmeta2b}
			\end{subfigure}
			\caption{Metamorphosis for $ \displaystyle  \frac{- q\ell_\gamma + p}{-s \ell_\gamma + r} > \frac{p}{r} > \frac{q}{s}$}\label{fig:nmeta2}
		\end{figure}
		
		\begin{enumerate}
			\item 
			For the case~$\ell_\gamma = -\ell_{\t \gamma} < 0$ as in~\ref{fig:nmeta2a}, there is again a contradiction 
			which prevents this configuration from happening, as in the magnetic-BSM case of Figure~\ref{fig:mmeta2a}. 
			
			\item 
			The case~$\displaystyle \ell_\gamma = -\ell_{\t \gamma} > 0$ as in Figure~\ref{fig:nmeta2a} does again occur 
			but does not contribute to the index in the region~$\displaystyle \mathcal{R}$  owing to the BSM prescription.
		\end{enumerate}
		
		\item  
		\textbf{Case} $ \displaystyle   \frac{p}{r} >  \frac{q}{s} > \frac{- q\ell_\gamma + p}{-s \ell_\gamma + r} $, as shown in Figure~\ref{fig:nmeta3}
		\begin{figure}[h]
			\centering
			\begin{subfigure}[b]{0.49\textwidth}
				\includegraphics[width=\textwidth]{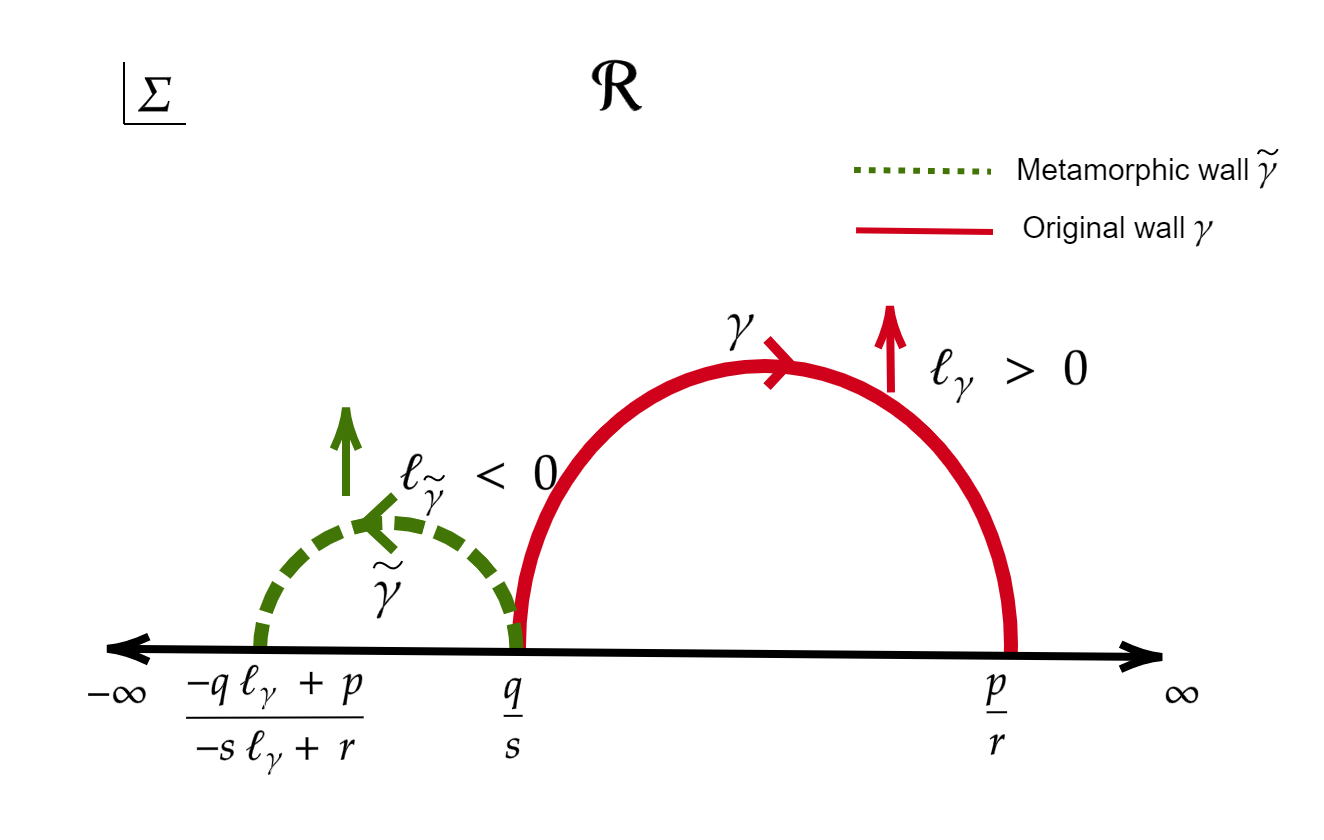}
				\caption{ Case $ \displaystyle  \frac{p}{r} >  \frac{q}{s} > \frac{- q\ell_\gamma + p}{-s \ell_\gamma + r} \, , \;\; \ell_\gamma = -\ell_{\t \gamma} > 0$}
				\label{fig:nmeta3a}
			\end{subfigure}
			\begin{subfigure}[b]{0.49\textwidth}
				\includegraphics[width=\textwidth]{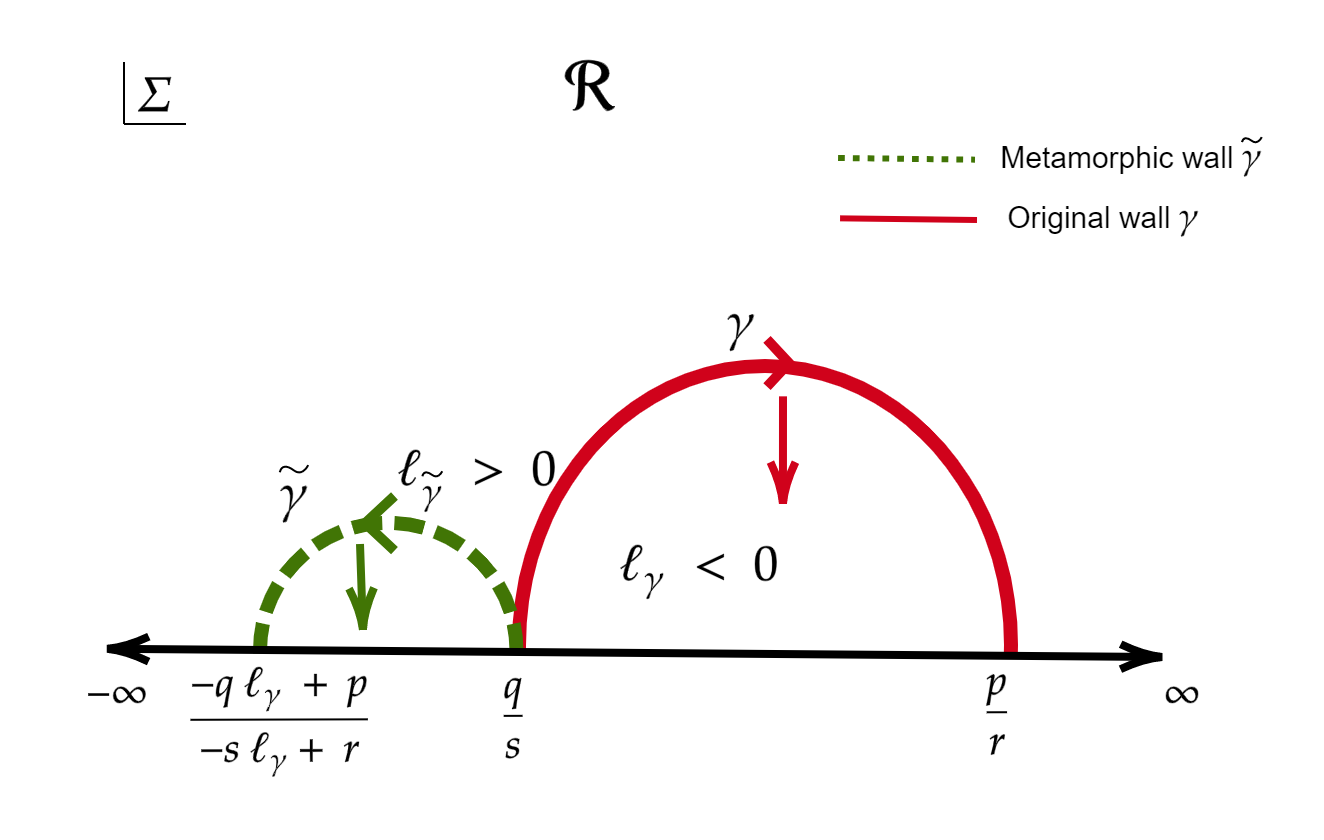}
				\caption{ Case $ \displaystyle \frac{p}{r} >  \frac{q}{s} > \frac{- q\ell_\gamma + p}{-s \ell_\gamma + r} \, , \;\; \ell_\gamma = -\ell_{\t \gamma} < 0$}
				\label{fig:nmeta3b}
			\end{subfigure}
			\caption{Metamorphosis for $ \displaystyle  \frac{p}{r} > \frac{q}{s} >  \frac{- q\ell_\gamma + p}{-s \ell_\gamma + r} $}\label{fig:nmeta3}
		\end{figure}
		
		\begin{enumerate}
			\item 
			For the case as in Figure~\ref{fig:nmeta3a}, there will be a contribution to the index in 
			the region~$ \displaystyle \mathcal{R} $ from~$\gamma$ and~$\tilde{\gamma}$. Here, just as in the magnetic-BSM case, 
			both these contributions are equal owing to Proposition~\ref{prop:el-BSM} and must be identified according to the BSM prescription.
			
			\item The case as shown in Figure~\ref{fig:nmeta3b} does not contribute to the black hole degeneracy 
			in the region~$\displaystyle \mathcal{R} $ since neither~$\gamma$ nor~$\tilde{\gamma}$ does.
		\end{enumerate}
	\end{enumerate}
	
	Just as in the previous section, the above analysis shows that the only electric-BSM walls that give a non-trivial contribution to~\eqref{eq:degagain} must satisfy~$m_\gamma \geq 0$,~$n_\gamma = -1$,~$\ell_\gamma > 0$, as well as~$ \displaystyle \frac{q}{s} > \displaystyle \frac{-q\ell_\gamma+ p}{-s\ell_\gamma + r}$. Once again, the last inequality leads to a stronger restriction on~$\ell_\gamma$, namely~$\ell_\gamma > r/s$.
	We now explicitly give the form of the walls for which electric-BSM occurs, for all values of $(n,\ell,m)$ with the usual restrictions that $4mn - \ell^2 < 0$, $m>0$ and $0\leq \ell\leq m$. We make use of the observation at the beginning of this section regarding the action of~$\tilde{S}$ on the metamorphic dual of a magnetic-BSM wall. 
	
	\begin{enumerate}
		\item
		\textbf{\underline{Case 1}}: $m>0$, $n=-1$ \\
		Acting on the metamorphic dual~\eqref{eq:mag-BSM-meta} of~\eqref{eq:walls-magBSM-1} with an~$\tilde{S}$-transformation, we obtain the walls
		\begin{equation}
		\label{eq:walls-elBSM-1s}
		\begin{pmatrix} 1 & 0 \\ s-\ell & 1 \end{pmatrix}  \, , \quad \text{with} \quad \ell < s \leq\frac12\lp \ell+ \sqrt{|\Delta|}\rp<m+1 \, .
		\end{equation}
		
		\item 
		\textbf{\underline{Case 2}}: $m>0$, $n = 0$ and $\ell > 0$ \\
		Acting on the metamorphic dual~\eqref{eq:mag-BSM-meta} of~\eqref{eq:walls-magBSM-2} with an~$\tilde{S}$-transformation, we obtain the walls
		\begin{equation}
		\label{eq:walls-elBSM-2}
		\begin{pmatrix} \ell & 1 \\ m & \frac{m+1}{\ell} \end{pmatrix}  \, , \quad \text{with} \quad \ell\,|\,(m+1) \, .
		\end{equation}
		This set of matrices has entries that are bounded from above by max$[m,(m+1)/\ell]$. 
		Among all matrices that contribute to the index, we obtain here the maximal entry~$m+1$ for~$\ell=1$.
		
		\item
		\textbf{\underline{Case 3}}: $m>0$, $n > 0$ and $\ell > 0$ \\
		Acting on the metamorphic dual~\eqref{eq:mag-BSM-meta} of the walls of Case 3 in Section~\ref{sec:mag-BSM} 
		with an~$\tilde{S}$-transformation, we obtain a finite set of electric-BSM walls. This can be shown analytically when acting on 
		walls with~$r=r_-$, and numerically when acting on walls with~$r=r_+$. Moreover, all entries are strictly bounded from above by~$m$. 
		See again Appendix~\ref{sec:finiteness} for some numerical checks.
	\end{enumerate}
	
	Just as in the magnetic-BSM analysis of Section~\ref{sec:mag-BSM}, in all the above cases we obtain a finite number of electric-BSM walls whose entries are bounded by~$m+1$ (see Equation~\eqref{eq:walls-elBSM-2}). In the next section, we turn to the final case that remains to be analyzed, which is when both transformed charges~$m_\gamma$ and~$n_\gamma$ are equal to~$-1$.

	\subsection{Dyonic metamorphosis case: $ \displaystyle m_\gamma= n_\gamma = -1 $}
	\label{sec:dyonic-BSM}
	
		The final case of metamorphosis occurs when both the electric and magnetic charges 
		attain their lowest possible values. In the previous two cases of BSM, a magnetic or electric 
		wall came with a single metamorphic dual, and as explained above the resulting four walls for a 
		given charge vector~$(n,\ell,m)$ have to be identified to obtain the correct contribution 
		to the polar coefficients~$\widetilde{c}_m(n,\ell)$. When both $ \displaystyle \mw = \nw = -1 $, there are 
		two centers to be identified and we can identify the magnetic and electric centers alternatively. 
		This generates an infinite sequence of dual walls \cite{Sen2011}. 
		The metamorphic duals can be generated in two ways depending on which center 
		we start the identification with. Since they are equivalent, we choose to start the identification with the magnetic center.
	\begin{definition}
	\label{def:dyonic-meta}
		Let $ \displaystyle \gamma $
		be a wall at which $ \displaystyle m_{\gamma} = n_{\gamma} = -1 $. The metamorphic duals are
		\begin{align} 
		\label{eq:inf_sequence}
		\t\gamma_{\,i} \=  \t \gamma_{\,i-1} \d M_{(i\,{\rm mod }\, 2)} \quad {for} \quad i>0 \, , \quad {and} \quad \t\gamma_0 = \gamma \, , 
		\end{align} 
		where $ M_1,\, M_0 $ are defined as
		\begin{align}
		\label{eq:mnmetidentify}
		\begin{split}
		M_1 \defeq \mat{1}{-\ell_\gamma}{0}{1}, \qquad M_0 \defeq \mat{1}{0}{\ell_\gamma}{1} \, .
		\end{split}
		\end{align}
	\end{definition}
	\ndt For example,~$\t \gamma_1 = \gamma \d M_1,\,  \t \gamma_2 = \t\gamma_1\d M_0, \, \t \gamma_3 = \t \gamma_2 \d M_1, \, \ldots$\footnote{Starting with the electric center, we would have ~$\t \gamma_1 = \gamma \d M_0,\,  \t \gamma_2 = \t\gamma_1\d M_1, \, \t \gamma_3 = \t \gamma_2 \d M_0, \, \ldots$} 
	Note that the identification of the electric center in~$ \displaystyle M_0 $ does not have a~`$ \displaystyle -\lw $' unlike in~\eqref{eq:el-BSM-meta} 
	and this dual wall will have the same sign of $ \displaystyle \lw $ as $ \displaystyle \gamma $. 
	
	\begin{proposition}
		For a given set of charges~$(n,\ell,m)$, the walls $ \displaystyle \t \gamma_{\,i>0}$ all have the same index contribution as $ \displaystyle \gamma $ to the 
		polar coefficients~$\widetilde{c}_m(n,\ell)$.	
	\end{proposition}
	
	\begin{proof}
		From the previous sections, we have already shown that the matrices that identify 
		magnetic centers~\eqref{eq:mag-BSM-meta} and electric centers~\eqref{eq:el-BSM-meta} leave the value of 
		electric and magnetic charges invariant while only flipping the sign of~$\displaystyle \lw $.
		 Therefore, the infinite set of walls generated in~\eqref{eq:inf_sequence} have the same contribution to the index.
	\end{proof}
	We now characterize dyonic metamorphosis. The possible cases for dyonic metamorphosis 
	are shown in Figure~\ref{fig:nmmeta}, where only one case as shown in Figure~\ref{fig:nmmeta2} can in principle 
	contribute to the black hole degeneracy in the attractor region~$\mathcal R$. The reason for this 
	is our BSM prescription: all walls and their metamorphic must contribute in the same region to contribute to 
	the polar coefficients~$\widetilde{c}_m(n,\ell)$. 
	\begin{figure}[h]
		\centering
		\begin{subfigure}[b]{0.49\textwidth}
			\includegraphics[width=\textwidth]{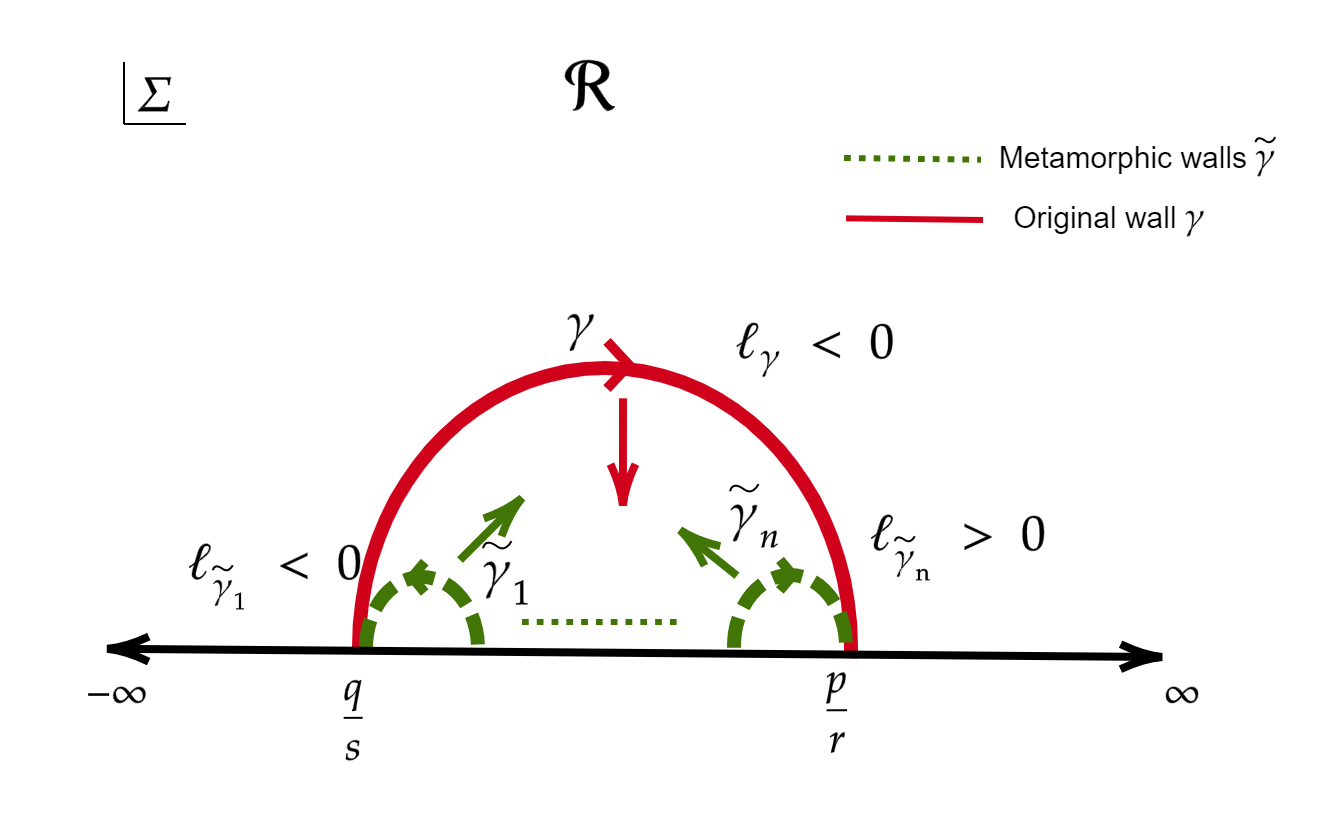}
			\caption{ Case of $\displaystyle m_\gamma = -1, n_\gamma = -1$ metamorphosis where all the metamorphic walls are inside the original wall.}
			\label{fig:nmmeta1}
		\end{subfigure} \
		\begin{subfigure}[b]{0.49 \textwidth}
			\includegraphics[width=\textwidth]{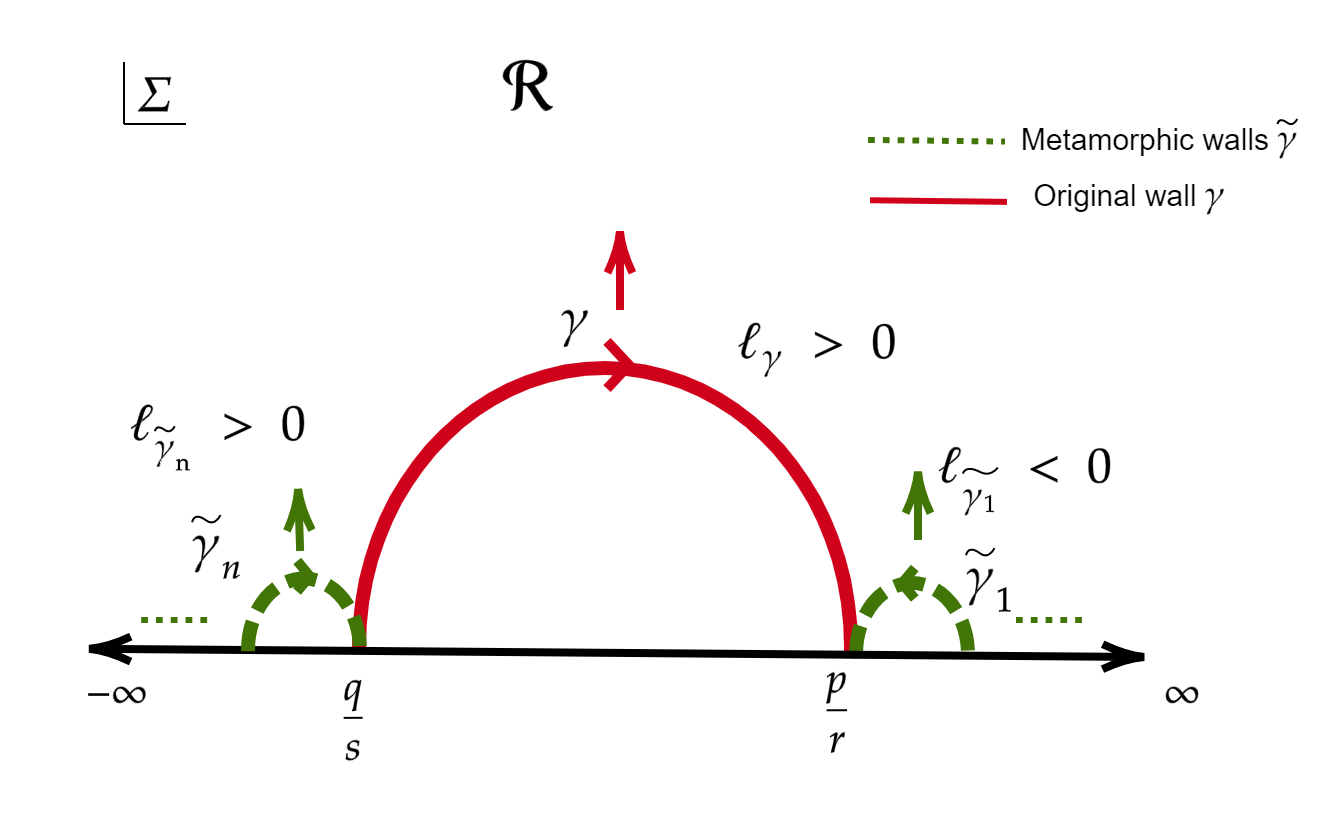}
			\caption{Case of $\displaystyle m_\gamma = -1, n_\gamma = -1$ metamorphosis where all the metamorphic walls are outside the original wall.}
			\label{fig:nmmeta2}
		\end{subfigure}
		\caption{Possible cases of metamorphosis for $\displaystyle m_\gamma = -1, n_\gamma = -1$. There are an infinite series of walls to be identified but we have not depicted them here in order to avoid cluttering of the images. }\label{fig:nmmeta}
	\end{figure}
	To obtain the explicit form of the dyonic-BSM walls
	we must solve the following system,
	\begin{align}
	\label{eq:dyonic-met-charges}
	\begin{split}
	\nw &\= s^2 n + q^2 m - s q \ell \= -1 \, , \\
	\mw &\= r^2 n + p^2 m - r p \ell \= -1  \, , \\
	\lw &\= -sr n- pq m+ \ell(ps + qr ) \= \sqrt{\abs{\Delta} + 4}  \, ,
	\end{split}
	\end{align}
	with $\wmat{p}{q}{r}{s}\in PSL(2,\mathbb{Z})$. 
	It is important to recall that the discriminant~$\Delta$ is a $U$-duality invariant. For this reason, the value of~$\ell_\gamma$ 
	is not independent and is fixed in terms of the charges~$(n,\ell,m)$ as $ \displaystyle  \lw^2 - 4 = \ell^2 - 4 m n = \abs{\Delta}$.
	We further restrict to the case where~$ \displaystyle \lw $ is positive 
	i.e.,~$ \displaystyle \lw = \sqrt{|\Delta| + 4} $ 
	so that the wall contributes to the region~$\mathcal{R}$. Given a charge vector~$(n,\ell,m)$, 
	there is an 
	infinite sequence of walls, all with associated transformed 
	charges~$ \displaystyle (\nw, \ \lw, \ \mw) = (-1, \ \sqrt{\abs{\Delta}+4}, \ -1) $, 
	which get identified by 
	the BSM prescription. \\
	
	We now study the explicit form of the contributing walls. When~$n=0$, the discriminant is~$\displaystyle |\Delta|= \ell^2 $ (with~$\ell > 0$). 
	This reduces \eqref{eq:dyonic-met-charges} to
	\begin{align}
	\label{eq:dyon-met-n0}
	\begin{split}
	\nw &\= q^2 m - s q \ell \= -1 \, , \\
	\mw &\= p^2 m - r p \ell \= -1 \, , \\
	\lw &\= -pq m+ \ell(ps + qr ) \= \sqrt{\abs{\Delta} + 4} \, . 
	\end{split}
	\end{align}
	Demanding that $\lw \in \Z$ implies that $\abs \Delta + 4$ is a perfect square 
	i.e.,~$\ell^2 + 2^2 = \ell_\gamma^2$ for~$\ell_\gamma \in\mathbb{Z}$. 
	We know, however, that there is no Pythagorean triple with 2 as an element. (The difference of two squares form an increasing 
	sequence~$3,5,7,8,\dots $ and this does not include~$2^2=4$.)
	Therefore, there is no dyonic metamorphosis for $ \displaystyle n = 0 $.\\
	
	When~$n\neq 0$, we can solve the system~\eqref{eq:dyonic-met-charges} after eliminating~$q$ 
	using the~$PSL(2,\mathbb{Z})$ relation~$q = (ps - 1)/r$. 
	From~$m_\gamma = -1$ we obtain
	\begin{equation}
	r \= r_\pm \= \frac1{2n}\Bigl(\ell p \pm \sqrt{p^2|\Delta| - 4n}\Bigr) \, .
	\end{equation}
	For each value of~$r$, the condition~$n_\gamma = -1$ is quadratic in~$s$ and yields two branches of solutions. We therefore arrive at 
	the following dyonic-BSM walls\footnote{As mentioned above, a consequence of~$U$-duality is 
	that the equation~$\ell_\gamma = \sqrt{|\Delta| + 4}$ is not independent and does not yield additional constraints.}
	\begin{equation}
	\label{eq:dyon-meta-nsol-1}
	\gamma_{+,\pm} =
	\wmat{p}{\frac{1}{2}\Bigl(\pm p \sqrt{\abs \Delta +4} + \sqrt{\abs \Delta p^2 - 4n}\Bigr)}
	{\frac1{2n}\Bigl(\ell p + \sqrt{p^2|\Delta| - 4n}\Bigr)}{\frac{1}{4n}\Bigl(\ell p + \sqrt{|\Delta|p^2-4n}\Bigr)\Bigl(\ell \pm \sqrt{|\Delta|+4}\Bigr) - mp} \, ,
	\end{equation}
	and
	\begin{equation}
	\label{eq:dyon-meta-nsol-2}
	\gamma_{-,\pm} =
	\wmat{p}{\frac{1}{2}\Bigl(\pm p \sqrt{\abs \Delta +4} - \sqrt{\abs \Delta p^2-4n}\Bigr)}
	{\frac1{2n}\Bigl(\ell p - \sqrt{p^2|\Delta| - 4n}\Bigr)}{\frac{1}{4n}\Bigl(\ell p - \sqrt{|\Delta|p^2-4n}\Bigr)\Bigl(\ell \pm \sqrt{|\Delta|+4}\Bigr) - mp} \, .
	\end{equation}
	Since~$-\gamma_{+,\pm}(-p) = \gamma_{-,\pm}(p)$ and we look for walls in~$PSL(2,\mathbb{Z})$, 
	we can focus on one type of walls, say~$\gamma_{-,\pm}$. 
	We therefore drop the first subscript and simply denote walls of the form~\eqref{eq:dyon-meta-nsol-2} as~$\gamma_\pm$.
	Examining the top-right entry of~$\gamma_\pm$, a necessary condition for these walls to have integer entries 
	is that~$\sqrt{|\Delta| + 4}, \sqrt{|\Delta|p^2 - 4n} \in \mathbb{Z}$. 	
	In the following we let~$y = p$ and~$D = |\Delta|$. The requirement that~$D+4$ is a perfect 
	square implies that~$D$ is not a square (as already observed above), and further that~$D$ 
	is congruent to 0 or 1 modulo 4. 
	
	The requirement~$\sqrt{D y^2 -4n} \in \mathbb{Z}$ can then be expressed as the requirement for~$y$ to be a solution of
	\begin{equation}
	\label{eq:Pell-gen}
	\sqrt{D\,y^2 - 4n} \= x \;\; \Longrightarrow \;\; x^2 - D\,y^2 \= -4n \, ,
	\end{equation}
	with~$x,y \in \mathbb{Z}$. We now split the discussion in two cases.

\begin{enumerate}
	\item \underline{\textbf{Case 1}}:~$n=-1$ \\
	In the case~$n=-1$, the condition~\eqref{eq:Pell-gen} takes the form
	\begin{equation}
	\label{eq:Pell}
	x^2 - D\,y^2 \= 4 \, .
	\end{equation}
	This equation is the so-called Brahmagupta-Pell equation and has been well-studied over the years.\footnote{In the literature, 
	it is common to denote ``the'' Pell equation as the equation 
	where the right-hand side is equal to one. However, the latter is a special case 
	of our equation~\eqref{eq:Pell-gen} with~$n=-1$, see e.g.,~\cite{cohen2008number, Conrad1, Conrad2}.}
	It is one of the classic Diophantine equations, and its solutions have been fully classified.
	In the language of modern algebraic number theory this problem is closely related to the problem of finding units in the ring of 
	integers of the real quadratic field~$\mathbb{Q}(\sqrt{D})$. 
	We will present the solution below in elementary terms, and later make some comments on the more formal interpretation.
	We follow the treatment of~\cite{cohen2008number, Conrad1, Conrad2}.
	The equation~\eqref{eq:Pell} has an infinity of solutions given as follows. 
	Let
	\begin{equation}
	\label{eq:unit}
	u \= u_0 + \sqrt{D}\,v_0 \, ,
	\end{equation}
	be such that~$u_0^2 - D\,v_0^2 = 4$ with the least strictly positive~$v_0$. 
	Then all solutions of~\eqref{eq:Pell} are given by~\cite{cohen2008number} 
	\begin{equation}
	\label{eq:sol-Pell}
	\frac{x + \sqrt{D}\,y}{2} \= \Bigl(\frac{u_0 + \sqrt{D}\,v_0}{2}\Bigr)^k  \quad \text{with} \;\; k \in \mathbb{Z} \, .
	\end{equation}

	In general the difficulty is to find the fundamental solution~$u$, as~$u_0$ and~$v_0$ need 
	not be small even for small~$D$.\footnote{A famous example (Fermat's challenge) is the equation~$a^2 - D\,b^2 =1$ with~$D = 61$, 
	where the fundamental solution is given by~$a_0 = 1766319049$ and~$b_0 = 226153980$. 
	As we see below, this example does not appear in the physical system we study because~$D+4=65$ is not a perfect square. 
	It is interesting to wonder, however, whether 
	such phenomena are relevant to generating large scales in nature. 
	We thank D.~Anninos for this suggestion and for interesting 
	conversations about this issue.} 
	In our case however, we can 
	can use the physics of the problem which guarantees that~$\ell_\gamma = \sqrt{D+4}$ is an integer. Therefore, 
	the fundamental solution is simply
	\begin{equation}
	\label{eq:our-unit}
	u \= \sqrt{D+4} + \sqrt{D} \;\; \Longleftrightarrow \;\; (u_0,v_0) \= (\sqrt{D+4},1) \, .
	\end{equation}
	From this solution we generate all other solutions from~\eqref{eq:sol-Pell}. Expanding that equation and matching 
	the coefficients of unity and~$\sqrt{D}$ leads to the recurrence
	\begin{equation}
	\label{eq:rec-rel-1}
	\begin{split}
	2\,x_{k+1} \=&\; \sqrt{D+4}\,x_k + D\,y_k \, , \\
	2\,y_{k+1} \=&\; x_k + \sqrt{D+4}\,y_k \, ,
	\end{split}
	\end{equation}
	for~$k\geq 0$. Given a solution~$x_k + \sqrt{D}\,y_k$ to~\eqref{eq:Pell}, the matrix~\eqref{eq:dyon-meta-nsol-2} reads
	\begin{equation}
	\label{eq:k-wall}
	\gamma_\pm(k) \=
	\wmat{y_k}{\frac{1}{2}\Bigl(\pm y_k \sqrt{D +4} - x_k\Bigr)}
	{\frac1{2}\bigl(x_k - \ell y_k\bigr)}{\frac{1}{4}\Bigl(x_k\bigl(\ell \pm \sqrt{D+4}\bigr) - y_k\bigl(D \pm \ell\sqrt{D+4}\bigr)\Bigr)} \, .
	\end{equation}
	Using the recursion relations~\eqref{eq:rec-rel-1}, we can now show that acting on the right of~$\gamma_+(k)$ 
	with~$M_1\cdot\tilde{S} = \wmat{\sqrt{D+4}}{1}{-1}{0}$ yields
	\begin{equation}
	\label{eq:plus-gen}
	\gamma_+(k)\cdot M_1\cdot\tilde{S} \= \gamma_+(k+1) \quad \forall \;\; k \geq 0 \, ,
	\end{equation}
	while acting on the right of~$\gamma_-(k)$ with~$M_0\cdot\tilde{S} = \wmat{0}{1}{-1}{\sqrt{D+4}}$ yields 
	\begin{equation}
	\label{eq:minus-gen}
	\gamma_-(k)\cdot M_0\cdot\tilde{S} \= \gamma_-(k+1) \quad \forall \;\; k \geq 0 \, .
	\end{equation}
	In the language of Definition~\ref{def:dyonic-meta}, the recurrence~\eqref{eq:plus-gen} can be written as 
	\begin{equation}
	\label{eq:gen-def-terms}
	\tilde{\gamma}_k = \begin{cases} \gamma_+(k)\cdot\tilde{S} \;\; &\text{for} \;\; k\geq 1 \;\; \text{odd} \\ \gamma_+(k) \;\; &\text{for} \;\; k\geq 2 \;\; \text{even} \end{cases} \, , \quad \text{and} \quad \tilde{\gamma}_0 = \gamma_+(0) \, ,
	\end{equation}
	where we used~$M_0 = \tilde{S}\cdot M_1 \cdot \tilde{S}$. Had we chosen to start 
	the identification with the electric center in Definition~\ref{def:dyonic-meta}, we would have 
	used~\eqref{eq:minus-gen} instead. 
	Equation~\eqref{eq:gen-def-terms} shows that all metamorphic duals of the dyonic-BSM walls~$\gamma_\pm(0)$ 
	are precisely all the solutions to the Brahmagupta-Pell equation~\eqref{eq:Pell}.
	
	The first representative of the orbit (the wall with~$k=0$) is given by 
	\begin{equation}
	\label{eq:wall-start}
	\gamma_\pm(0) \= \wmat{0}{-1}{1}{\frac12\bigl(\ell \pm \sqrt{D+4}\bigr)} \, .
	\end{equation}
	Note that~$\ell^2 = D - 4m \equiv D+4 \, (\text{mod} \, 2)$, 
	which implies that~$\ell \equiv \sqrt{D+4} \, (\text{mod} \, 2)$, so that
	the bottom-right entry of~\eqref{eq:wall-start} is always an integer. 
	Since~$\ell < \sqrt{D +4}$, the wall~$\gamma_-(0)$ is an element of~$\Gamma_S^-$. Therefore, 
	provided that $\sqrt{D+4} \in \mathbb{Z}$,   
	the wall~$\gamma_+(0) \in \Gamma_S^+$ and 
	all the metamorphic duals to be identified according to the BSM prescription are generated by the right action of~$M_1\cdot\tilde{S}$. 
	These dual walls are given by all the solutions to~\eqref{eq:Pell} as in~\eqref{eq:k-wall}. This completely characterizes dyonic-BSM in the 
	case~$n=-1$. Furthermore, since~$0 \leq \ell \leq m$, the first representative of 
	this orbit~$\gamma_+(0)$ clearly has entries bounded by~$0 < \frac12\bigl(\ell + \sqrt{D+4}\bigr) \leq m+1$.   
	
	\item \underline{\textbf{Case 2}}:~$n\geq 1$ \\
	In this case we are interested in the solutions to the so-called generalized Brahmagupta-Pell equation~\eqref{eq:Pell-gen}
	\begin{equation} \label{eq:gen-Pell}
	x^2 - D\,y^2 \= -4n \, .
	\end{equation}
	As before, we have that~$D > 0$ is not a square. Unlike in the~$n=-1$ case, 
	this equation does not necessarily have a solution for general~$D$ and~$n$. However, 
	when there is a solution~$(x_0,y_0)$ 
	then there are infinitely many solutions 
	which are all generated 
	by multiplication with powers of the fundamental unit given in~\eqref{eq:our-unit},
	\begin{equation}
	\label{eq:gen-Pell-sol}
	x + \sqrt{D}\,y \= (x_0 + \sqrt{D}\,y_0)\,\biggl(\frac{\sqrt{D+4} + \sqrt{D}}{2}\biggr)^k  \quad \text{for any} \;\; k \in \mathbb{Z} \, .
	\end{equation}
	By repeating the same steps as in Case 1 above, one can again show that the orbit of metamorphic duals is precisely 
	the solution set of the generalized Pell equation, and is generated by the matrices~\eqref{eq:plus-gen} and~\eqref{eq:minus-gen}
	acting on
	\begin{equation}
	\gamma_\pm(0) \= \wmat{y_0}{\frac12\bigl(\pm y_0 \sqrt{D+4} - x_0\bigr)}{\frac{1}{2n}
	\bigl(\ell y_0 - x_0\bigr)}{\frac{1}{4n}\bigl(\ell y_0 - x_0\bigr)\bigl(\ell \pm \sqrt{D+4}\bigr) - my_0} \, .
	\end{equation}
	As in~\eqref{eq:wall-start}, we have~$\sqrt{D+4} \in \mathbb{Z}$ and~$\ell \equiv \sqrt{D+4} \; (\text{mod} \, 2)$.
	In order for the matrix entries to be integer, 
	a sufficient condition is~$x_0 \equiv \ell y_0 \, (\text{mod} \, 2n)$. 
	By~\eqref{eq:gen-Pell} we have $x_0^2 \, \equiv D \, y_0^2 \, (\text{mod} \, 2n)$.
	Together with the fact that~$D \, \equiv \ell^2 \, (\text{mod} \, 2n)$, 
	this implies 
	that $x_0^2 \, \equiv \ell^2 \, y_0^2 \, (\text{mod} \, 2n)$, so that 
	if~$n$ is square free, then we automatically have $x_0 \equiv \ell y_0 \, (\text{mod} \, 2n)$. 
	Once this condition is met, 
	the full 
	orbit of metamorphic duals is generated by~$M_1\cdot\tilde{S}$ or~$M_0\cdot\tilde{S}$ as before.

	Since~$k$ runs over all integers in~\eqref{eq:gen-Pell-sol}, it is clear that 
	every Pell orbit---and therefore every dyonic BSM orbit---has 
	an element with smallest~$|y|$, which is called the fundamental solution.
	Although there is no existence theorem for solutions to the generalized Brahmagupta-Pell 
	equation~\eqref{eq:gen-Pell} when~$n\geq 1$, there
	is a powerful theorem~\cite{Conrad1,Conrad2} which 
	states that the fundamental solution 
	is bounded according to
	\begin{equation}
	\label{eq:sol-bound}
	x^2 \; \leq \; 2n\bigl(\sqrt{D+4} + \sqrt{D}\bigr) \, , \qquad y^2 \; \leq \; 2n\biggl(\frac{\sqrt{D+4} + \sqrt{D}}{D}\biggr) \, .
	\end{equation}
	These bounds are very restrictive, and in particular, they imply that the set of dyonic-BSM orbits is finite, 
	with a representative whose entries are strictly less than~$m+1$.
\end{enumerate}
	
	We have thus fully characterized the dyonic-BSM walls and explained how the infinite orbit of metamorphic duals 
	defined in Definition~\ref{def:dyonic-meta} is in one-to-one correspondence with the infinite orbit of solutions to 
	the (generalized) Brahmagupta-Pell equation. Crucial to being able to solve the problem was the fact that~$U$-duality 
	fixes~$\ell_\gamma = \sqrt{|\Delta| + 4}$ to be an integer. 

\vspace{0.2cm}

	It is instructive to restate the solutions of the Brahmagupta-Pell equation in the language of algebraic 
	number theory~\cite{cohen2008number, Conrad1, Conrad2}. Consider the real quadratic field~$K=\mathbb{Q}(\sqrt{D})$ 
	where $D>0$ is not a square. We will denote elements of this field either as~$(x,y)$ or as~$x+\sqrt{D} \, y$
	with $x,y \in \mathbb{Q}$. The norm of this element is~$N(x,y)=x^2-Dy^2$. 
	By a change of variables~(\cite{cohen2008number}, p.~355) 
	one can bring the basic Brahmagupta-Pell equation to the form 
	\begin{equation}
	\label{eq:Pell1}
	x^2 - D\,y^2 \= 1 \, .
	\end{equation}
	Thus we are looking for elements of norm~1. By multiplicativity of the norm, it is clear that if~$u = x+\sqrt{D} \,y$ is a 
	solution of~\eqref{eq:Pell1}, then so is~$u^k$ for~$k \in\mathbb{Z}$. (It is easy to check, by rationalizing 
	denominators and using~\eqref{eq:Pell1}, that negative powers are also good solutions.) The problem of finding all solutions 
	to the basic Brahmagupta-Pell equation is then precisely the problem of finding all units in the order~$\IZ[\sqrt{D}]$. 
	Denoting the discriminant of~$K$ as~$D_0$, we have~$D= D_0 f^2$. When~$f=1$ the solution to this problem is given
	by Dirichlet's unit theorem, that all solutions are generated as powers of the fundamental unit~$a_0+\sqrt{D} \, b_0$
	which is the unit with least positive~$b_0$. In fact this statement holds even when~$f>1$ (one can use a proof by induction 
	on the number of prime powers of~$f$). By changing variables back, we obtain the formulation~\eqref{eq:sol-Pell}.
	
	For the general case we have, after the change of variables mentioned above,
	\begin{equation}
	\label{eq:gen-Pell1}
	x^2 - D\,y^2 \= - n \, ,
	\end{equation}
	with~$n \in \IZ$ (our interest in this paper is in~$n \ge -1$ with~$n\neq 0$). In this case, we are looking for elements in~$K$ with norm~$-n$. 
	Once again it is easy to see, by the multiplicativity of the norm, that given one such element~$(x_0,y_0)$ with~$N(x_0,y_0)=-n$
	we have an infinite number of elements with the same norm generated by multiplying~$x_0+\sqrt{D} \, y_0$ by arbitrary powers of a unit. 
	The main theorem in this case says that there are a finite number of fundamental solutions~$(x_0,y_0)$ which lie in the 
	range~$|x_0| \le \sqrt{|n| u}$, $|y_0| \le \sqrt{|n|u/D}$, where~$u$ is any unit satisfying~$u>1$ and~$N(u)=1$. This last condition, translated 
	back to our variables is presented in~\eqref{eq:sol-bound}.

	\subsection*{Summary}
	
	For convenience, we now summarize the results of Sections~\ref{sec:negwomet} and~\ref{sec:metamorphosis} where 
	we have characterized all the walls contributing to the negative discriminant degeneracies~\eqref{eq:degagain}. There 
	are two notable points. 
	
	\vspace{0.1cm}
	
	\ndt \emph{Finiteness.} Examining the various cases (with and without BSM), we see that the set of relevant walls is \emph{finite}, and in fact 
	small in the following sense: it consists of S-walls with entries bounded (in absolute value) from above by~$m+1$, where the  
	upper bound is optimal for certain values of the original charges $(n,\ell,m)$, as evidenced e.g., in \eqref{eq:walls-magBSM-2}. Moreover, 
	all walls are such that~$|p/r| \leq 1$ and~$|q/s| \leq 1$ and so their endpoints lie in the strip~$\Sigma_1 \in [-1,+1]$ in 
	the~$\Sigma$ moduli space. 
	
	\vspace{0.1cm}

	\ndt \emph{Structure.} 
	The structure of walls of electric and magnetic BSM form an orbit generated by the corresponding BSM transformation which 
	acts as~$\IZ/2\IZ$. 
	The dyonic bound state metamorphosis has a very interesting characterization. 
	We already knew that there is an infinite set of different-looking gravitational configurations, all with the same total dyonic 
	charge invariants with negative discriminant, which are related by $U$-duality to each other. 
	The phenomenon of  BSM~\cite{Sen2011,Chowdhury:2012jq} says that these configurations actually should 
	not be considered as distinct physical configurations; rather, they must be identified as different avatars of the same physical entity. 
	Our considerations in this section show that the following sets are in one-to-one correspondence:
	\begin{enumerate}
	\item[1.] The orbit of dyonic metamorphic duals with charges~$(n,\ell,m)$ with~$\ell^2 - 4 m n = D >0$, and 
	\item[2a.] Solutions to the generalized Brahmagupta-Pell equation~$x^2-D \, y^2 = -4n$ 
	with fundamental solution~$(x_0,y_0)$, and the conditions $\sqrt{D+4} \in \mathbb{Z}$ 
	and~$x_0 \equiv \ell\,y_0  \, (\text{mod} \, 2n)$,  or, equivalently, 
	\item[2b.] The set of algebraic integers of norm~$-n$ in the order~$\IZ[\sqrt{D}]$ of the real quadratic field~$K=\mathbb{Q}(\sqrt{D})$
	with~$\frac12(\sqrt{D+4}+\sqrt{D})$ as the fundamental unit, as well as the second congruence condition above. 
	\end{enumerate}
	Moreover, these sets are isomorphic to each other (and to~$\IZ$) as an additive group. The generators of the groups are given, respectively, by
	the generators in Definition~\ref{def:dyonic-meta} (modulo~$\tilde{S}$), and by multiplication in~$K$ by the fundamental unit.

	\section{The exact black hole formula and experimental checks}
	\label{sec:expchecks}
	
	In this section we assemble all the elements of the previous sections into one formula, and 
	then we present checks of this formula.
	So far we have seen that the walls of marginal stability contributing to the polar coefficients according to
	Equation~\eqref{eq:degagain} are a subset of~$PSL(2,\IZ)$. 
	Bound state metamorphosis is an equivalence relation on the set~$PSL(2,\IZ)$ and therefore divides it into orbits~$\mu$.
	We denote the set of orbits as (cf.~Equation~\eqref{defgset})
	\be \label{defgbsm}
	\gbsm \= PSL(2,\IZ)/\text{BSM} \,. 
	\ee
	The orbits are of the following three types:
	\begin{enumerate}
		\item Walls for which there are no metamorphosis. 
		These walls have no duals and therefore lie in an orbit of length 1.
		\item Walls with either electric or magnetic metamorphosis, for which there is exactly one dual 
		with the same contribution to the index. These walls lie in an orbit of length 2.
		\item Walls with dyonic metamorphosis for which there are an infinite number of dual walls. 
		These walls lie in an orbit of infinite length with a group structure isomorphic to~$\IZ$.
	\end{enumerate}
	We have seen that the contribution of an orbit to the index
	is one if all its elements contribute, and zero otherwise. This can be encoded in the following 
	function defined on orbits (recalling Equations~\eqref{eq:thetastep}, \eqref{eq:thetastepR} for the 
	definition of the~$\theta$ function), 
	\be  \label{eq:new-theta}
	\Theta (\mu) \= \prod_{\gamma \in \mu} \ \theta(\gamma, \mathcal R) \,, \qquad \mu \in \gbsm \,,
	\ee
	which can be lifted to a function on the space of walls as (using the same notation)
	\be  \label{eq:new-thetagamma}
	\Theta (\gamma) \= \Theta (\mu) \,, \quad \gamma \in \mu \,.
	\ee
	We now have all the elements to present the full formula for the polar degeneracies~\eqref{eq:polar-def}. 
	in the range~$0 \leq \ell \leq m$ we have:\footnote{Recall from the discussion in Section~\ref{sec:exact-entropy} 
	that the~$\widetilde{c}_m(n,\ell)$ are 
	coefficients of a (mock) Jacobi form of index~$m$, and as such they are a function 
	of~$\Delta = 4mn - \ell^2$ and~$\ell$ mod~$(2m)$. Recall also that the modular properties 
	imply that this can further reduced to~$0\leq \ell \leq m$.}
	\be \label{eq:degwmeta}
	\widetilde{c}_m(n,\ell)  \= \frac12 
	\sum_{\substack{\gamma \in \gbsm}}\,(-1)^{\ell_\gamma+1}\,\Theta(\gamma)\,\vert\ell_\gamma\vert\,d(m_\gamma)\,d(n_\gamma) \, .
	\ee

The sum in the above formula runs over~$\gbsm$ which was defined as a coset of~$PSL(2,\IZ)$ in~\eqref{defgbsm}. 
We can also write the formula so that the sum runs over a smaller set, by using the symmetry of the theory and 
making a choice in obtaining the coset representative. Such a choice makes the formula more explicit and is useful 
for computations. We had already illustrated the idea of two such formulas in our preliminary discussion in Section~\ref{sec:towards}
where we didn't take BSM into account. 
In that case we had a sum over~$PSL(2,\IZ)$ in~\eqref{eq:deg1T} but by using the involution~$\widetilde S$, 
we could equivalently write it as a sum over~$\Gamma_S^+$ as in~\eqref{eq:deg}
with an additional factor of~$\frac12$. 
When BSM is present this discussion needs to be modified. When we have pure electric or pure magnetic BSM, 
the orbits of length~2 discussed in Case~2 above are actually part of a full symmetry orbit of length~4 via the 
following identifications:
\be \label{metaSid}
(n_\gamma,\ell_\gamma,m_\gamma) \= 
(N \neq -1,L>0,-1) \stackrel{\widetilde \gamma_\text{m}}{\longmapsto} 
(N,-L,-1) \stackrel{\widetilde S}{\longmapsto} 
(-1,L,N) \stackrel{\widetilde  \gamma_\text{e}}{\longmapsto} 
( -1,-L,N) \,.
\ee
In particular, the combined symmetry of BSM and~$\widetilde S$ implies an identification of 
two walls in~$\Gamma^+_S$, namely the first and the third of the above sequence.

	By definition, a given wall belongs to one and only one symmetry orbit, and, as we have shown in the previous sections, 
	when~$0\leq \ell \leq m$, every orbit has a non-empty intersection with the set
	\be  \label{Gsmp1}
	\left \lbrace \wmat{p}{q}{r}{s} \subset \Gamma_S^+ \; \bigg\vert \; \abs p , \abs q, \abs r, \abs s \leq m +1  \right \rbrace \,. 
	\ee
	The set~$\gset$ is defined as the set of representative of orbits of BSM combined with~$\widetilde S$ 
	in this finite set having a non-zero value of~$\Theta$.  With this definition we rewrite the 
	degeneracies of negative discriminant states for~$0\leq \ell \leq m$ as
	\be \label{eq:degwmetaW}
	\widetilde{c}_m(n,\ell)  \=  
	\sum_{\substack{\gamma \in \gset}}\,(-1)^{\ell_\gamma+1}\, \vert\ell_\gamma\vert\,d(m_\gamma)\,d(n_\gamma) \, .
	\ee

We now present checks of this formula. 
Table~\ref{tab:examples} lists negative discriminant states for $m=1,\dots, 5$. 
Column~I lists the charge invariants~$(m,n,\ell)$ with discriminant~$\Delta=4mn-\ell^2<0$. 
Note that we have changed the order of the charge invariants here with respect to the rest of the paper.
The organization is as follows: we first list~$m$ which is the index of the mock Jacobi form, followed by~$n$ and~$\ell$.
The range of~$\ell$ is~$m, \dots, 0$ which covers all the cases as explained in Section~\ref{sec:exact-entropy},
and~$n$ runs over all values that produce a negative discriminant with non-zero coefficient for~$\psi^\text{F}_m$. 
Column~II lists the walls~$\gamma \in \gset$ which contribute to the degeneracy of states with these charge invariants.
These walls have been discussed in Sections~\ref{sec:negwomet} and~\ref{sec:metamorphosis}.
The walls in Column II, as stated in \eqref{eq:wallcategory}, are semicircles from $q/s \rightarrow p/r$, 
where $\biggl( \, \begin{matrix} p & q \\ r & s \end{matrix} \, \biggr)$ is a \slz matrix. 
Column~III shows the transformed charges at the wall~$\gamma$. In the~$\gamma$-transformed S-duality frame 
the decay products are~$(Q_\gamma,0)$ and~$(0,P_\gamma)$ with invariants~$(m_\gamma,n_\gamma,\ell_\gamma)$
(cf.~\eqref{eq:chargebreak} and \eqref{eq:transformedcharges}). 
Cases with either~$m_\gamma=-1$ or~$n_\gamma=-1$ correspond, respectively, to magnetic and electric metamorphosis. 
An example is~$(m,n, \ell)=(1,-1,0)$ where we have the walls~$(m_\gamma,n_\gamma, \ell_\gamma)=(-1,0,2)$ and~$(-1,0,-2)$
(with contribution~48) are identified due to magnetic BSM as shown in the table.  
According to the discussion around~\eqref{metaSid}, we also need to identify these walls with~$(0,-1,2)$ and~$(0,-1,-2)$ (which 
we have not explicitly displayed in the table). In a similar manner, we have only displayed pure magnetic, but not pure 
electric BSM phenomena in the table. 
The cases with~$m_\gamma=n_\gamma=-1$ correspond to dyonic metamorphosis, in which 
case an infinite number of walls must be identified (see Section~\ref{sec:dyonic-BSM}). An example is~$(m,n, \ell)=(1,-1,1)$. 
Here we have exhibited four walls corresponding to the 
first two solutions to the Brahmagupta-Pell equation~\eqref{eq:Pell} (the trivial solution with~$p=0$ and the first 
non-trivial one with~$p=1$) 
and their respective first metamorphic duals ($\t\gamma$ built with~$M_1$ in Definition~\ref{def:dyonic-meta}). 
Column~IV is the index contribution of each wall and Column~V is the total index~$\widetilde{c}_m(n,\ell)$ according to our 
formula~\eqref{eq:degwmetaW}. This agrees with with a direct calculation of the polar degeneracies of~$\psi_m^\text{F}(\tau,z)$.

We have run similar checks up to~$m =30$ which includes~$1650$ polar coefficients and find 
perfect agreement between the formula~\eqref{eq:degwmetaW} and the polar coefficients of~$\psi_m^\text{F}$. 
In Appendix~\ref{sec:furtherchecks} we show more values of the total index (corresponding to Columns~I and~V 
of Table~\ref{tab:examples}). 

In~\cite{Murthy:2015zzy}, two of the authors of this paper compared the degeneracies of negative discriminant states 
of~$\psi^\text{F}_m$ for $m = 1, \cdots, 7$ and the result of the formula~\eqref{eq:sugra-wrong-polar} which, 
as explained in Section~\ref{sec:localization}, was found by a combination of physical and mathematical ideas. 
It was observed in that paper that~\eqref{eq:sugra-wrong-polar} agrees with the data from~$\psi^\text{F}_m$ 
for many but not all the cases of negative discriminant states. The main formula of this paper~\eqref{eq:degwmeta}
or, equivalently, \eqref{eq:degwmetaW}, removes these discrepancies completely. 
As an example, the approximate formula~\eqref{eq:sugra-wrong-polar} gives~$6400$ for the charges~$(m,n, \ell)=(3,-1,0)$,
while the correct answer is~$6404$. In the table we explicitly see the correction of~$4$ coming from the subleading contribution 
of dyonic-BSM walls.

A natural extension of this technique is to extend it to the cases of~$\mathbb Z_N$ CHL orbifolds. 
We have preliminary data which we hope to analyze in the near future, thereby generalizing the formula 
presented here. Here, the comparison against coefficients of the inverse of the orbifolded Igusa cusp form, $\Phi_k$ 
would involve the Rademacher technique for mock Jacobi forms on congruence subgroups of~$SL(2,\Z)$. 
This would be a generalization of the analysis for \hbps black holes under CHL orbifolds as studied 
in~\cite{Nally:2018raf} and the techniques studied in~\cite{Duncan:2009sq, Cheng:2012qc}.  

\vspace{1cm}

\setlength{\extrarowheight}{0.08cm}

\begin{longtable}{|C|C|C|C|C|}
	\hline
	\text{\textbf{I.}} & \text{\textbf{II.}}  & \text{\textbf{III.}}
	&  \text{\textbf{IV.}} & \text{\textbf{V.}}  \\
	\text{\textbf{Charges}} & \text{\textbf{Walls }}  & \text{\textbf{Transf.~charges}}
	&  \text{\textbf{Contribution}} & \text{\textbf{Net Index}}  \\
	(m , n, \ell\,; \Delta) & \gamma = \hormat{p}{q}{r}{s} 
	&  (m_\gamma, \ n_\gamma, \ \ell_\gamma )  & \text{\textbf{from wall}} &  \widetilde{c}_m(n,\ell) \\
	\hline
	\multirow{5}*{$(1,-1,1\,;-5)$} & \hormat{0}{-1}{1}{2} & (-1,\ -1,\ 3) & \multirow{5}*{3} & \multirow{5}*{3} \\
	& \hormat{0}{-1}{1}{-1} & (-1,\ -1,\ -3) & & \\
	& \hormat{1}{0}{1}{1} & (-1,\ -1,\ 3) & & \\
	& \hormat{1}{-3}{1}{-2} & (-1,\ -1,\ -3) & & \\[-1.5mm]
	& \vdots & \vdots & & \\
	\hline
	\multirow{2}*{$(1,-1,0\,;-4)$} & \hormat{0}{-1}{1}{1} & (-1,\ 0, \ 2) & \multirow{2}*{48} & \multirow{2}*{48} \\
	& \hormat{0}{-1}{1}{-1} & (-1,\ 0, \ -2) & & \\
	\hline
	\multirow{3}*{$(1, 0, 1\,; -1)$} & \hormat{1}{0}{1}{1} & ( 0, \, 0, \ 1) & 576 & \multirow{3}*{600} \\
	\cline{2-4}
	& \hormat{1}{0}{2}{1} & (-1, \, 0, \ 1) & \multirow{2}*{24} & \\
	& \hormat{1}{-1}{2}{-1} & (-1, \, 0, \ -1) & &  \\
	\hline
	\multirow{5}*{$(2,-1,2\,;-12)$} & \hormat{0}{-1}{1}{3} & (-1, \ -1, \ 4) & \multirow{5}*{4} & \multirow{5}*{4} \\
	& \hormat{0}{-1}{1}{-1} & (-1, \ -1, \ -4)  & & \\
	& \hormat{1}{0}{1}{1} & (-1, \ -1, \ 4) & & \\
	& \hormat{1}{-4}{1}{-3} & (-1, \ -1, \ -4)  & & \\[-1.5mm]
	& \vdots & \vdots & & \\
	\hline
	\multirow{2}*{$(2,-1,1\,;-9)$} & \hormat{0}{-1}{1}{2} & (-1, \ 0, \ 3) & \multirow{2}*{72} & \multirow{2}*{72} \\
	& \hormat{0}{-1}{1}{-1} & (-1, \ 0, \ -3) & & \\
	\hline
	\multirow{2}*{$(2,-1,0\,;-8)$} & \hormat{0}{-1}{1}{1} & (-1, \ 1, \ 2) & \multirow{2}*{648} & \multirow{2}*{648} \\
	& \hormat{0}{-1}{1}{-1} & (-1, \ 1, \ -2) & & \\
	\hline
	(2,0,2\,;-4) & \hormat{1}{0}{1}{1} & (0, \ 0, \ 2) & 1152 & 1152 \\
	\hline
	\multirow{4}*{$(2, 0, 1\,; -1)$} & \hormat{1}{0}{1}{1} & ( 1, \, 0, \ 1) & 7776 & \multirow{4}*{8376} \\
	\cline{2-4}
	& \hormat{1}{0}{2}{1} & (0, \, 0, \ 1) & 576 & \\
	\cline{2-4}
	& \hormat{1}{0}{3}{1} & (-1, \, 0, \ 1) & \multirow{2}*{24} &  \\
	& \hormat{1}{-1}{3}{-2} & (-1, \, 0, \ -1) & &  \\
	\hline
	\multirow{5}*{$(3,-1,3\,;-21)$} & \hormat{0}{-1}{1}{4} & (-1,\ -1,\ 5) & \multirow{5}*{5} & \multirow{5}*{5} \\
	& \hormat{0}{-1}{1}{-1} & (-1,\ -1,\ -5) & & \\
	& \hormat{1}{0}{1}{1} & (-1,\ -1,\ 5) & & \\
	& \hormat{1}{-5}{1}{-4} & (-1,\ -1,\ -5) & & \\[-1.5mm]
	& \vdots & \vdots & & \\
	\hline
	\multirow{2}*{$(3,-1,2\,;-16)$} & \hormat{0}{-1}{1}{3} & (-1, \ 0,\ 4) & \multirow{2}*{96} & \multirow{2}*{96} \\
	& \hormat{0}{-1}{1}{-1} & (-1, \ 0,\ -4) & & \\
	\hline \pagebreak \hline
	\multirow{2}*{$(3,-1,1\,;-13)$} & \hormat{0}{-1}{1}{2} & (-1, \ 1,\ 3) & \multirow{2}*{972} & \multirow{2}*{972} \\
	& \hormat{0}{-1}{1}{-1} & (-1, \ 1,\ -3) & & \\
	\hline
	\multirow{7}*{$(3, -1, 0\,; -12)$} & \hormat{0}{-1}{1}{1} & ( -1, \, 2, \ 2) & \multirow{2}*{6400} & \multirow{7}*{6404}  \\
	& \hormat{0}{-1}{1}{-1} &  ( -1, \, 2, \ -2) & & \\
	\cline{2-4}
	& \hormat{0}{-1}{1}{2} & (-1, \, -1, \ 4) & \multirow{5}*{4} & \\
	& \hormat{0}{-1}{1}{-2} & (-1,\, -1, \ -4) & & \\
	& \hormat{1}{0}{2}{1} & (-1, \, -1, \ 4) & &  \\
	& \hormat{1}{-4}{2}{-7} & (-1, \, -1, \ -4) & &  \\[-1.5mm]
	& \vdots & \vdots & & \\
	\hline
	(3,0,3\,;-9) & \hormat{1}{0}{1}{1} & (0, \ 0, \, 3) & 1728 & 1728 \\
	\hline
	\multirow{3}*{$(3,0,2\,;-4)$} & \hormat{1}{0}{1}{1} & (1, \ 0, \ 2) & 15552 & \multirow{3}*{15600} \\
	\cline{2-4}
	& \hormat{1}{0}{2}{1} & (-1, \ 0, \ 2) & \multirow{2}*{48} & \\
	& \hormat{1}{-2}{2}{-3} & (-1, \ 0, \ -2) & & \\
	\hline
	\multirow{5}*{$(3, 0, 1\,; -1)$} & \hormat{1}{0}{1}{1} & ( 2, \, 0, \ 1) & 76800 & \multirow{5}*{85176}  \\
	\cline{2-4}
	& \hormat{1}{0}{2}{1} & (1, \, 0, \ 1) & 7776 &  \\
	\cline{2-4}
	& \hormat{1}{0}{3}{1} & (0, \, 0, \ 1) & 576 & \\
	\cline{2-4}
	& \hormat{1}{0}{4}{1} & (-1, \, 0, \ 1) & \multirow{2}*{24} &  \\
	& \hormat{1}{-1}{4}{-3} & (-1, \, 0, \ -1) & &  \\
	\hline
	\multirow{5}*{$(4,-1,4\,;-32)$} & \hormat{0}{-1}{1}{5} & (-1,\ -1,\ 6) & \multirow{5}*{6} & \multirow{5}*{6} \\
	& \hormat{0}{-1}{1}{-1} & (-1,\ -1,\ -6) & & \\
	& \hormat{1}{0}{1}{11} & (-1,\ -1,\ 6) & & \\
	& \hormat{1}{-6}{1}{-5} & (-1,\ -1,\ -6) & & \\[-1.5mm]
	& \vdots & \vdots & & \\
	\hline
	\multirow{2}*{$(4,-1,3\,;-25)$} & \hormat{0}{-1}{1}{4} & (-1, \ 0,\ 5) & \multirow{2}*{120} & \multirow{2}*{120} \\
	& \hormat{0}{-1}{1}{-1} & (-1, \ 0,\ -5) & & \\
	\hline
	\multirow{2}*{$(4,-1,2\,;-20)$} & \hormat{0}{-1}{1}{3} & (-1, \ 1,\ 4) & \multirow{2}*{1296} & \multirow{2}*{1296} \\
	& \hormat{0}{-1}{1}{-1} & (-1, \ 1,\ -4) & & \\
	\hline
	\multirow{2}*{$(4,-1,1\,;-25)$} & \hormat{0}{-1}{1}{2} & (-1, \ 2,\ 3) & \multirow{2}*{9600} & \multirow{2}*{9600} \\
	& \hormat{0}{-1}{1}{-1} & (-1, \ 2,\ -3) & & \\
	\hline
	\multirow{4}*{$(4, -1, 0\,; -16)$} & \hormat{0}{-1}{1}{1} & (-1, \, 3, \ 2) & \multirow{2}*{51300} & \multirow{4}*{51396}  \\
	& \hormat{0}{-1}{1}{-1} & (-1, \, 3, \ -2) & & \\
	\cline{2-4}
	& \hormat{0}{-1}{1}{2} & (-1, \ 0, \ 4) & \multirow{2}*{96} &  \\
	& \hormat{0}{-1}{1}{-2} & (-1, \ 0, \ -4) & &  \\
	\hline
	(4,0,4\,;-16) & \hormat{1}{0}{1}{1} & (0, \ 0, \ 4) & 2304 & 2304 \\
	\hline
	(4,0,3\,;-9) & \hormat{1}{0}{1}{1} & (1, \ 0, \ 3) & 23328 & 23328 \\
	\hline
	\multirow{2}*{$(4,0,2\,;-4)$} & \hormat{1}{0}{1}{1} & (2, \ 0, \ 2) & 153600 & \multirow{2}*{154752} \\
	\cline{2-4}
	& \hormat{1}{0}{2}{1} & (0, \ 0, \ 2) & 1152 & \\
	\hline
	\multirow{6}*{$(4, 0,1\,; -1)$} & \hormat{1}{0}{1}{1} & (3, \ 0, \ 1) & 615600 & \multirow{6}*{700776}  \\
	\cline{2-4}
	& \hormat{1}{0}{2}{1} & (2, \ 0, \ 1) & 76800 &  \\
	\cline{2-4}
	& \hormat{1}{0}{3}{1} & (1, \ 0, \ 1) & 7776 & \\
	\cline{2-4}
	& \hormat{1}{0}{4}{1} & (0, \ 0, \, 1) & 576 &  \\
	\cline{2-4}
	& \hormat{1}{0}{5}{1} & (-1, \, 0, \ 1) & \multirow{2}*{24} &  \\
	& \hormat{1}{-1}{5}{-4} & (-1, \, 0, \ -1) & &  \\
	\hline
	\multirow{5}*{$(5,-1,5\,;-45)$} & \hormat{0}{-1}{1}{6} & (-1,\ -1,\ 7) & \multirow{5}*{7} & \multirow{5}*{7} \\
	& \hormat{0}{-1}{1}{-1} & (-1,\ -1,\ -7) & & \\
	& \hormat{1}{0}{1}{1} & (-1,\ -1,\ 7) & & \\
	& \hormat{1}{-7}{1}{-6} & (-1,\ -1,\ -7) & & \\[-1.5mm]
	& \vdots & \vdots & & \\
	\hline
	\multirow{2}*{$(5,-1,4\,;-36)$} & \hormat{0}{-1}{1}{5} & (-1, \ 0,\ 6) & \multirow{2}*{144} & \multirow{2}*{144} \\
	& \hormat{0}{-1}{1}{-1} & (-1, \ 0,\ -6) & & \\
	\hline
	\multirow{2}*{$(5,-1,3\,;-29)$} & \hormat{0}{-1}{1}{4} & (-1, \ 1,\ 5) & \multirow{2}*{1620} & \multirow{2}*{1620} \\
	& \hormat{0}{-1}{1}{-1} & (-1, \ 1,\ -5) & & \\
	\hline
	\multirow{2}*{$(5,-1,2\,;-24)$} & \hormat{0}{-1}{1}{3} & (-1, \ 2,\ 4) & \multirow{2}*{12800} & \multirow{2}*{12800} \\
	& \hormat{0}{-1}{1}{-1} & (-1, \ 2,\ -4) & & \\
	\hline
	\multirow{7}*{$(5,-1,1\,;-21)$} & \hormat{0}{-1}{1}{2} & (-1, \ 3, \ 3) & \multirow{2}*{76950} & \multirow{7}*{76955} \\
	& \hormat{0}{-1}{1}{-1} & (-1, \ 3,\ -3) & & \\
	\cline{2-4}
	& \hormat{0}{-1}{1}{3} & (-1, \ -1, \ 5) & \multirow{5}*{5} & \\
	& \hormat{0}{-1}{1}{-2} & (-1, \ -1, \ -5) & & \\
	& \hormat{1}{0}{2}{1} & (-1, \ -1, \ 5) & & \\
	& \hormat{1}{-5}{2}{-9} & (-1, \ -1, \ -5) & & \\[-1.5mm]
	& \vdots & \vdots & & \\
	\hline
	\multirow{4}*{$(5,-1,0\,;-20)$} & \hormat{0}{-1}{1}{1} & (-1, \ 4, \ 2) & \multirow{2}*{352512} & \multirow{4}*{353808} \\
	& \hormat{0}{-1}{1}{-1} & (-1, \ 4, \ -2) & & \\
	\cline{2-4}
	& \hormat{0}{-1}{1}{2} & (-1, \ 1, \ 4) & \multirow{2}*{1296} & \\
	& \hormat{0}{-1}{1}{-2} & (-1, \ 1, \ -4) & & \\
	\hline
	(5,0,5\,;-25) & \hormat{1}{0}{1}{1} & (0,\ 0,\ 5) & 2880 & 2880 \\
	\hline
	(5,0,4\,;-16) & \hormat{1}{0}{1}{1} & (1,\ 0,\ 4) & 31104 & 31104 \\
	\hline \pagebreak \hline
	\multirow{3}*{$(5,0,3\,;-9)$} & \hormat{1}{0}{1}{1} & (2,\ 0,\ 3) & 230400 & \multirow{3}*{230472} \\
	\cline{2-4}
	& \hormat{1}{0}{2}{1} & (-1,\ 0,\ 3) & \multirow{2}*{72} & \\
	& \hormat{1}{-3}{2}{-5} & (-1,\ 0,\ -3) & & \\
	\hline	
	\multirow{4}*{$(5,0,2\,;-4)$} & \hormat{1}{0}{1}{1} & (3,\ 0,\ 2) & 1231200 & \multirow{4}*{1246800} \\
	\cline{2-4}
	& \hormat{1}{0}{2}{1} & (1,\ 0,\ 2) & 15552 & \\
	\cline{2-4}
	& \hormat{1}{0}{3}{1} & (-1,\ 0,\ 2) & \multirow{2}*{48} & \\
	& \hormat{1}{-2}{3}{-5} & (-1,\ 0,\ -2) & & \\
	\hline
	\multirow{7}*{$(5,0,1\,; -1)$} & \hormat{1}{0}{1}{1} & (4, \ 0, \ 1) & 4230144 & \multirow{7}*{4930920}  \\
	\cline{2-4}
	& \hormat{1}{0}{2}{1} & (3, \ 0, \ 1) & 615600 &  \\
	\cline{2-4}
	& \hormat{1}{0}{3}{1} & (2, \, 0, \ 1) & 76800 &  \\
	\cline{2-4}
	& \hormat{1}{0}{4}{1} & (1, \, 0, \ 1) & 7776 & \\
	\cline{2-4}
	& \hormat{1}{0}{5}{1} & (0, \, 0, \ 1) & 576 & \\
	\cline{2-4}
	& \hormat{1}{0}{6}{1} & (-1, \, 0, \ 1) & \multirow{2}*{24} &  \\
	& \hormat{1}{-1}{6}{-5} & (-1, \, 0, \ -1) & &  \\
	\hline
	\multirow{8}*{$(5,1,5\,; -5)$} & \hormat{1}{0}{1}{1} & (1, \ 1, \ 3) & 314928 & \multirow{8}*{315255} \\
	\cline{2-4}
	& \hormat{1}{0}{2}{1} & (-1, \ 1, \ 1)& \multirow{2}*{324} &  \\
	& \hormat{1}{-1}{2}{-1}& (-1, \ 1, \ -1) & &   \\
	\cline{2-4}
	& \hormat{1}{1}{2}{3} & (-1, \ -1, \ 3) & \multirow{5}*{3} & \\
	&\hormat{1}{-2}{2}{-3} & (-1 \ -1, \ -3) & & \\
	& \hormat{2}{1}{3}{2} & (-1, \ -1, \ 3) &  &\\
	& \hormat{2}{-5}{3}{-7} & (-1, \ -1, \ -3)  &  & \\[-1.5mm]
	& \vdots & \vdots & & \\
	\hline
	\caption{Table of examples detailing original charge vector, contributing walls, associated charge breakdowns 
		at walls and index contributions.}
	\label{tab:examples}
\end{longtable}

\section*{Acknowledgements}

We would like to thank Gabriel Cardoso, Alejandra Castro, Atish Dabholkar, Justin David, Shamit Kachru, Richard Nally, 
Suresh Nampuri, James Newton, Brandon Rayhaun, Ashoke Sen, Dimitri Skliros, Alberto Zaffaroni, and especially 
Francesca Ferrari for valuable discussions. The work of A.~C., A.~K. and T.~W. is supported by the Austrian Science 
Fund (FWF): P 285552 and a Scientific \& Technological Cooperation between Austria and India: Project No.~IN 27/2018. 
The work of S.~M.~is supported by the ERC Consolidator Grant N.~681908, ``Quantum black holes: A microscopic 
window into the microstructure of gravity'', and by the STFC grant ST/P000258/1. The work of V.~R.~is supported in 
part by INFN, the KU Leuven C1 grant ZKD1118 C16/16/005, and by the ERC Starting Grant 637844-HBQFTNCER. 
The work of A.~K. is further supported by  the FWF project: W1252-N27, and the Austrian Marshall Plan fellowship.

\ndt A.~C., A.~K., S.~M., V.~R. acknowledge the \textit{School and Workshop on Supersymmetric Localization and 
Holography: Black Hole Entropy and Wilson Loops} held in 2018 at ICTP, Trieste for creating a stimulating environment 
where this project was first discussed. Part of the research was done when the authors attended a workshop on 
moonshine at the Erwin-Schr\"odinger-Institut (ESI) in Vienna in September 2018. We thank the ESI for support 
and providing a stimulating work environment. A.~K. and S.~M. also thank the workshop on ``\textit{Number theory 
and quantum physics}" organized by the Bethe Centre for Theoretical Physics, Uni.~Bonn, for creating a stimulating 
environment where part of this work was done.

\appendix

\section{(Mock) Jacobi forms and the Rademacher expansion}
\label{app:Jac}

A Jacobi form \cite{eichler1985theory} of weight~$w$ and index~$m$ with respect to the fundamental modular group~$SL(2,\mathbb{Z})$ is a holomorphic function~$\varphi(\tau,z):\mathbb{H}\times \mathbb{C} \rightarrow \mathbb{C}$ (where~$\mathbb{H}$ is the upper half-plane) which satisfies two functional equations,
\begin{align}
\label{eq:app-mod-transform}  
\varphi\Bigl(\frac{a\tau+b}{c\tau+d},\frac{z}{c\tau+d}\Bigr) \= 
&\, (c\tau+d)^w\,e^{\frac{2\pi \i mc z^2}{c\tau+d}}\,\varphi(\tau,z)  \qquad \forall 
\quad \wmat{a}{b}{c}{d} \in SL(2,\mathbb{Z}) \, , \\[2mm]
\label{eq:app-elliptic-transfo}
\varphi(\tau, z+\lambda\tau+\mu) \=&\, e^{-2\pi \i m(\lambda^2 \tau + 2 \lambda z)} \varphi(\tau, z) 
\qquad \quad \; \forall \quad \lambda,\,\mu \in \mathbb{Z} \, . 
\vspace{.3cm}
\end{align}
Due to the periodicity properties encoded in the above equations, $\varphi(\tau,z)$ has a Fourier expansion 
\begin{equation}
\varphi(\tau,z) \= \sum_{n,\ell\in\mathbb{Z}}\,c(n,\ell)\,q^n\,\zeta^\ell \, ,
\end{equation}
where~$q\!\!\defeq\!\!e^{2\pi \i \tau}$ and~$\zeta\!\!\defeq\!\!e^{2\pi \i z}$.
Owing to \eqref{eq:app-elliptic-transfo}, a Jacobi form of weight $w$ and index $m$ can be 
decomposed into a vector-valued modular form of weight $w-1/2$ via its theta-decomposition
\begin{equation}
\label{eq:app-theta-decomp}
\varphi(\tau,z) \= \sum_{\ell \in \mathbb{Z}/2m\mathbb{Z}} h_\ell(\tau)\,\vartheta_{m,\ell}(\tau,z) \, ,
\end{equation} 
where the components $h_\ell(\tau)$ take the form 
\begin{equation}
h_\ell(\tau) \= \sum_\Delta\,c(n,\ell)\,q^{\Delta/4m} \, , \quad \Delta = 4mn - \ell^2 \, .
\end{equation}
The $\vartheta_{m,\ell}(\tau,z)$ denote the standard weight 1/2, index $m$ theta functions, 
\begin{equation}
\label{eq:app-vartheta}
\vartheta_{m,\ell}(\tau,z) \defeq \sum\limits_{\substack{r\in\mathbb{Z}\\ r\equiv \ell \, \text{mod}\, 2m}}\,q^{r^2/4m}\,\zeta^{r} \, .
\end{equation}

The Rademacher expansion provides a powerful tool to reconstruct the Fourier coefficients of Jacobi forms.
We illustrate it here for weights $w+1/2$ smaller or equal to zero, modular group $SL(2,\mathbb{Z})$ and generic multiplier system $\psi(\gamma)$.  
Once the modular properties of the modular forms $h_\ell(\tau)$ are known, the only extra ingredient required to determine the Fourier coefficients $c(n,\ell)$ with $\Delta \geq 0$ are the \emph{polar coefficients}, i.e. the terms with negative powers of $q$ in the Fourier expansion
\begin{equation}
h_\ell(\tau) \= \sum_{\widetilde{\Delta} < 0} c(\tilde{n},\tilde{\ell})\,q^{\widetilde{\Delta}/4m} +\sum_{{\Delta} \geq 0} c(n,\ell)\,q^{{\Delta}/4m}  \, .
\end{equation}
In turn, the Rademacher expansion for the Fourier coefficients of $h_\ell(\tau)$ takes the form
\begin{equation}
\label{eq:app-Rad}
c(n,\ell) \= 2\pi\,\sum_{k=1}^{\infty}\,\sum_{\substack{\widetilde{\ell}\in\mathbb{Z}/2m\mathbb{Z} \\ \widetilde{\Delta} < 0}}\,c(\tilde{n},\tilde{\ell})\,\frac{Kl\bigl(\frac{\Delta}{4m},\frac{\widetilde{\Delta}}{4m}\,;k,\psi\bigr)_{\ell\widetilde{\ell}}}{k}\,\Bigl(\frac{|\widetilde{\Delta}|}{\Delta}\Bigr)^{\frac{1-w}{2}}\,I_{1-w}\Bigl(\frac{\pi}{mk}\sqrt{|\widetilde{\Delta}|\Delta}\Bigr) \, .
\end{equation}
Here, $I_\rho(x)$ is the I-Bessel function of weight $\rho$, which has the following integral representation for $x \in \mathbb{R}^*$,
\begin{equation}
\label{eq:app-Bessel}
I_\rho(x) \= \frac{1}{2\pi\mathrm{i}}\,\Bigl(\frac{x}{2}\Bigr)^\rho\,\int_{\epsilon-\mathrm{i}\infty}^{\epsilon+\mathrm{i}\infty}\,t^{-\rho-1}\,e^{\,t+\tfrac{x^2}{4t}}\,\mathrm{d}t \, ,
\end{equation}
and asymptotics
\begin{equation}
\label{eq:app-Bessel-asymptotic}
I_{\rho}(x) \underset{x \rightarrow \infty}{\sim} \frac{e^x}{\sqrt{2\pi x}}\Bigl(1 - \frac{\mu - 1}{8x} + \frac{(\mu - 1)(\mu - 3^2)}{2!(8x)^3} - \frac{(\mu - 1)(\mu - 3^2)(\mu - 5^2)}{3!(8x)^5} + \ldots\Bigr) \, ,
\end{equation}
with $\mu = 4\rho^2$. In \eqref{eq:app-Rad}, $Kl\bigl(\tfrac{\Delta}{4m},\tfrac{\widetilde{\Delta}}{4m}\,;k,\psi)_{\ell\widetilde{\ell}}\;$ is the generalized Kloosterman sum
\begin{equation}
\label{eq:app-Klooster}
Kl(\mu,\nu\,;k,\psi)_{\ell\widetilde{\ell}} \defeq \sum_{\substack{0 \leq h< k \\ (h,k) = 1}}\,e^{2\pi \i\,\bigl(-\tfrac{h}{k}\mu+\tfrac{h'}{k}\nu \bigr)}\;\psi(\gamma)_{\ell\widetilde{\ell}} \, ,
\end{equation}
with $\gamma = \begin{pmatrix} h' & -\tfrac{hh'+1}{k} \\ k & -h \end{pmatrix} \in SL(2,\mathbb{Z})$ and $hh' \equiv -1$ (mod $k$).\\

There exists a generalization of the Rademacher expansion applicable to cases where the function $\varphi$ is a \emph{mock} Jacobi form~\cite{bringmann2006f,Bringmann:2010sd,bringmann2011extension,bringmann2012coefficients}. As discussed in the main text, the function that is relevant to our story is the mock Jacobi form $\psi_m^\mathrm{F}$. In this case, the generalized Rademacher expansion for the Fourier coefficients $c_m^\mathrm{F}(n,\ell)$, $\Delta \geq 0$, was obtained in \cite{Ferrari2017a}. It reads:
\begin{align}
\label{eq:app-mixed-mock-coeffs}
c^{\mathrm{F}}_m(n,\ell) \=&\, 2\pi\,\sum_{k=1}^\infty\,\sum_{\substack{\widetilde{\ell} \in \mathbb{Z}/2m\mathbb{Z} \\ 4m\widetilde{n} - \widetilde{\ell}^2 < 0}}\,c^{\mathrm{F}}_m(\widetilde{n},\widetilde{\ell})\,\frac{Kl\bigl(\frac{\Delta}{4m},\frac{\widetilde{\Delta}}{4m}\,;k,\psi\bigr)_{\ell\widetilde{\ell}}}{k}\,\biggl(\frac{|\widetilde{\Delta}|}{\Delta}\biggr)^{23/4}\,I_{23/2}\biggl(\frac{\pi}{m k}\sqrt{|\widetilde{\Delta}|\Delta}\biggr) \nonumber \\[1mm]
&\,+ \sqrt{2m}\,\sum_{k=1}^\infty\,\frac{Kl\bigl(\frac{\Delta}{4m},-1\,;k,\psi\bigr)_{\ell 0}}{\sqrt{k}}\,\biggl(\frac{4m}{\Delta}\biggr)^6\,I_{12}\biggl(\frac{2\pi}{k\sqrt{m}}\sqrt{\Delta}\biggr) \\[1mm]
&\,-\frac{1}{2\pi}\,\sum_{k=1}^\infty\,\sum_{\substack{j\in\mathbb{Z}/2m\mathbb{Z} \\ g \in \mathbb{Z}/2mk\mathbb{Z} \\ g \equiv j (\text{mod }2m)}} \frac{Kl\bigl(\frac{\Delta}{4m},-1-\frac{g^2}{4m}\,;k,\psi\bigr)_{\ell j}}{k^2}\,\biggl(\frac{4m}{\Delta}\biggr)^{25/4}\,\times \nonumber \\
&\qquad\quad\;\; \times \int_{-1/\sqrt{m}}^{+1/\sqrt{m}}\,f_{k,g,m}(u)\,I_{25/2}\Biggl(\frac{2\pi}{k\sqrt{m}}\sqrt{\Delta(1-m u^2)}\Biggr)(1-m u^2)^{25/4}\,\mathrm{d}u \, , \nonumber
\end{align}
where the multiplier system $\psi(\gamma)$ is given explicitly in \cite{Ferrari2017a} in terms of the (known) multiplier system of the Jacobi theta functions \eqref{eq:app-vartheta}, and the function $f_{k,g,m}$ in the last line is given by
\begin{equation}
\label{eq:app-fkgmu}
f_{k,g,m}(u) \defeq \begin{cases} \displaystyle\frac{\pi^2}{\text{sinh}^2(\frac{\pi u}{k}-\frac{\pi\mathrm{i}g}{2mk})} & \text{if} \, g\not\equiv 0 \,(\text{mod } 2mk) \, , \\[4mm] \displaystyle\frac{\pi^2}{\text{sinh}^2(\frac{\pi u}{k})}-\frac{k^2}{u^2} & \text{if} \, g\equiv 0 \,(\text{mod } 2mk) \, .
\end{cases}
\end{equation}
The last two terms in \eqref{eq:app-mixed-mock-coeffs} arise due to the mock modular nature of $\psi_m^{\mathrm{F}}$. Note that these terms have no free parameter, which is a consequence of the fact that the shadow of~$\psi_m^{\mathrm{F}}$ has a single polar coefficient equal to one. The first line is the standard Rademacher expansion for a Jacobi form of weight $-10$ and index $m$. Although the above formula may appear daunting, its main feature is that the coefficients $c_m^\mathrm{F}(n,\ell)$ for $\Delta \geq 0$ (the left-hand side) are completely determined by the polar coefficients $c_m^\mathrm{F}(n,\ell)$ for $\Delta < 0$, and the modular properties of $\psi_m^\mathrm{F}$ such as its weight, index, and multiplier system.

\section{Checks of finiteness of the set~$\gset$}
\label{sec:finiteness}

In this appendix, we present the magnetic-BSM walls of Section~\ref{sec:mag-BSM}, Case 3 that have~$r=r_+$. As explained there, we were not able to derive an analytic upper bound on their~$p$ entry. However, we were able to check that for a given charge vector~$(n,\ell,m)$ only a single\footnote{As explained in the main text, this wall comes with its metamorphic dual, whose contribution to the index gets identified.} wall contributes to the polar coefficients. Below we give the explicit form of these walls for all charge vectors with~$m\leq 30$ where they exist. They are built as~$PSL(2,\mathbb{Z})$ matrices with~$p,s$ entries consistent with~\eqref{eq:p-lower} and~\eqref{eq:s-window}, together with a numerical upper bound~$p \leq 10^6$. The associated transformed charges are given in the third column. Lastly, we also display the corresponding electric-BSM walls constructed following the procedure outlined in Sections~\ref{sec:el-BSM}  and~\ref{sec:expchecks}, together with their transformed charges.

\begin{longtable}{|C|C|C|C|C|}
	\hline
	\text{\textbf{Charges}} & \text{\textbf{ Mag. Walls }}  & \text{\textbf{Mag.~charges}} & \text{\textbf{ Elec. Walls }} & \text{\textbf{Elec.~charges}} \\
	(m , n, \ell\,; \Delta) & \gamma = \hormat{p}{q}{r}{s} &  (m_\gamma, \ n_\gamma, \ \ell_\gamma ) & \gamma = \hormat{p}{q}{r}{s} &  (m_\gamma, \ n_\gamma, \ \ell_\gamma ) \\
	\hline
	(10,2,9\,;-1) & \hormat{3}{2}{7}{5} & (-1,\ 0,\ 1) & \hormat{1}{3}{2}{7} & (0,\ -1,\ 1) \\
	\hline
	(14,2,11\,;-9) & \hormat{1}{2}{3}{7} & (-1,\ 0,\ 3) & \hormat{1}{1}{2}{3} & (0,\ -1,\ 3) \\
	(14,3,13\,;-1) & \hormat{4}{3}{9}{7} & (-1,\ 0,\ 1) & \hormat{1}{4}{2}{9} & (0,\ -1,\ 1) \\
	\hline
	(16,3,14\,;-4) & \hormat{2}{3}{5}{8} & (-1,\ 0,\ 2) & \hormat{1}{2}{2}{5} & (0,\ -1,\ 2) \\
	\hline
	(18,4,17\,;-1) & \hormat{5}{4}{11}{9} & (-1,\ 0,\ 1) & \hormat{1}{5}{2}{11} & (0,\ -1,\ 1) \\
	\hline
	(20,3,16\,;-16) & \hormat{1}{3}{3}{10} & (-1,\ 0,\ 4) & \hormat{1}{1}{2}{3} & (0,\ -1,\ 4) \\
	\hline
	(21,2,13\,;-1) & \hormat{3}{2}{10}{7} & (-1,\ 0,\ 1) & \hormat{1}{3}{3}{10} & (0,\ -1,\ 1) \\
	\hline
	(22,5,21\,;-1) & \hormat{6}{5}{13}{11} &  (-1,\ 0,\ 1) & \hormat{1}{6}{2}{13} & (0,\ -1,\ 1) \\
	\hline
	(23,3,17\,;-13) & \hormat{1}{2}{3}{7} & (-1,\ 1,\ 3) & \hormat{1}{1}{2}{3} &  (1,\ -1,\ 3) \\
	\hline
	(24,5,22\,;-4) & \hormat{3}{5}{7}{12} & (-1,\ 0,\ 2) & \hormat{1}{3}{2}{7} &  (0,\ -1,\ 2) \\
	\hline
	(26,4,21\,;-25) & \hormat{1}{4}{3}{13} & (-1,\ 0,\ 5) & \hormat{1}{1}{2}{3} & (0,\ -1,\ 5)  \\
	(26,5,23\,;-9) & \hormat{2}{5}{5}{13} & (-1,\ 0,\ 3) & \hormat{1}{2}{2}{5} & (0,\ -1,\ 3) \\
	(26,6,25\,;-1) & \hormat{7}{6}{15}{13} & (-1,\ 0,\ 1) & \hormat{1}{7}{2}{15} & (0,\ -1,\ 1) \\
	\hline
	(27,2,15\,;-9) & \hormat{1}{2}{4}{9} & (-1,\ 0,\ 3) & \hormat{1}{1}{3}{4} & (0,\ -1,\ 3) \\
	\hline
	(29,4,22\,;-20) & \hormat{1}{3}{3}{10} & (-1,\ 1,\ 4) & \hormat{1}{1}{2}{3} & (1,\ -1,\ 4) \\
	\hline
	(30,3,19,\;-1) & \hormat{4}{3}{13}{10} & (-1,\ 0,\ 1) & \hormat{1}{4}{3}{13} & (0,\ -1,\ 1) \\
	(30,7,29,\;-1) & \hormat{8}{7}{17}{15} & (-1,\ 0,\ 1) & \hormat{1}{8}{2}{17} & (0,\ -1,\ 1) \\
	\hline
\end{longtable}

\section{A sample of polar degeneracies}
\label{sec:furtherchecks}

In this section, we use our formula~\eqref{eq:degwmeta} to compute the polar coefficients of~$\psi_m^\mathrm{F}(\tau, z)$ for a few sample cases of~$ \displaystyle m $. 
The results presented below precisely agree with the polar coefficients extracted from the inverse of the Igusa cusp form~$\Phi_{10}^{-1}$ following the method outlined in the introduction below~\eqref{eq:Fourier-eta}.
For the sake of brevity, we only present a few examples owing to the large amount of data. \\
\mpage{
	\begin{longtable}{|C|C|}
		\hline
		(m,n,\ell) & \textbf{Degeneracy} \\  \hline
		(8,-1,-8) & 10 \\  \hline
		(8,-1,-7) & 216 \\  \hline
		(8,-1,-6) & 2592 \\  \hline
		(8,-1,-5) & 22400 \\  \hline
		(8,-1,-4) & 153900 \\  \hline
		(8,-1,-3) & 881280 \\  \hline
		(8,-1,-2) & 4295024 \\  \hline
		(8,-1,-1) & 17807488 \\  \hline
		(8,-1,0) & 61062180 \\  \hline
		(8,0,-8) & 4608 \\  \hline
		(8,0,-7) & 54432 \\  \hline
		(8,0,-6) & 460800 \\  \hline
		(8,0,-5) & 3078000 \\  \hline
		(8,0,-4) & 16922880 \\  \hline
		(8,0,-3) & 77538312 \\  \hline
		(8,0,-2) & 293278848 \\  \hline
		(8,0,-1) & 897317904 \\  \hline
		(12,-1,-12) & 14 \\  \hline
		(12,-1,-11) & 312 \\  \hline
		(12,-1,-10) & 3888 \\  \hline
\end{longtable}}
\mpage{
	\begin{longtable}{|C|C|}
		\hline
		(m,n,\ell) & \textbf{Degeneracy} \\  \hline
		(12,-1,-9) & 35200 \\  \hline
		(12,-1,-8) & 256500 \\  \hline
		(12,-1,-7) & 1586304 \\  \hline
		(12,-1,-6) & 8589760 \\  \hline
		(12,-1,-5) & 41513472 \\  \hline
		(12,-1,-4) & 181071642 \\  \hline
		(12,-1,-3) & 715942400 \\  \hline
		(12,-1,-2) & 2558054736 \\  \hline
		(12,-1,-1) & 8144997288 \\  \hline
		(12,-1,0) & 22401525768 \\  \hline
		(12,0,-12) & 6912 \\  \hline
		(12,0,-11) & 85536 \\  \hline
		(12,0,-10) & 768000 \\  \hline
		(12,0,-9) & 5540400 \\  \hline
		(12,0,-8) & 33841152 \\  \hline
		(12,0,-7) & 180384960 \\  \hline
		(12,0,-6) & 853994880 \\  \hline
		(12,0,-5) & 3621813000 \\  \hline
		(12,0,-4) & 13762586880 \\  \hline
		(12,0,-3) & 46454793840 \\  \hline
\end{longtable}}
\mpage{
	\begin{longtable}{|C|C|}
		\hline
		(m,n,\ell) & \textbf{Degeneracy} \\  \hline
		(12,0,-2) & 137011625088 \\  \hline
		(12,0,-1) & 346542104640 \\ \hline
		(15,-1,-15) & 17 \\  \hline
		(15,-1,-14) & 384 \\  \hline
		(15,-1,-13) & 4860 \\  \hline
		(15,-1,-12) & 44800 \\  \hline
		(15,-1,-11) & 333450 \\  \hline
		(15,-1,-10) & 2115072 \\  \hline
		(15,-1,-9) & 11810920 \\  \hline
		(15,-1,-8) & 59304960 \\  \hline
		(15,-1,-7) & 271607175 \\  \hline
		(15,-1,-6) & 1145472010 \\  \hline
		(15,-1,-5) & 4474748016 \\  \hline
		(15,-1,-4) & 16230894480 \\  \hline
		(15,-1,-3) & 54579105710 \\  \hline
		(15,-1,-2) & 168940316442 \\  \hline
		(15,-1,-1) & 473847914250 \\  \hline
		(15,-1,0) & 1169926333888 \\  \hline
		(15,0,-15) & 8640 \\  \hline
		(15,0,-14) & 108864 \\  \hline
\end{longtable}}
\newpage
\begin{minipage}{3in}
	\begin{longtable}{|C|C|}
		\hline
		(m,n,\ell) & \textbf{Degeneracy} \\  \hline
		(15,0,-13) & 998400 \\  \hline
		(15,0,-12) & 7387200 \\  \hline
		(15,0,-11) & 46531584 \\  \hline
		(15,0,-10) & 257692800 \\  \hline
		(15,0,-9) & 1280987136 \\  \hline
		(15,0,-8) & 5794286592 \\  \hline
		(15,0,-7) & 24054966432 \\  \hline
		(15,0,-6) & 92055592800 \\  \hline
		(15,0,-5) & 324742634880 \\  \hline
		(15,0,-4) & 1050674127360 \\  \hline
		(15,0,-3) & 3084121200240 \\  \hline
		(15,0,-2) & 8086496395392 \\  \hline
		(15,0,-1) & 18639111229056 \\ \hline
		(20,-1,-20) & 22 \\  \hline
		(20,-1,-19) & 504 \\  \hline
		(20,-1,-18) & 6480 \\  \hline
		(20,-1,-17) & 60800 \\  \hline
		(20,-1,-16) & 461700 \\  \hline
		(20,-1,-15) & 2996352 \\  \hline
		(20,-1,-14) & 17179520 \\  \hline
		(20,-1,-13) & 88957440 \\  \hline
		(20,-1,-12) & 422500050 \\  \hline
		(20,-1,-11) & 1861392000 \\  \hline
		(20,-1,-10) & 7670991600 \\  \hline
		(20,-1,-9) & 29756263680 \\  \hline
		(20,-1,-8) & 109143179628 \\  \hline
		(20,-1,-7) & 379708336000 \\  \hline
		(20,-1,-6) & 1255072397760 \\  \hline
		(20,-1,-5) & 3941870551552 \\  \hline
		(20,-1,-4) & 11741887027420 \\  \hline
		(20,-1,-3) & 33017035944960 \\  \hline
		(20,-1,-2) & 86858448321760 \\  \hline
		(20,-1,-1) & 210502750565336 \\  \hline
		(20,-1,0) & 458681404549752 \\  \hline
		(20,0,-20) & 11520 \\  \hline
	\end{longtable}
\end{minipage}
\begin{minipage}{3in}
	\begin{longtable}{|C|C|}
		\hline
		(m,n,\ell) & \textbf{Degeneracy} \\  \hline
		(20,0,-19) & 147744 \\  \hline
		(20,0,-18) & 1382400 \\  \hline
		(20,0,-17) & 10465200 \\  \hline
		(20,0,-16) & 67682304 \\  \hline
		(20,0,-15) & 386539200 \\  \hline
		(20,0,-14) & 1992646656 \\  \hline
		(20,0,-13) & 9415715400 \\  \hline
		(20,0,-12) & 41236992000 \\  \hline
		(20,0,-11) & 168761815200 \\  \hline
		(20,0,-10) & 649227576960 \\  \hline
		(20,0,-9) & 2357493364944 \\  \hline
		(20,0,-8) & 8100477591552 \\  \hline
		(20,0,-7) & 26357479662696 \\  \hline
		(20,0,-6) & 81109429456896 \\  \hline
		(20,0,-5) & 235139573743080 \\  \hline
		(20,0,-4) & 637627506612480 \\  \hline
		(20,0,-3) & 1600236038494008 \\  \hline
		(20,0,-2) & 3668952120405120 \\  \hline
		(20,0,-1) & 7591723325520696 \\ \hline
		(25,-1,-25) & 27 \\  \hline
		(25,-1,-24) & 624 \\  \hline
		(25,-1,-23) & 8100 \\  \hline
		(25,-1,-22) & 76800 \\  \hline
		(25,-1,-21) & 589950 \\  \hline
		(25,-1,-20) & 3877632 \\  \hline
		(25,-1,-19) & 22548120 \\  \hline
		(25,-1,-18) & 118609920 \\  \hline
		(25,-1,-17) & 573392925 \\  \hline
		(25,-1,-16) & 2577312000 \\  \hline
		(25,-1,-15) & 10867238100 \\  \hline
		(25,-1,-14) & 43281838080 \\  \hline
		(25,-1,-13) & 163714769010 \\  \hline
		(25,-1,-12) & 590657356800 \\  \hline
		(25,-1,-11) & 2039489782215 \\  \hline
	\end{longtable}
\end{minipage}
\newpage
\begin{minipage}{3in}
	\begin{longtable}{|C|C|}
		\hline
		(m,n,\ell) & \textbf{Degeneracy} \\  \hline
		(25,-1,-10) & 6757400879544 \\  \hline
		(25,-1,-9) & 21524693185035 \\  \hline
		(25,-1,-8) & 65996213108640 \\  \hline
		(25,-1,-7) & 194862258910125 \\  \hline
		(25,-1,-6) & 553836164704200 \\  \hline
		(25,-1,-5) & 1512859258863720 \\  \hline
		(25,-1,-4) & 3959104942633920 \\  \hline
		(25,-1,-3) & 9871029907055100 \\  \hline
		(25,-1,-2) & 23235254202421080 \\  \hline
		(25,-1,-1) & 50912869133641230 \\  \hline
		(25,-1,0) & 101777445949328016 \\  \hline
		(25,0,-25) & 14400 \\  \hline
		(25,0,-24) & 186624 \\  \hline
		(25,0,-23) & 1766400 \\  \hline
		(25,0,-22) & 13543200 \\  \hline
		(25,0,-21) & 88833024 \\  \hline
		(25,0,-20) & 515385600 \\  \hline
		(25,0,-19) & 2704306176 \\  \hline
		(25,0,-18) & 13037144400 \\  \hline
		(25,0,-17) & 58419072000 \\  \hline
		(25,0,-16) & 245471731200 \\  \hline
		(25,0,-15) & 973841356800 \\  \hline
		(25,0,-14) & 3667210825824 \\  \hline
		(25,0,-13) & 13163221094712 \\  \hline
		(25,0,-12) & 45182542960512 \\  \hline
		(25,0,-11) & 148662826021776 \\  \hline
		(25,0,-10) & 469629919909200 \\  \hline
		(25,0,-9) & 1425524443410600 \\  \hline
		(25,0,-8) & 4157177177438208 \\  \hline
		(25,0,-7) & 11632258851459120 \\  \hline
		(25,0,-6) & 31142251091455056 \\  \hline
	\end{longtable}
\end{minipage}
\begin{minipage}{3in}
	\begin{longtable}{|C|C|}
		\hline
		(m,n,\ell) & \textbf{Degeneracy} \\  \hline
		(25,0,-5) & 79392136978466280 \\  \hline
		(25,0,-4) & 191354425929177600 \\  \hline
		(25,0,-3) & 431942335951930920 \\  \hline
		(25,0,-2) & 903444520233320160 \\  \hline
		(25,0,-1) & 1734243812507148504 \\ \hline
		(30,-1,-30) & 32 \\  \hline
		(30,-1,-29) & 744 \\  \hline
		(30,-1,-28) & 9720 \\  \hline
		(30,-1,-27) & 92800 \\  \hline
		(30,-1,-26) & 718200 \\  \hline
		(30,-1,-25) & 4758912 \\  \hline
		(30,-1,-24) & 27916720 \\  \hline
		(30,-1,-23) & 148262400 \\  \hline
		(30,-1,-22) & 724285800 \\  \hline
		(30,-1,-21) & 3293232000 \\  \hline
		(30,-1,-20) & 14063484600 \\  \hline
		(30,-1,-19) & 56807412480 \\  \hline
		(30,-1,-18) & 218286358680 \\  \hline
		(30,-1,-17) & 801606412800 \\  \hline
		(30,-1,-16) & 2823908929200 \\  \hline
		(30,-1,-15) & 9572984572928 \\  \hline
		(30,-1,-14) & 31308644147760 \\  \hline
		(30,-1,-13) & 98994300336408 \\  \hline
		(30,-1,-12) & 303118553078000 \\  \hline
		(30,-1,-11) & 899973382487040 \\  \hline
		(30,-1,-10) & 2593304705881944 \\  \hline
		(30,-1,-9) & 7256059956824960 \\  \hline
		(30,-1,-8) & 19714696668645120 \\  \hline
		(30,-1,-7) & 51987444460877112 \\  \hline
		(30,-1,-6) & 132888904085878840 \\  \hline
		(30,-1,-5) & 328541317658460288 \\  \hline
	\end{longtable}
\end{minipage}
\newpage
\begin{minipage}{3.1in}
	\begin{longtable}{|C|C|}
		\hline
		(m,n,\ell) & \textbf{Degeneracy} \\  \hline
		(30,-1,-4) & 782726769159905520 \\  \hline
		(30,-1,-3) & 1786711012854816640 \\  \hline
		(30,-1,-2) & 3873826859341935240 \\  \hline
		(30,-1,-1) & 7877457636088694664 \\  \hline
		(30,-1,0) & 14774702983837211616 \\  \hline
		(30,0,-30) & 17280 \\  \hline
		(30,0,-29) & 225504 \\  \hline
		(30,0,-28) & 2150400 \\  \hline
		(30,0,-27) & 16621200 \\  \hline
		(30,0,-26) & 109983744 \\  \hline
		(30,0,-25) & 644232000 \\  \hline
		(30,0,-24) & 3415965696 \\  \hline
		(30,0,-23) & 16658573400 \\  \hline
		(30,0,-22) & 75601152000 \\  \hline
		(30,0,-21) & 322181647200 \\  \hline
		(30,0,-20) & 1298455142400 \\  \hline
		(30,0,-19) & 4976928977904 \\  \hline
		(30,0,-18) & 18225998438400 \\  \hline
	\end{longtable}
\end{minipage}
\begin{minipage}{3.1in}
	\begin{longtable}{|C|C|}
		\hline
		(m,n,\ell) & \textbf{Degeneracy} \\  \hline
		(30,0,-17) & 64008602395200 \\  \hline
		(30,0,-16) & 216236828000256 \\  \hline
		(30,0,-15) & 704444493333240 \\  \hline
		(30,0,-14) & 2217472328601600 \\  \hline
		(30,0,-13) & 6755213522446272 \\  \hline
		(30,0,-12) & 19937873507586048 \\  \hline
		(30,0,-11) & 57052739503386000 \\  \hline
		(30,0,-10) & 158314644037023360 \\  \hline
		(30,0,-9) & 425845859881717296 \\  \hline
		(30,0,-8) & 1109163700842332160 \\  \hline
		(30,0,-7) & 2791657574839862400 \\  \hline
		(30,0,-6) & 6767317387211658624 \\  \hline
		(30,0,-5) & 15722850796882492680 \\  \hline
		(30,0,-4) & 34777416386698174464 \\  \hline
		(30,0,-3) & 72630503135639181864 \\  \hline
		(30,0,-2) & 141950626331053105152 \\  \hline
		(30,0,-1) & 257636988474238025304 \\ \hline
	\end{longtable}
\end{minipage}

\vspace{0.5in}

\bibliography{Dyon_DH}

\bibliographystyle{JHEP}
\markboth{}{}

\end{document}